\documentclass[12pt,a4paper,titlepage]{report}
\usepackage[centertags]{amsmath}
\usepackage{amsfonts}
\usepackage{amssymb}
\usepackage{amsthm}
\usepackage{bbm}
\usepackage{newlfont}
\usepackage{a4}
\usepackage{enumerate}
\def\nonumchapter#1{%
    \chapter*{#1}
    \addcontentsline{toc}{chapter}{#1}}
\def\prefacesection#1{%
    \chapter*{#1}
    }
\newlength{\defbaselineskip}
\newcommand{\setlinespacing}[1]%
           {\setlength{\baselineskip}{#1 \defbaselineskip}}

\newcommand{\map}{\rightarrow}
\newcommand{\der}{\operatorname{der}}
\newcommand{\coc}{\operatorname{coc}}
\newcommand{\ad}{\operatorname{ad}}
\newcommand{\Span}{\operatorname{span}}
\newcommand{\inv}{\operatorname{inv}}

\newcommand{\slp}{\mathop{\mathrm {sl} }\nolimits}

\newcommand{\End}{\operatorname{End}}
\newcommand{\gl}{\operatorname{gl}}
\newcommand{\jor}{\operatorname{jor}}

\newcommand{\g}{\operatorname{g}}
\renewcommand{\j}{\operatorname{j}}
\newcommand{\fa}{\operatorname{\psi}}
\newcommand{\fb}{\operatorname{\phi}}
\newcommand{\fc}{\operatorname{\phi^0}}

\newcommand{\aut}{\mathop{\mathrm {Aut} }\nolimits}
\newcommand{\tr}{\mathop{\mathrm {tr} }\nolimits}
\newcommand{\el}{{\cal L}}

\newcommand{\q}{\quad}

\newcommand{\B}{{\cal B}}

\newcommand{\A}{{\cal A}}

\renewcommand{\epsilon}{\varepsilon}

\newcommand{\ep}{\varepsilon}
\newcommand{\la}{\lambda}
\newcommand{\al}{\alpha}
\renewcommand{\rho}{\varrho}
\renewcommand{\phi}{\varphi}

\newcommand{\R}{{\mathbb{R}}}

\newcommand{\N}{{\mathbb N}}

\newcommand{\Com}{{\mathbb C}}

\newcommand{\Z}{\mathbb{Z}}
\newcommand{\C}{\circ}

\newcommand{\set}[2]{\left\{#1 \, |\, #2 \right\}}

\newcommand{\wt}{\widetilde}
\newcommand{\gc}[2]{\coc_{( #1 )}\, #2}
\newcommand{\gd}[2]{\der_{( #1 )}\, #2}
\newcommand{\ga}[2]{\aut_{( #1 )}\, #2}

\newcommand{\comb}[2]{\begin{pmatrix}
     #1\\
     #2
  \end{pmatrix}}

\theoremstyle{definition}
\newtheorem{defn}{Definice}[section]
\newtheorem{thm}[defn]{Theorem}
\newtheorem{lemma}[defn]{Lemma}
\newtheorem{tvr}[defn]{Proposition}
\newtheorem{cor}[defn]{Corollary}

\theoremstyle{remark}

\newtheorem{example}{Example}

\begin{document}

\begin{titlepage}
\begin{center}
\Large{CZECH TECHNICAL UNIVERSITY IN PRAGUE} \\\normalsize Faculty of Nuclear Sciences
and Physical Engineering
\end{center}
\addvspace{150pt}
\begin{center}
\LARGE{\bf{DOCTORAL THESIS}}
\end{center}
\addvspace{50pt}
\begin{center}
\LARGE INVARIANTS OF LIE ALGEBRAS
\end{center}
\addvspace{50pt}
\begin{center}
\Large Ji\v{r}\'{\i} Hrivn\'ak \bigskip\bigskip\\ \normalsize Supervisor: Prof. Ing. {J.}
Tolar, DrSc.
\end{center}
\addvspace{80pt}
\begin{center}
\large September 12, 2007
\end{center}
\end{titlepage}
\begin{titlepage}
\addvspace{350pt}
\bigskip
This thesis is the result of my own work, except where explicit
reference is made to the work of others and has not been submitted
for another qualification to this or any other university.
\begin{flushright}
\addvspace{30pt}
\bigskip
Ji\v{r}\'{\i} Hrivn\'ak
\end{flushright}
\end{titlepage}
\setlinespacing{1.25}
\prefacesection{Acknowledgements}
\thispagestyle{empty} \pagenumbering{roman} I would like to thank
Petr Novotn\'y, prof. Edita Pelantov\'a, prof. Ji\v{r}\'{\i}
Patera and all who contributed to this work during the past
several years by advice, consultation or remark.

Especially, I would like to thank prof. Ji\v{r}\'{\i} Tolar for
his kind supervision during the past seven years.

\newpage
\addcontentsline{toc}{chapter}{Contents} \pagenumbering{arabic}
\setlinespacing{1.15} \tableofcontents
\newpage
\setlinespacing{1.25}

\nonumchapter{Introduction}

Finite--dimensional complex Lie algebras form extremely useful and
frequent part of Lie theory in physics and elsewhere. With an
increasing amount of theory and applications concerning Lie
algebras of various dimensions, it is becoming necessary to
ascertain applicable tools for handling them. Their miscellaneous
characteristics constitute such tools and have also found
applications: {\it Casimir operators} \cite{inv}, {\it derived},
{\it lower central} and {\it upper central sequences}, Lie {\it
algebra of derivations}, {\it radical}, {\it nilradical}, {\it
ideals}, {\it subalgebras} \cite{Jacobson,ide} and recently {\it
megaideals} \cite{Pop1}. These characteristics are central when
considering possible affinities among Lie algebras.

Physically motivated relations between two Lie algebras, namely
{\it contractions} and {\it deformations}, were extensively
studied for instance in~\cite{G,PdeM}. When investigating these
kinds of relations in dimensions higher than five, one can
encounter insurmountable difficulties.

Firstly, aside the {\it semisimple} ones, Lie algebras are
completely classified only up to dimension 5 and the {\it
nilpotent} ones up to dimension 6. In higher dimensions, lists of
only special types such as {\it rigid} Lie algebras \cite{Goze} or
Lie algebras with fixed structure of nilradical are determined
\cite{SW}. For a detailed survey of classification results in
lower dimensions see the very recent paper \cite{Pop1} and
references therein.

Secondly, known invariant characteristics of Lie algebras are, in
some cases, insufficient for their classification. After
performing the standard identification procedure in \cite{ide},
one ends with the set of classical invariants \cite{ccc}. Even
though there has been progress in extending this set of classical
invariants, the results still turn out to be insufficient,
especially for the description of nilpotent Lie algebras.
Alongside the classification of orbit closures of
four--dimensional complex Lie algebras~\cite{Bur1}, so called
$C_{pq}$ invariants were introduced; in~\cite{AY} the invariants
$\chi_i$ were used for the identification of four--dimensional
complex Lie algebras. In~\cite{Bur2,Bur3}, the dimensions of
cohomology spaces with respect to the adjoint and trivial
representations were used as invariants. However, none of the
above invariants is able to resolve a nilpotent parametric
continuum of Lie algebras. These nilpotent continua frequently
appear as results of the graded contraction procedure.

The graded contraction procedure, originally introduced in
\cite{PdeM}, was later extended to graded contractions of the
representations of Lie algebras~\cite{MP} and to the Jordan
algebras~\cite{KP}. In~\cite{HN2,W2} all graded contractions of
$\slp(3,\Com)$ corresponding to the Pauli grading and toroidal
grading, respectively, were found. Due to the difficulties
mentioned above, the thorough study of resulting parametric
continua was omitted in these papers. Later it became evident that
the ultimate result -- graded contractions corresponding to all
four gradings of $\slp(3,\Com)$ -- involves classifying hundreds
of parametric Lie algebras and is, without more appropriate
invariants, out of reach. Of course, when one has two identical
sets of invariants for two given algebras, one can always try to
find explicitly an isomorphism in a given basis -- that is, find a
regular solution of a very large system of quadratic equations. If
the search for an explicit regular solution fails, one has to
exclude the existence of any regular solution -- this, in some
cases, can be done by hand. There has been progress in developing
algorithms for these direct calculations: employing modern
computational methods and the theory of Gr$\mathrm{\ddot{o}}$bner
bases, the classification of the three and four--dimensional
solvable Lie algebras has been obtained~\cite{Graaf}. Application
of these types of direct computations seems, however, not
surmountable for hundreds of parametric Lie algebras in higher
dimensions.

Thus, the {\it goals} of this thesis are the following:
\begin{itemize}
\item add new objects to the existing set of invariants of Lie algebras
\item formulate the properties of these
invariants and -- in view of possible alternative classifications
-- investigate their behaviour on known lower--dimensional Lie
algebras
\item demonstrate that these invariants
are -- in view of their application on graded contractions of
$\slp(3,\Com)$ -- also effective in higher dimensions
\item formulate a necessary contraction criterion involving these
invariants and apply it to lower--dimensional cases
\item investigate possible application of these invariant characteristics
to Jordan algebras
\end{itemize}

One can see from the references that our investigation belongs to
a very lively domain of Lie theory. Only Chapter 1. is completely
of {\it review character}. Chapters 2, 3, 4 and Appendices are
devoted to thorough description of {\it new results}.
\newpage {\it Original results} of this thesis are
\begin{itemize}
\item contained in \cite{NH}---\cite{HNx}: Chapter 2, Section 4.2
\item unpublished: Chapter 3, Section 4.1, Appendices
\end{itemize}

The {\it most significant original results} are the following:
\begin{itemize}
\item Chapter 2: Corollaries~\ref{corgender}, \ref{invarfunc},
\ref{inewww}, Theorems~\ref{Lielist}, \ref{prunklass} and
\ref{class3dim}
\item Chapter 3: Corollaries~\ref{invdim}, \ref{invarfunc2}, Theorems~\ref{klass2} and \ref{class4dim}
\item Chapter 4: Theorems~\ref{dimderconthm}, \ref{conthmmain},
\ref{contract3dimJ} and~\ref{zvast}
\item Appendix A: classification of
$(\alpha,\beta,\gamma)$--derivations of two and three--dimensional
Lie algebras
\item Appendix B: the invariant functions allowing
the classification of four--dimensional Lie algebras
\end{itemize}

In Chapter~\ref{CHinv}, the definitions and theorems used in this
work are summarized. We define invariant characteristics of
algebras and their mutual independence. We review the concept of a
Lie algebra and their cohomology. We also state basic facts about
linear groups and Jordan algebras.

In Chapter~\ref{CHgen}, the concept of the derivation of a Lie
algebra is generalized; $(\alpha,\beta,\gamma)$--derivations are
introduced and their pertinent properties shown. All possible
intersections of spaces containing these derivations are
investigated. Examples of spaces of
$(\alpha,\beta,\gamma)$--derivations for low--dimensional Lie
algebras are presented. In special cases, the spaces of
$(\alpha,\beta,\gamma)$--derivations form Lie or Jordan operator
algebras. These algebras are investigated and the corresponding
Lie groups constructed. Invariant functions $\fa, \fa^0$ are
defined and their values estimated. The invariant function $\fa$
is used as the classification tool of three--dimensional Lie
algebras.

In Chapter~\ref{CHtwi}, the concept of cocycles of Chevalley
cohomology is generalized; $\kappa$--twisted cocycles are
introduced and shown that for two--dimensional twisted cocycles,
analogous properties to the properties of
$(\alpha,\beta,\gamma)$--derivations hold. Examples of selected
spaces of twisted cocycles are presented. Two invariant functions
$\fb$ and $\fc$ are defined and their behaviour on
low--dimensional Lie algebras demonstrated. The invariant
functions $\fa$ and $\fb$ are used to classify all
four--dimensional Lie algebras. New algorithm for the
identification of a four--dimensional Lie algebra is also
formulated.

In Chapter~\ref{CHcon}, possible application of the invariant
functions $\fa,\fb,\fc$ to contractions is considered. Necessary
criterion for existence of a continuous contraction is formulated.
The invariant function $\fa$ is used to classify continuous
contractions among three--dimensional Lie algebras. This function
is also employed to the classification of two--dimensional Jordan
algebras and their continuous contractions. The invariant
functions are used to distinguish among results of graded
contraction procedure. Application of the invariant functions on
nilpotent parametric continua of Lie algebras resulting from
contractions of the Pauli graded $\slp(3,\Com)$ is demonstrated.


In Conclusion, we shortly review other generalizations of
derivations, make notes on a computation of
$(\alpha,\beta,\gamma)$--derivations and twisted cocycles. We
summarize achieved results and make comments concerning further
applications.

In Appendix~\ref{APA}, explicit matrices of
$(\alpha,\beta,\gamma)$--derivations for two and
three--dimensional Lie algebras and for two--dimensional Jordan
algebras are completely classified.

In Appendix~\ref{IFLJ}, the tables of the invariant functions
$\fa,\,\fb,\,\fc$ for two, three and four--dimensional Lie
algebras and for two--dimensional Jordan algebras are listed.

\newpage

\chapter{Invariants of Lie Algebras}\label{CHinv}
\section{Invariant Characteristics of Algebras}
In order to unify the notation and definitions, we amass in this
section basic notions concerning linear algebras. Except for the
definition of independent invariant, the content of this section
may be found e. g. in \cite{Goto,Jacobson}.

We call a vector space $V$ of finite dimension $n$ over the field
of complex numbers $\Com$ together with a bilinear map $ V\times
V\ni (a,b)\mapsto a\cdot b \in V$ a (complex) {\bf algebra}
$\A=(V,\,\cdot\,)$. We call the map $\cdot$ a {\bf multiplication}
of $\A$. An algebra $\A$ is {\bf associative} if the rule $(a\cdot
b)\cdot c =a\cdot( b \cdot c)$ holds for all $a,b,c \in \A$. Let
$\mathcal{X} = (x_1, \dots , x_n)$ be some basis of $\A$. Then the
numbers $c^k_{ij}\in \Com$ defined by
\begin{equation}\label{STRUCT}
  x_i\cdot x_j=\sum_{k=1}^n c^k_{ij} x_k
\end{equation}
are called {\bf structural constants} with respect to the basis
$\mathcal{X}$. A vector subspace $B$ of $\A$ is called a {\bf
subalgebra} of $\A$ if $a\cdot b \in B$ for all $a,b\in B$. For
arbitrary subsets $B,C$ of $\A$ the symbol $B\cdot C$ denotes the
linear span of all products of elements $b\cdot c$, where $b\in B$
and $c\in C$. A subalgebra $B$ is called an {\bf ideal} in $\A$ if
$\A\cdot B\subset B$ and $B\cdot \A\subset B$.

Given an ideal $B$ in an algebra $\A$ then the factor space $$\A/B
= \set{ [a] = a+B }{ a \in \A } $$ with a well--defined
multiplication $ [a] \cdot [b] = [a \cdot b], \ \forall a,b \in
\mathcal{A}$ is called the {\bf factor algebra} and denoted
$\A/B$.

Suppose we have two algebras $\A,\,\wt\A$ over $\Com$ with
multiplications $\cdot$ and $\ast$, respectively. Then a linear
map $f:\A\map \wt\A$ is called a {\bf homomorphism} if the
relation
\begin{equation}\label{homomorph}
f(a\cdot b)=f (a) \ast f (b)
\end{equation}
holds for all $a,b\in \A.$ The kernel $\ker f$ is an ideal in
$\A$. If $\ker f = 0$ then $f$ is called an {\bf isomorphism} and
algebras the $\A,\wt\A$ are called {\bf isomorphic}, $\A\cong
\wt\A$. Let us denote the group of all regular linear maps on an
arbitrary vector space $V$ by $GL(V)$; for an algebra
$\A=(V,\,\cdot\,)$ we define the symbol $GL(\A)$ by
$GL(\A)=GL(V)$. An isomorphism $f:\A\map \A$ is called an {\bf
automorphism} of an algebra $\A=(V,\,\cdot\,)$; the set of all
automorphisms forms a multiplicative group $\aut \A\subset GL
(\A)$, i.~e.
\begin{equation}\label{aaaa}
\aut \A=\set{f\in GL(\A)}{f(a\cdot b)= f(a) \cdot f(b)\q \forall
a,b \in \A}.
\end{equation}
We denote by $\End V$ a vector space of all linear operators on
the vector space $V$; if we have an algebra $\A=(V,\,\cdot\,)$, we
define the symbol $\End \A$ by $\End\A=\End V$. Considering the
composition of linear operators as a multiplication, $\End V$
becomes an associative algebra. If $f:\A \map\wt\A $ is an
isomorphism of complex algebras $\A$ and $\wt\A$, then the mapping
$\rho :\End\A\map \End\wt\A$, defined for all $D\in \End\A$ by
\begin{equation}\label{rho}
\rho (D) = f D f^{-1}
\end{equation}
is an isomorphism of the associative algebras $\End \A$ and $
\End\wt\A$, i.~e. $\End\A\cong \End\wt\A$.

A map $D\in\End \A $ which satisfies $D(a\cdot b) = (D a)\cdot b +
a\cdot (D b)$ for all $a,b\in\A$ is called a {\bf derivation} of
$\A$; we denote the set of all derivations by $\der \A$, i.~e.
\begin{equation}\label{dddd}
\der \A=\set{D\in \End \A}{D(a\cdot b)=  (Da)\cdot b + a\cdot (D
b)\q \forall a,b \in \A}.
\end{equation}
One can easily prove that $\der \A$ forms a linear subspace of the
vector space $\End \A$.

The relation of isomorphism $\cong$ between two algebras is an
equivalence relation, i.~e. it is symmetric, reflexive and
transitive. Thus, the set of all algebras is decomposed into {\bf
isomorphism classes} -- cosets of this equivalence; we denote such
coset containing an algebra $\A$ by $[\A]$. Then, indeed, $\B\in
[\A]$ holds if and only if $\B\cong \A$. Suppose we have a
non--empty set $M$ and some subset $\Theta$ of the set of all
isomorphism classes of algebras. We call a map
\begin{equation}\label{invvv}
 \Phi : \Theta \map M
\end{equation}
an {\bf invariant characteristic} of $\Theta$ or shortly an {\bf
invariant}. In other words, the mapping~$\Phi$ assigns to each
coset $[\A]\in \Theta$, which contains all mutually isomorphic
algebras, some element in~$M$.

For example, we may consider $\Theta$ equal to the set of all
isomorphism classes of algebras, set $M=\N_0$ and for any
$[\A]\in\Theta$ we may define $\Phi [\A]= \dim \A =n$. It is clear
that this mapping $\Phi$ is well--defined and indeed has the same
value for all isomorphic algebras. We sometimes say that the
dimension of $\A$ is an invariant characteristic of algebras. We
call such an invariant, for which $M=\N_0$, a {\bf numerical
invariant}. As another example of a numerical invariant for an
arbitrary algebra may serve $\Phi_{\der} [\A] =\dim \der \A$;
indeed, for $\A,\wt\A\in [\A]$ one has
\begin{equation}\label{rhoder}
  \rho (\der\A)=\der\wt\A,
\end{equation}
i.~e. $\dim\der\A=\dim \der\wt\A $. Thus, the mapping
$\Phi_{\der}$ is well--defined.

As a number of known invariants for any given type of algebra
grows, there arises natural claim to reflect their independence.
We refine the notion of independence in the following way. Suppose
we have a subset $\Theta$ of the set of all isomorphism classes of
algebras and a set of invariants $\Omega$ of $\Theta$. Then we
call an invariant $\Phi$ of $\Theta$ {\bf independent} on the set
of invariants $\Omega$ if there exist two cosets
$[\A],\,[\B]\in\Theta$ such that
\begin{align}
1.\q \q& \Psi [\A] = \Psi [\B],\q \forall \Psi  \in \Omega\q \q
\nonumber \\2. \q\q & \Phi [\A] \neq \Phi [\B].\label{independent}
\end{align}
Note that this definition does not, in general, exclude some
possible relation among invariants in $\Omega$ and $\Phi$.
However, it is pertinent that the invariant $\Phi$ does not depend
{\it only} on those in $\Omega$. Since the main aim of the notion
of invariant characteristic is to distinguish among different
cosets of algebras, the independent invariant $\Phi$ thus
distinguishes cosets $[\A]$ and $[\B]$.


\section{Lie Algebras}
Basic facts included in this section, concerning Lie algebras,
their representations and ideals may be found for instance
in~\cite{Goto,Jacobson}. Suppose we have a (complex) algebra $\el$
with the multiplication $[\,,\,]$ which for all $x,y,z\in \el$
satisfies
\begin{enumerate}[(1)]
  \item  $[x,x]=0$ (anti--commutativity)
  \item  $[x,[y,z]]+[z,[x,y]]+[y,[z,x]]=0$ (Jacobi's identity).
\end{enumerate}

Such an algebra $\el$ is then called a (complex) {\bf Lie
algebra}. In terms of structure constants (\ref{STRUCT})
anti--commutativity and Jacobi's identity may be written as
\begin{enumerate}[(1')]
  \item  $c^m_{ij}+c^m_{ji}=0$
  \item  $\sum_{l=1}^{n}(c^l_{jk}c^m_{il}+c^l_{ij}c^m_{kl}+c^l_{ki}c^m_{jl})=0,\ \forall i,j,k,m \in \{1,\dots,n\}$.
\end{enumerate}

Suppose $V$ is an arbitrary vector space. Then we may introduce a
new multiplication on the associative algebra $\End V$. For two
linear operators $X,Y\in\End V$ we put
\begin{equation}\label{glgl}
 [X,Y]= X Y-Y X.
\end{equation}
Then we indeed obtain a Lie algebra which we denote by $\gl V$.
For an algebra $\A=(V,\,\cdot\,)$ we define $\gl\A=\gl V$. Then we
have:
\begin{tvr}\label{dersub}
  The set of all derivations $\der \A$ of an algebra $\A$ is a Lie
  subalgebra of~$\gl \A$.
\end{tvr}
If $f:\A \map\wt\A $ is an isomorphism of complex algebras $\A$
and $\wt\A$, then the mapping $\rho :\gl\A\map \gl\wt\A$, defined
by~(\ref{rho}) is an isomorphism of the Lie algebras $\gl \A$ and
$ \gl\wt\A$, i.~e. $\gl\A\cong \gl\wt\A$. From this fact and
from~Proposition~\ref{dersub} and~(\ref{rhoder}) one also obtains
that
\begin{equation}\label{izoder}
  \der\A\cong \der\wt\A.
\end{equation}
holds.

If we choose some basis in $\A$ then to each operator from $\gl
\A$ is assigned a matrix; the space of these matrices is denoted
by $\gl (n,\Com)$ and forms also a Lie algebra with respect to the
matrix multiplication $[\,,\,]$. An important matrix algebra of
traceless matrices $\slp (n,\Com)$ is defined by the relation
\begin{equation}\label{slnc}
\slp (n,\Com)=\{S\in \gl (n,\Com)| \tr S=0\}.
\end{equation}
Having two Lie algebras
$\el_1=(V_1,\,[\,,\,]_1),\,\el_2=(V_2,\,[\,,\,]_2)$ one may define
on the direct sum of the vector spaces $$V_1\oplus
V_2=\set{(x,y)}{x\in V_1,y\in V_2} $$ a Lie multiplication
$$[(x_1,y_1),(x_2,y_2)]_\oplus=([x_1,x_2]_1,[y_1,y_2]_2)$$ and
obtain a Lie algebra $\el_1\oplus\el_2=(V_1\oplus
V_2,\,[\,,\,]_\oplus )$ called the {\bf direct sum} of $\el_1$ and
$\el_2$.

Let $V$ be a vector space over $\Com$. We call a homomorphism $$f
: \el \map \gl V $$ a {\bf representation} of a Lie algebra $\el$
over $\Com$ on the vector spaces $V$ and we denote it by $(V,f)$.
A map $\ad_\el: \el\map \gl \el$ defined for all $x,y\in \el$ by
the relation $$(\ad_\el x)y =[x,y]  $$ is a representation
$(\ad_\el,\el)$ and is called the {\bf adjoint representation}.
\begin{tvr}
For any Lie algebra $\el$ the set $\ad \el = \set{\ad_\el x}{x\in
\el }$ is an ideal in $\der \el $.
\end{tvr}

Having $(V,f)$ a representation of $\el$ we define a
representation on the dual space $V^*$. For $s\in \gl V$, $\la \in
V^*$ and $v\in V $ consider the mapping $s^{t}\in \gl V^*$ defined
by $$(s^t \la) v = \la (s v) .$$ If we put $f^*=-f^t$ then $(V^*,
f^*)$ is a representation called the {\bf dual representation} of
$(V,f)$.

We point out some important ideals of a Lie algebra $\el$. We
denote by $C(\el)$ a {\bf center} of $\el$ defined by
\begin{equation} C(\el) =
\set{x \in \el}{[x,y]=0,\, \forall y \in \el }
\end{equation}
and a {\bf derived algebra} $\el^2$ of $\el$
\begin{equation}
\el^2 = [\el,\el].
\end{equation}
A {\bf centralizer of the adjoint representation}
$C_{\ad}(\el)\subset\gl \el $ is defined as follows:
\begin{equation}\label{centralizer}
C_{\ad}(\el)=\set{A\in\gl \el}{[A,\ad_\el(x)]=0,\, \forall x\in
\el}.
\end{equation}

Next, we introduce sequences of ideals. Sequence of ideals $D^0(\el)
\supset D^1(\el) \supset \ldots$ defined by
\begin{equation} D^0(\el) = \el, \qquad
    D^{k+1}(\el) = [D^k(\el),D^k(\el)], \qquad k\in\N_0
\end{equation}
is called a {\bf derived sequence} of $\el$. Sequence of ideals
$\el^1 \supset \el^2 \supset \ldots$ defined by
\begin{equation} \el^1 = \el,
    \qquad \el^{k+1} = [\el^k,\el], \qquad k\in\N
\end{equation} is called a {\bf descending central sequence}.
Sequence of ideals $C^0(\el) \subset C^1(\el) \ldots $ defined by
\begin{equation} C^0(\el) = 0, \qquad C^{k+1}(\el)/C^k(\el)
=C(\el/C^k(\el)), \qquad k\in\N_0
\end{equation} is called an {\bf ascending central sequence}. We
define three sequences of numerical invariants $d_k,\,c_k,\,l_k,\,
k=0,1,\dots$ by the relations
\begin{align}\label{seq1}
 d_k (\el)  &=\dim D^{k}(\el) \\  l_k (\el)&=\dim \el^{k+1} \label{seq2}  \\  c_k (\el)&=\dim C^{k+1}(\el).\label{seq3}
\end{align}
If, for some $k$ and a Lie algebra $\el$, $D^k(\el)=0$ holds then
$\el$ is called {\bf solvable}; if $\el^k=0$ then it is called
{\bf nilpotent}. If $\el$ contains no solvable ideal then it is
called {\bf semisimple}; if it contains only trivial ideals $\el$
and $\{0\}$ and $\dim (\el)>1$ then $\el$ is called {\bf simple}.

\begin{thm}[Engel]\label{Engel}
Let $\el$ be a Lie algebra. Then $\el$ is nilpotent if and only if
$\ad_\el x$ is nilpotent for every $x\in \el$.
\end{thm}

It is well known that the sum of two nilpotent ideals in a Lie
algebra is again a nilpotent ideal. The sum of all nilpotent
ideals is a maximal nilpotent ideal called a {\bf nilradical}. It
is also true that the sum of two solvable ideals is a solvable
ideal. The sum of all solvable ideals is a maximal solvable ideal
called a {\bf radical}. It is clear that if radical is zero then
such a Lie algebra is semisimple.

The element $F$ of the universal enveloping algebra
\cite{Goto,Jacobson} of $\el$ which satisfies
\begin{equation}\label{deca}
x F-F x = 0, \qquad \forall x\in \el,
\end{equation}
is called a {\bf Casimir operator} \cite{inv,ccc}. These operators
can be calculated in the following way. We take the representation
of the elements of the basis $(e_1,\dots,e_n)$ of $\el$ by vector
fields
\begin{equation}
e_i \rightarrow\hat{x}_i =
\sum_{j,k=1}^nc_{ij}^kx_k\frac{\partial}{\partial x_j}.
\end{equation}
These vector fields act on the space of continuously
differentiable functions $F(x_1,\ldots,x_n)$. We call a function
$F$ {\bf formal invariant} of $\el$ if it is a solution of
\begin{equation*}
\hat{x}_iF=0,\qquad i\in\{1,\dots,n\}.
\end{equation*}
The number of algebraically independent formal invariants is
\begin{equation}\label{formal}
\tau(\el) = \dim\el-r(\el),
\end{equation}
where $r(\el)$ is the rank of the antisymmetric matrix $M_{\el}$
and $(M_{\el})_{ij} = \sum_{k}c_{ij}^ke_k$:
\begin{equation*}
r(\el) =
\begin{array}[t]{c}
  \sup \\
  ^{(e_1,\ldots,e_n)} \\
\end{array} \operatorname{rank}(M_{\el}).
\end{equation*}
The map $\tau$ defined via~(\ref{formal}) forms a numerical
invariant. In \cite{inv}, a procedure for obtaining Casimir
invariant from a polynomial formal invariant $F(x_1,\ldots,x_n)$
is formulated.

Let $p,q\in \N$ be fixed numbers. Suppose there exist $u,v\in\el$
such that $\tr (\ad_\el u)^p\neq 0$,  $\tr(\ad_\el v)^q\neq 0$ and
$\tr [(\ad_\el u)^p (\ad_\el v)^q] \neq 0$. If there exists
$C_{pq}\in \Com$ such that for all $x,\,y\in\el$ the equality
\begin{equation}\label{Cpq}
\tr (\ad_\el x)^p \tr(\ad_\el y)^q = C_{pq} \tr [(\ad_\el x)^p
(\ad_\el y)^q]
\end{equation}
holds, then $C_{pq}$ is called the {\bf $C_{pq}$--invariant} of
$\el$ \cite{Bur1}.

In \cite{AY} were introduced the functions: \begin{eqnarray*}
&p_{111}(x)= -\tr \ad_\el x
\\ &p_{222}(x)= \frac{1}{2}\left( (\tr \ad_\el x )^2- \tr (\ad_\el x)^2
\right)\\ &p_{333}(x)= -\frac{1}{6}\left((\tr \ad_\el x)^3-3\tr
\ad_\el x \, \tr (\ad_\el x)^2 + 2\tr (\ad_\el x)^3 \right)
\end{eqnarray*}
If there exists $u\in \el$ such that $p_{222}(u)\neq 0$,
$p_{111}^2(u)\neq 0$ and exists $\chi_1 \in \Com$ such that for
all $x\in\el$ it holds:
\begin{equation}\label{Chi}  p_{222}(x) =\chi_{1}p_{111}^2(x)
\end{equation}
then we have the {\bf invariant} $\chi_{1}(\el)$. Similarly are
defined invariants $\chi_{2}(\el)$ and $\chi_{3}(\el)$, i.~e. by
relations:
\begin{equation}   p_{333}(x)
=\chi_{2}p_{111}^3(x) ,\ p^2_{333}(x) =\chi_{3}p^3_{222}(x).
\end{equation}

\section{Chevalley Cohomology of Lie Algebras}

The content of this section may be found for instance
in~\cite{Goto}. Let $V$ be a vector space over $\Com$ and let
$(V,f)$ be a representation of $\el$. We call a $q-$linear map
$c:\underbrace{\el\times\el\times\dots\times\el}_{q-times}\map V$
a {\bf $V$--cochain} of dimension $q$, if for all pairs of indices
$i,j,\,(1\leq i<j\leq q)$ the relation
$$c(x_1,\dots,\begin{array}[t]{c}
  x_i \\
  i
\end{array},\dots,\begin{array}[t]{c}
  x_j \\
  j
\end{array},\dots,x_q)+c(x_1,\dots,\begin{array}[t]{c}
  x_j \\
  i
\end{array},\dots,\begin{array}[t]{c}
  x_i \\
  j
\end{array},\dots,x_q)=0 $$
holds. We denote by $C^q(\el,V)$ the vector space of all
$V$-cochains of dimension $q$ for $q\in \N$ and  $C^0(\el,V)=V$.
We define a map $d:  C^q(\el,V)\map  C^{q+1}(\el,V)$ for
$q=0,1,2,\dots$ by
\begin{align}\label{koho1}
d c(x)&=f(x)c\q c\in  C^0(\el,V)\\ \nonumber d
c(x_1,\dots,x_{q+1}) &= \sum_{i=1}^{q+1}
(-1)^{i+1}f(x_i)c(x_1,\dots,\hat{x_i},\dots,x_{q+1})+\\ \nonumber
&+ \sum_{\begin{smallmatrix}
  i,j=1 \\
  i<j
\end{smallmatrix}}^{q+1} (-1)^{i+j} c([x_i,x_j],x_1,\dots,\hat{x_i},\dots,\hat{x_j},\dots,x_{q+1})
\end{align}
where the symbol $\hat{x_i}$ means that the term $x_i$ is omitted.
We summarize the crucial results concerning the map $d$.
\begin{thm}
For the map $d:  C^q(\el,V)\map  C^{q+1}(\el,V)$, defined by
(\ref{koho1}), it holds:
\begin{equation}\label{ddd}
dd=0.
\end{equation}
\end{thm}
Such $z\in C^q(\el,V)$ for which $d z=0$ holds is called {\bf
cocycle} of dimension $q$ corresponding to $f$; the set of all
cocycles of dimension $q$ corresponding to $f$ is denoted by
$Z^{q} (\el,f)$. An element $w\in C^q(\el,V)$ for which such $c\in
C^{q-1}(\el,V)$ exists that $d c=w$ is called {\bf coboundary};
the set of all coboundaries of dimension $q$ is denoted by $B^{q}
(\el,f)$, i.~e. $B^{q} (\el,f)=dC^{q-1} (\el,V)$. The spaces
$B^{q} (\el,f)$ and  $Z^{q} (\el,f)$ are vector subspaces of
$C^{q} (\el,V)$ and from (\ref{ddd}) we have $B^{q} (\el,f)
\subset  Z^{q} (\el,f)$. The factor space $ Z^{q} (\el,f)/B^{q}
(\el,f)=H^{q} (\el,f) $ is then called a {\bf cohomology space} of
dimension $q$ of $\el$ with respect to the representation $(V,f)$.
Directly from the definition, one obtains the following
proposition.
\begin{tvr}\label{kohoder}
$$ Z^{1} (\el,\ad_\el)=\der \el,\q B^{1} (\el,\ad_\el)=\ad \el.
$$
\end{tvr}

\section{Linear Groups}
The content of this section may be found for instance
in~\cite{Goto}. We denote the set of $n\times n$ regular matrices
by $GL(n,\Com)$. A closed set which is a subgroup of $GL(n,\Com)$
is called a {\bf linear group}. If a linear group is
$\Com$--holomorphic submanifold of $GL(n,\Com)$ then we call it
{\bf complex}. In general, a group $G$ is called a {\bf complex
Lie group} if $G$ is a $\Com$--holomorphic manifold and the map
$$G\times G \ni (a,b)\mapsto ab^{-1}\in G $$ is
$\Com$--holomorphic. A complex linear group forms a complex Lie
group. The exponential map $\exp : \gl (n,\Com)\map GL(n,\Com) $
for $A\in \gl (n,\Com)$ has the form  $$ \exp A =
\sum_{k=0}^{\infty} \frac{1}{k!} A^k .$$ It is well known that if
$G$ is a linear group then the set
\begin{equation}\label{Lieg}
\g= \{ X\in \gl (n,\Com)| \exp (\R X)\subset G \}
\end{equation}
forms a Lie algebra over $\R$ and the following propositions hold:
\begin{tvr}
  Let $G$ be a linear group in $GL(n,\Com)$ and $\g$ its Lie
  algebra. Then $G$ is complex if and only if $\g$ is a subalgebra
  of $\gl ( n,\Com )$ over $\Com$.
\end{tvr}
\begin{tvr}\label{prunikgrup}
  Let $G_1$ and $G_2$ be linear groups in $GL(n,\Com)$ and let $\g_1$ and $\g_2$ be their Lie
  algebras. Then the Lie algebra of $G_1\cap G_2$ is
  $\g_1\cap\g_2$.
\end{tvr}
A subgroup $G$ of $GL(n,\Com)$ is called an {\bf algebraic group}
if there exists a set of polynomials $P\subset \Com
[x_{11},x_{12},\dots,x_{nn}]$ such that $$G=\{(a_{ij})\in
GL(n,\Com)\,|\, p(a_{11},a_{12},\dots,a_{nn})=0 \q\forall p\in P
\} $$ An algebraic group is a linear group. Moreover,
\begin{tvr}\label{alggroup}
  An algebraic group in $GL(n,\Com)$ is a complex linear group.
\end{tvr}
A subgroup of $GL(\A)$, where $\A=(V,\,\cdot\,)$, is called an
{\bf algebraic group}, if it is represented by an algebraic group
in $GL(n,\Com)$ with respect to some basis of $V$.
\begin{thm}\label{groupI}
  Let $\A$ be an algebra over $\Com$. Then the automorphism group
  $\aut \A$ is an algebraic group in $GL(\A)$ and the Lie algebra
  of~$\aut \A$ is~$\der \A$.
\end{thm}

\section{Jordan Algebras}\label{jordan}
Basic facts about Jordan algebras included in this section may be
found in~\cite{nic2,emch,koech}. Suppose we have a (complex)
algebra $\mathcal{J}$ with multiplication $\circ$ which satisfies
for all $x,y\in \mathcal{J}$
\begin{enumerate}[(1)]
\item $x \C y= x \C y$ (commutativity)
\item $x \C (x^2\C y)= x^2 \C (x\C y)$ (Jordan's identity)
\end{enumerate}
where $x^2= x \C x$. Such an algebra $\mathcal{J}$ is called a
(complex) {\bf Jordan algebra}. In terms of structure constants
(\ref{STRUCT}) commutativity and Jordan's identity may be written
as
\begin{enumerate}[(1')]
  \item  $c^m_{ij}-c^m_{ji}=0$
  \item
  $\sum_{h,l=1}^{n}(c^h_{ik}c^l_{mj}c^r_{hl}-c^h_{ik}c^l_{hj}c^r_{lm}+c^h_{im}c^l_{kj}c^r_{hl}-c^h_{im}c^l_{hj}c^r_{lk}+c^h_{mk}c^l_{ij}c^r_{hl}-c^h_{mk}c^l_{hj}c^r_{li})=0,$\newline
  $\forall i,j,k,m,r \in \{1,\dots,n\}$.
\end{enumerate}

For an arbitrary vector space $V$ we may introduce a new
multiplication on the associative algebra $\End V$. For two linear
operators $X,Y\in\End V$ we put
\begin{equation}\label{gjgj}
 X \circ Y = \frac{1}{2}(X Y+Y X).
\end{equation}
In this way we obtain a Jordan algebra which we denote by $\jor
V$.

Note that we have defined two different products on $\End V$
(formulas (\ref{glgl}) and (\ref{gjgj})). It was pointed out in
\cite{albert} that these two products together with the original
associative composition of linear mappings exhaust all products of
the type $\la XY+\mu YX$ on $\End V$. We refine this statement in
the following way. We say that a subspace $U\subset \End V$ is
$(\la,\mu)$--{\bf closed} if there exist $\la,\mu\in \Com$, not
both zero, such that for all $X,Y\in U$ $$\la XY+\mu YX\in U $$
holds. Then one can easily prove the following
result~\cite{albert}:
\begin{tvr}\label{Albertprop}
Let $\la,\mu\in \Com$ not both zero and $U$ is $(\la,\mu)$--closed
set of $\End V$. Then $U$ is some of the following:
\begin{enumerate}[(a)]
  \item an associative subalgebra of $\End V$
  \item a Lie subalgebra of $\gl V$
  \item a Jordan subalgebra of $\jor V$.
\end{enumerate}
\end{tvr}


\chapter{Generalized Derivations}\label{CHgen}

\section{$(\alpha,\beta,\gamma)$--derivations}\label{gender}

We defined a derivation of an arbitrary algebra $\A$ as a linear
operator $D \in \End \el$ satisfying relation (\ref{dddd}). For
Lie algebras, several non--equivalent ways generalizing this
definition have recently been studied~\cite{Bresar,Hartwig,Leger}.
In this chapter we bring forward another type of generalization
introduced in~\cite{NH,HN6,HNx}.

Let $\A=(V,\,\cdot\,)$ be an arbitrary algebra. We call a linear
operator $D\in \End\A$ an
\textbf{$(\alpha,\beta,\gamma)$--derivation} of $\A$ if there
exist $\alpha,\beta,\gamma\in \Com$ such that for all $x,y\in \A$
the following relation is satisfied
\begin{equation}\label{gd}
\alpha D (x\cdot y) = \beta (D x)\cdot y + \gamma\, x\cdot (D y).
\end{equation}
For given $\alpha, \beta, \gamma \in \Com$ we denote the set of
all $(\alpha,\beta,\gamma)$--derivations as
$\gd{\alpha,\beta,\gamma}{ \A}$, i.~e.
\begin{equation}
\gd{\alpha,\beta,\gamma}{ \A} = \{ D \in \End\A\ |\ \alpha D
(x\cdot y) = \beta (D x)\cdot y + \gamma\, x\cdot (D y),\ \
\forall x,y \in \A\}.
\end{equation}
It is clear that $\gd{\alpha,\beta,\gamma}{ \A}$ is a linear
subspace of $\End\A$. The advantage of such a generalization of
derivations can be seen from the following crucial results.
\begin{thm}\label{tvr1}
Let $f:\A \map \wt\A$ be an isomorphism of complex algebras $\A$
and $\wt\A$. Then the mapping $\rho :\End\A\map \End\wt\A$,
defined by~(\ref{rho}), is an isomorphism of the vector spaces
$\gd{\alpha,\beta,\gamma}{\A}$ and
$\gd{\alpha,\beta,\gamma}{\wt\A}$, i.~e. for any $\alpha, \beta,
\gamma\in \Com$
\begin{equation}\label{invar}
\rho(\gd{\alpha,\beta,\gamma}{\A})
=\gd{\alpha,\beta,\gamma}{\wt\A}.
\end{equation}
\end{thm}
\begin{proof}
Suppose we have $\A=(V,\,\cdot\,)$ and $\wt\A=(\wt{V},\,\ast\,)$.
The isomorphism relation~(\ref{homomorph}) implies that for all
$x,y\in \wt\A$ $$x\ast y=f (f^{-1}(x)\cdot f^{-1}(y)). $$ By
rewriting the definition (\ref{gd}) we have for $D\in
\gd{\alpha,\beta,\gamma}{\A}$ $$\alpha D (f^{-1}(x)\cdot
f^{-1}(y)) = \beta (Df^{-1}x)\cdot f^{-1}y + \gamma\, f^{-1}x\cdot
(Df^{-1}y) .$$ Applying the mapping $f$ on this equation and
taking into account that $\alpha, \beta, \gamma\in \Com$ one has
\begin{equation}
\alpha f D f^{-1} (x\ast y) = \beta (f D f^{-1}x)\ast y  + \gamma
\,x\ast (f D f^{-1}y),
\end{equation}
i.~e. $f Df^{-1}\in \gd{\alpha,\beta,\gamma}{\wt\A}$.
\end{proof}
\begin{cor}\label{corgender}
For any $\alpha, \beta, \gamma\in \Com$ the dimension of the
vector space $\gd{\alpha,\beta,\gamma}{\A}$ is an invariant
characteristic of algebras.
\end{cor}

Now we restrict our investigations to commutative or
anti--commutative algebras and it follows immediately from
(\ref{gd}) that for any $\varepsilon \in \Com \backslash\{0\}$ it
holds:
\begin{equation}\label{vla1}
\gd{\alpha,\beta,\gamma}{ \A} =
\gd{\varepsilon\alpha,\varepsilon\beta,\varepsilon\gamma}{ \A}=
\gd{\alpha,\gamma,\beta}{ \A}.
\end{equation} Furthermore, we have the
following important property.
\begin{lemma} Let $\A$ be a commutative or
anti--commutative algebra. Then for any $\alpha, \beta, \gamma \in
\Com$
\begin{equation}\label{vla2}
\gd{\alpha,\beta,\gamma}{\A} = \gd{0,\beta-\gamma,\gamma
-\beta}{\A}\, \cap\, \gd{2\alpha,\beta+\gamma,\beta+\gamma}{\A}
\end{equation}
holds.
\end{lemma}

\begin{proof} Suppose any $\alpha, \beta, \gamma \in \Com$ are given. We carry out the proof
for an anti--commutative algebra $\A_-$ with a multiplication
$[\,,\,]_{-}$, i.~e. for all $x,y\in \A_-$ the relation
$[x,y]_-=-[y,x]_- $ is satisfied; the proof for a commutative
algebra is analogous. Then for $D \in
\gd{\alpha,\beta,\gamma}{\A_-}$
 and arbitrary $x,y\in \A_-$ we have
    \begin{eqnarray}
    \alpha D[x,y]_- & = & \beta[Dx,y]_- + \gamma [x,Dy]_- \label{trik1} \\
    \alpha D[y,x]_- & = & \beta[Dy,x]_- + \gamma [y,Dx]_-. \nonumber
    \end{eqnarray}
    By adding and subtracting equations (\ref{trik1}) we obtain
     \begin{eqnarray}
    0 & = & (\beta -\gamma)([Dx,y]_- - [x,Dy]_-) \label{trik2} \\
    2\alpha D[x,y]_- & = & (\beta + \gamma)([Dx,y]_- + [x,Dy]_-) \nonumber
     \end{eqnarray}
    and thus $\gd{\alpha,\beta,\gamma}{\A_-} \subset \gd{0,\beta-\gamma,\gamma -\beta}{\A_-}\,
    \cap\,
\gd{2\alpha,\beta+\gamma,\beta+\gamma}{\A_-}$. Similarly, starting
with equations (\ref{trik2}) we obtain equations (\ref{trik1}) and
the remaining inclusion is proven.
\end{proof}
Further, we proceed to formulate the theorem which reveals the
structure of the spaces $\gd{\alpha,\beta,\gamma}{\A}$; the three
original parameters are in fact reduced to only one.
\begin{thm}\label{klass} Let $\A$ be a commutative or
anti--commutative algebra. Then for any $\alpha, \beta, \gamma \in
\Com$ there exists $\delta\in \Com$ such that the subspace
$\gd{\alpha,\beta,\gamma}{\A} \subset \End \A$ is equal to some of
the four following subspaces:
\begin{enumerate}
 \item $\gd{\delta,0,0}{\A}$
 \item $\gd{\delta,1,-1}{\A}$
 \item $\gd{\delta,1,0}{\A}$
 \item $\gd{\delta,1,1}{\A}$.
\end{enumerate}
\end{thm}
\begin{proof}
\begin{enumerate}
    \item Suppose $\beta + \gamma = 0 $. Then either $\beta = \gamma
    =0$ or $\beta = -\gamma \neq 0$.
    \begin{enumerate}
    \item For $\beta = \gamma = 0$, we
    have $$\gd{\alpha,\beta,\gamma}{\A}  =  \gd{\alpha,0,0}{\A}.$$
    \item For $\beta = -\gamma \neq 0$, it follows from~(\ref{vla1}),
    (\ref{vla2}):
 $$\gd{\alpha,\beta,\gamma}{\A} = \gd{0,\beta-\gamma,\gamma-\beta}{\A}\cap
 \gd{2\alpha,0,0}{\A}
    = \gd{0,1,-1}{\A}\cap \gd{\alpha,0,0}{\A}.$$ On the other hand it holds $$ \gd{\alpha,1,-1}{\A} = \gd{0,2,-2}{\A}\cap
    \gd{2\alpha,0,0}{\A}
    = \gd{0,1,-1}{\A}\cap \gd{\alpha,0,0}{\A}$$ and therefore
    $$\gd{\alpha,\beta,\gamma}{\A}= \gd{\alpha,1,-1}{\A}.$$
    \end{enumerate}
    \item Suppose $\beta + \gamma \neq 0 $. Then either $\beta - \gamma \neq
    0$ or $\beta = \gamma \neq 0 $.
    \begin{enumerate}
        \item For $\beta - \gamma \neq 0$, we have
    $$\gd{\alpha,\beta,\gamma}{\A} = \gd{0,\beta-\gamma,\gamma-\beta}{\A}\cap
    \gd{2\alpha,\beta+\gamma,\beta+\gamma}{\A} = \gd{0,1,-1}{\A}\cap
    \gd{\frac{2\alpha}{\beta+\gamma},1,1}{\A}$$
    and taking into account (\ref{vla2}), this is equal to
    $\gd{\frac{\alpha}{\beta+\gamma},1,0}{\A}$, i.~e.
    $$\gd{\alpha,\beta,\gamma}{\A} =
    \gd{\frac{\alpha}{\beta+\gamma},1,0}{\A}.$$
        \item For $\beta = \gamma \neq 0 $ we have
        $$\gd{\alpha,\beta,\gamma}{\A} = \gd{\frac{\alpha}{\beta},1,1}{\A}.$$
    \end{enumerate}
\end{enumerate}
\end{proof}

We define two complex functions with fundamental property
--- invariance under isomorphisms. We use the one--parametric sets $\gd{\alpha,1,0}{\A}$ and
$\gd{\alpha,1,1}{\A}$ from Theorem~\ref{klass} to define invariant
functions of an arbitrary algebra $\A$. Functions
$\fa\el,\fa^0\el:\Com \rightarrow \{0,1,\dots,(\dim\A)^2\}$
defined by the formulas
\begin{align}
(\fa\A)(\alpha) =& \dim\gd{\alpha,1,1}{\A}\\ (\fa^0\A)(\alpha) =&
\dim\gd{\alpha,1,0}{\A}
\end{align}
are called {\bf invariant functions} corresponding to
$(\alpha,\beta,\gamma)$--derivations of an algebra $\A$.

The following statement follows immediately from
Theorem~\ref{tvr1}.
\begin{cor}\label{invarfunc} If two complex algebras $\A,\wt\A$ are isomorphic, $\A \cong
\wt\A$, then it holds:
\begin{enumerate}
\item $\fa\A = \fa\wt\A$,
\item $\fa^0\A = \fa^0\wt\A$.
\end{enumerate}
\end{cor}
Note that sometimes in the literature the name 'invariant
functions' denotes (formal) Casimir invariants; their form,
however, depends on the choice of the basis of $\el$. Here by
invariant functions we rather mean 'basis independent' complex
functions, such as~$\fa$ and~$\fa^0$.

\section{$(\alpha,\beta,\gamma)$--derivations of Lie Algebras}
Suppose we have a complex Lie algebra $\el$ and let us discuss in
detail possible outcome of Theorem~\ref{klass}.
\begin{thm}\label{Lielist} Let $\el$ be a complex Lie algebra and
$\alpha,\beta,\gamma\in\Com$ not all zero. Then the space
$\gd{\alpha,\beta,\gamma}{\el}$ is equal to some of the following:
\begin{enumerate}
\item Lie algebra of derivations $\gd{1,1,1}{\el}\subset
\gl\el$,
\item Lie algebra $\gd{0,1,1}{\el}\subset \gl\el$,
\item associative algebra $\gd{1,1,0}{\el}=C_{\ad}(\el)\subset \gl\el $,
\item associative algebra $\gd{1,0,0}{\el}\subset\End\el$ of dimension
\begin{equation}\label{fle1}
\dim \gd{1,0,0}{\el} = \operatorname{codim}\el^2
\dim\el,\end{equation}\label{GD100}\vspace{-24pt}
\item associative algebra $\gd{0,1,0}{\el}\subset\End\el$ of dimension
\begin{equation}\label{fle2}\dim \gd{0,1,0}{\el} = \dim \el \, \dim C(\el),\end{equation}\label{GD010}\vspace{-24pt}
\item Jordan algebra $\gd{1,1,-1}{\el}\subset \jor\el$,
\item Jordan algebra $\gd{0,1,-1}{\el}\subset \jor\el$,
\item subspace $\gd{\delta,1,0}\el$, for some $\delta\in\Com$, $\delta\neq
0,1$.
\item subspace $\gd{\delta,1,1}\el$, for some $\delta\in\Com$, $\delta\neq
0,1$.
\end{enumerate}
\end{thm}
\begin{proof}
We list the four cases in Theorem~\ref{klass} and discuss all
possible values of the parameter $\delta\in \Com$:
    \begin{enumerate}
        \item $\gd{\delta,0,0}{\el}$:
                \begin{enumerate}
                \item Since we assumed some $\alpha,\beta, \gamma\in \Com$ non-zero, the case $\delta = 0 $ cannot occur.
                \item For $\delta \neq 0$, the space $\gd{1,0,0}{\el}$ is an associative subalgebra of $\End \el$, which maps the derived
                algebra $\el^2 = [\el,\el]$ to the zero vector:
                $$\gd{1,0,0}{\el} = \{A \in \End\el\ |\ A(\el^2) =
                0\},$$
                and therefore its dimension is as follows:
                      $$\dim \gd{1,0,0}{\el} = \operatorname{codim}\el^2 \dim\el.$$
                \end{enumerate}
        \item $\gd{\delta,1,-1}{\el}$:
                \begin{enumerate}
                \item For $\delta = 0$, we have a Jordan algebra
                $\gd{0,1,-1}{\el}\subset \jor \el$,
                $$\gd{0,1,-1}{\el} = \{A \in \End\el\ |\ [Ax,y] = [x,A y],\ \ \forall
                x,y\in\el \}.$$ The proof of the property $A,B\in \gd{0,1,-1}{\el}\Rightarrow\frac{1}{2}(AB+BA)\in \gd{0,1,-1}{\el}
                $ is straightforward.
                \item For $\delta \neq 0$, we obtain the Jordan algebra
                $\gd{1,1,-1}{\el}\subset \jor \el$ as an intersection of two Jordan
                algebras:
                \begin{align*} \gd{\delta,1,-1}{\el} =& \gd{0,1,-1}{\el} \cap
                \gd{\delta,0,0}{\el} = \gd{0,1,-1}{\el} \cap \gd{1,0,0}{\el}\\=&\gd{1,1,-1}{\el}
                .\end{align*}
                \end{enumerate}
        \item $\gd{\delta,1,0}{\el}$:
                \begin{enumerate}
                \item For $\delta = 0$, we get an associative algebra of all linear operators
                of the vector space $\el$, which maps the whole $\el$ into its center $ C(\el)$:
                $$ \gd{0,1,0}{\el} = \{A \in \End\el\ |\ A(\el) \subseteq
                C(\el)\}, $$ and its dimension is
                        $$\dim \gd{0,1,0}{\el} = \dim \el \, \dim C(\el).$$

                \item For $\delta = 1$, the space $\gd{1,1,0}{\el}$ is the centralizer of the adjoint representation $C_{\ad} (\el)$, see (\ref{centralizer}).

                \item For the remaining values of $\delta$, the space $\gd{\delta,1,0}{\el}$
                forms, in the general case of a Lie algebra $\el$,
                only a vector subspace of
                $\End\el$. Thus we have the one--parametric set of vector
                spaces:
                 $$ \gd{\delta,1,0}{\el} = \gd{0,1,-1}{\el} \cap
                \gd{2\delta,1,1}{\el}. $$
                \end{enumerate}
        \item $\gd{\delta,1,1}{\el}$:
                \begin{enumerate}
                \item For $\delta = 0$, we have a Lie algebra
                $$\gd{0,1,1}{\el} = \{A \in \End\el\ |\ [Ax,y] = -[x,A y],\ \ \forall x,y\in\el \}.$$
                Verification of the property $A,B\in
                \gd{0,1,1}{\el}\Rightarrow(AB-BA)\in \gd{0,1,1}{\el}
                $ is straightforward.
                \item For $\delta = 1$, we get the algebra of derivations of
                $\el$,
                    $$ \gd{1,1,1}{\el} = \der \el.$$
                \item For the remaining values of $\delta$, the space $\gd{\delta,1,1}{\el}$
                forms, in the general case of a Lie algebra $\el$,
                 only a vector subspace of
                $\End\el$.
                \end{enumerate}
    \end{enumerate}
\end{proof}

Since the definition of $(\alpha,\beta,\gamma)$--derivations
partially overlaps other generalizations, some of the sets from
the above theorem naturally appeared already in the literature.
For instance, a considerable amount of theory concerning relations
between $\gd{1,1,0}{\el}$ and $\gd{0,1,-1}{\el}$ has been
developed in~\cite{Leger}. Later on, we are mostly interested in
the form of the one--parametric spaces $\gd{\delta,1,1}\el$ and
$\gd{\delta,1,0}\el$.

\begin{example}
If $\el$ is a simple complex Lie algebra then $\gd{1,1,1} \cong
\el$, $$ \gd{1,0,0}{\el} = \gd{0,1,0}{\el} = \{0\},$$ and
$\gd{1,1,0}{\el}$ is the one--dimensional Lie algebra containing
multiples of the identity operator.
\end{example}
\subsection{Intersections of the Spaces $\gd{\alpha,\beta,\gamma}{\el}$
}\label{intersections} A thorough study of various intersections
of two different subspaces $\gd{\alpha,\beta,\gamma}{\el}$ turned
out to be very valuable. Besides new independent invariants, we
also obtain a new operator algebra with a non--trivial structure,
as well as restrictions on $\fa\el,\,\fa^0\el$. We commence with:
\begin{thm}\label{tvrprun1}
Let $f:\A \map \wt\A$ be an isomorphism of complex algebras $\A$
and $\wt\A$. Then the mapping $\rho :\End\A\map \End\wt\A$,
defined by~(\ref{rho}), is an isomorphism of the vector spaces
$\gd{\alpha,\beta,\gamma}{\A}\cap\gd{\alpha',\beta',\gamma'}{\A}$
and
$\gd{\alpha,\beta,\gamma}{\wt\A}\cap\gd{\alpha',\beta',\gamma'}{\wt\A}$,
i.~e. for any $\alpha, \beta, \gamma,\alpha', \beta', \gamma'\in
\Com$
\begin{equation}\label{rhoabcprun}
\rho(\gd{\alpha,\beta,\gamma}{\A}\cap\gd{\alpha',\beta',\gamma'}{\A}
)
=\gd{\alpha,\beta,\gamma}{\wt\A}\cap\gd{\alpha',\beta',\gamma'}{\wt\A}.
\end{equation}
\end{thm}
\begin{proof}
Suppose we have $\A=(V,\cdot)$ and $\wt\A=(\wt{V},\ast)$. The
isomorphism relation~(\ref{homomorph}) implies that for all
$x,y\in \wt\A$ the relation $x\ast y=f (f^{-1}(x)\cdot f^{-1}(y))$
holds. By rewriting the definition~(\ref{gd}) we obtain $D\in
\gd{\alpha,\beta,\gamma}{\A}\cap\gd{\alpha',\beta',\gamma'}{\A}$
if and only if both of the equations
\begin{align*}
\alpha D (f^{-1}(x)\cdot f^{-1}(y)) &= \beta (Df^{-1}x)\cdot
f^{-1}y + \gamma\, f^{-1}x\cdot (Df^{-1}y)\\ \alpha' D
(f^{-1}(x)\cdot f^{-1}(y)) &= \beta' (Df^{-1}x)\cdot f^{-1}y +
\gamma'\, f^{-1}x\cdot (Df^{-1}y)
\end{align*}
are satisfied for all $x,y\in \wt\A$. Applying the mapping $f$ on
these two equations and taking into account that $\alpha, \beta,
\gamma, \alpha', \beta', \gamma'$ are complex numbers, one has
\begin{align*}
\alpha f D f^{-1} (x\ast y) &= \beta (f D f^{-1}x)\ast y  + \gamma
\,x\ast (f D f^{-1}y),\\ \alpha' f D f^{-1} (x\ast y) &= \beta' (f
D f^{-1}x)\ast y  + \gamma' \,x\ast (f D f^{-1}y),
\end{align*}
i.~e. $f Df^{-1}\in
\gd{\alpha,\beta,\gamma}{\wt\A}\cap\gd{\alpha',\beta',\gamma'}{\wt\A}$.
\end{proof}
\begin{cor}\label{corgenderprun}
For any $\alpha, \beta, \gamma,\alpha', \beta', \gamma'\in \Com$
the number
$$\dim\,(\gd{\alpha,\beta,\gamma}{\A}\cap\gd{\alpha',\beta',\gamma'}{\A})$$
is an invariant characteristic of algebras.
\end{cor}

In our search for new invariants of complex Lie algebras, we
systematically explored all possible intersections of the spaces
$\gd{\alpha,\beta,\gamma}{\el}$. We classify these intersections
in the following theorem.

\begin{thm}\label{prunklass}
Let $\el$ be a complex Lie algebra. Suppose
$\alpha,\beta,\gamma\in\Com$ are not all zero and
$\alpha',\beta',\gamma'\in\Com$ are not all zero. Then the
intersection
$\gd{\alpha,\beta,\gamma}{\el}\cap\gd{\alpha',\beta',\gamma'}{\el}$
is equal to some of the cases 1. -- 9. of Theorem~\ref{Lielist} or
to some of the following:
\begin{enumerate}
\item associative algebra  $\gd{1,0,0}{\el}\cap\gd{0,1,0}{\el}\subset\End{\el}$ of
dimension \begin{equation}\label{dimprunex}
\dim(\gd{1,0,0}{\el}\cap\gd{0,1,0}{\el}) = \operatorname{codim}
\el^2
    \dim C(\el),\end{equation}
\item Lie algebra $\gd{1,1,1}{\el}\cap \gd{0,1,1}{\el}\subset
\gl\el$.
\end{enumerate}
\end{thm}
\begin{proof}
According to Theorem~\ref{Lielist}, for given
$\alpha,\beta,\gamma,\alpha',\beta',\gamma' \in\Com$ there exists
$\delta\in\Com$ such that $\gd{\alpha,\beta,\gamma}\el$ is equal
to some of the spaces:
$$\gd{\delta,1,1}\el,\gd{\delta,1,0}\el,\gd{1,0,0}\el,\gd{1,1,-1}\el,\gd{0,1,-1}\el,$$
and $\delta'\in\Com$ such that $\gd{\alpha',\beta',\gamma'}\el$ is
equal to some of the spaces:
$$\gd{\delta',1,1}\el,\gd{\delta',1,0}\el,\gd{1,0,0}\el,\gd{1,1,-1}\el,\gd{0,1,-1}\el.$$
There are 17 possible pairs the above spaces. Five pairs consist
of equal spaces and corresponding intersections lead to some of
the cases 1. -- 9. of Theorem~\ref{Lielist}. Then there are $10$
obvious pairs with different subspaces, plus two pairs
$\gd{\delta,1,1}\el\cap\gd{\delta',1,1}\el$ and
$\gd{\delta,1,0}\el\cap\gd{\delta',1,0}\el$ with
$\delta\neq\delta'$. Assuming $\delta\neq\delta'$, the following
equalities among intersections are obvious:
\begin{align}
&\gd{1,0,0}{\el} \cap \gd{0,1,0}{\el}  =\gd{1,0,0}{\el} \cap
\gd{\delta,1,0}{\el}=\gd{\delta,1,0}{\el}\cap
\gd{\delta',1,0}{\el}\label{obv1}\\ &\gd{1,0,0}{\el} \cap
\gd{0,1,1}{\el} =\gd{1,0,0}{\el} \cap \gd{\delta,1,1}{\el} =
\gd{\delta,1,1}{\el} \cap \gd{\delta',1,1}{\el}\label{obv2}\\
&\gd{1,0,0}{\el} \cap \gd{0,1,-1}{\el}=\gd{1,0,0}{\el} \cap
\gd{1,1,-1}{\el}=\gd{1,1,-1}{\el} \cap
\gd{0,1,-1}{\el}\label{obv3}
\end{align}
Using Lemma~\ref{vla2} we obtain:
\begin{align}
&\gd{1,1,-1}{\el}=\gd{0,1,-1}{\el} \cap
\gd{1,0,0}{\el}\label{lemm1}\\
 &
\gd{\delta,1,0}{\el}=\gd{0,1,-1}{\el}\cap\gd{2\delta,1,1}{\el}\label{lemm2}
\end{align}
The equality~(\ref{lemm1}) implies that all {\it three}
intersections in~(\ref{obv3}) are equal to the case~6. of
Theorem~\ref{Lielist}. Using successively~(\ref{lemm1}),
(\ref{lemm2}) and (\ref{obv1}), we obtain:
\begin{align}
\gd{1,1,-1}{\el} \cap \gd{\delta,1,1}{\el}&=\gd{1,0,0}{\el} \cap
\gd{0,1,-1}{\el}\cap\gd{\delta,1,1}{\el} \nonumber\\
&=\gd{1,0,0}{\el}\cap\gd{\frac{\delta}{2},1,0}{\el}=\gd{1,0,0}{\el}\cap\gd{0,1,0}{\el}.\label{nobv1}
\end{align}
Using successively~(\ref{lemm2}), (\ref{obv3}) and (\ref{lemm1}),
(\ref{nobv1}) we obtain:
\begin{align*}
\gd{1,1,-1}{\el} \cap \gd{\delta,1,0}{\el}&=\gd{1,1,-1}{\el} \cap
\gd{0,1,-1}{\el}\cap\gd{2\delta,1,1}{\el} \nonumber\\
&=\gd{1,1,-1}{\el}\cap\gd{2\delta,1,1}{\el}=\gd{1,0,0}{\el}\cap\gd{0,1,0}{\el}.
\end{align*}
Using successively~(\ref{lemm2}), (\ref{obv2}), (\ref{lemm1}),
(\ref{nobv1}) and assuming firstly $\delta'\neq 2 \delta$, we
obtain:
\begin{align*}
\gd{\delta,1,0}{\el} \cap \gd{\delta',1,1}{\el}&=\gd{0,1,-1}{\el}
\cap \gd{2\delta,1,1}{\el}\cap\gd{\delta',1,1}{\el} \nonumber\\
&=\gd{0,1,-1}{\el}\cap\gd{1,0,0}{\el}\cap\gd{0,1,1}{\el}\nonumber\\
&= \gd{1,1,-1}{\el}\cap\gd{0,1,1}{\el}
=\gd{1,0,0}{\el}\cap\gd{0,1,0}{\el}.
\end{align*}
Secondly, using twice~(\ref{lemm2}) we obtain:
\begin{align*}
\gd{\delta,1,0}{\el} \cap \gd{2\delta,1,1}{\el}&=\gd{0,1,-1}{\el}
\cap \gd{2\delta,1,1}{\el}\cap\gd{2\delta,1,1}{\el} \nonumber\\
&=\gd{\delta,1,0}{\el}.
\end{align*}
Since from (\ref{vla1}) follows
$\gd{\delta,1,0}{\el}=\gd{\delta,0,1}{\el}$, we have for
$A\in\gd{\delta,1,0}{\el}$ and all $x,y\in \el$: $$\delta
A[x,y]=[Ax,y]=[x,Ay]. $$ Thus, we have $A\in \gd{0,1,-1}{\el} $
and the inclusion
\begin{equation}\label{inclustrap}
\gd{\delta,1,0}{\el}\subset \gd{0,1,-1}{\el}
\end{equation}
implies
\begin{equation*}
\gd{\delta,1,0}{\el}=\gd{0,1,-1}{\el}\cap\gd{\delta,1,0}{\el}.
\end{equation*}

The space $\gd{1,0,0}{\el}\cap\gd{0,1,0}{\el}$, as the
intersection of two associative subalgebras of~$\End \el$, forms
also an associative subalgebra of~$\End \el$; the space
$$\gd{1,0,0}{\el}\cap\gd{0,1,1}{\el}=\gd{1,1,1}{\el}\cap\gd{0,1,1}{\el},$$
as an intersection of two Lie subalgebras of~$\gl \el$, forms also
a Lie subalgebra of~$\gl \el$.  
\end{proof}

\begin{lemma}\label{trap}
Let $\el$ be a complex Lie algebra. Then for all
$\alpha,\beta,\gamma\in\Com$ $$\gd{1,0,0}{\el}\cap\gd{0,1,0}{\el}
\subset\gd{\alpha,\beta,\gamma}{\el}.$$
\end{lemma}
\begin{proof}
Let $A\in\gd{1,0,0}{\el}\cap\gd{0,1,0}{\el}.$ Since from
(\ref{vla1}) follows $\gd{0,1,0}{\el}=\gd{0,0,1}{\el}$, we have
for all $x,y\in\el$:
\begin{align}
A[x,y]=& 0 \label{trap1}\\ 0=& [Ax,y]  \label{trap2}\\ 0=&[x,Ay]
\label{trap3}
\end{align}
Multiplying (\ref{trap1}), (\ref{trap2}), (\ref{trap3}) by
$\alpha,\beta,\gamma\in\Com$, respectively, and summing these
equations one has: $$\alpha A[x,y]=\beta  [Ax,y]+\gamma[x,Ay],  $$
i.~e. $A\in \gd{\alpha,\beta,\gamma}{\el}.$
\end{proof}
\begin{cor}\label{inewww}
Let $\el$ be a complex Lie algebra. Then the following
inequalities hold:
\begin{align}
\operatorname{codim} \el^2 \dim C(\el)&\leq  \fa^0\el\leq \dim
\gd{0,1,-1}{\el}\label{ineq1} \\ \fa^0\el (\alpha)&\leq\fa\el
(2\alpha),\q \forall \al\in \Com\label{ineq2} \\
\operatorname{codim} \el^2 \dim C(\el)&\leq  \fa\el.\label{ineq3}
\end{align}
\end{cor}
\begin{proof}
The inequality~(\ref{ineq3}) and the first part of the
inequality~(\ref{ineq1}) follow directly from Lemma~\ref{trap} and
Theorem~\ref{prunklass}. Since we have from~(\ref{lemm2}) the
inclusion $\gd{\delta,1,0}{\el}\subset\gd{2\delta,1,1}{\el}$, the
inequality~(\ref{ineq2}) follows. The second part of the
inequality~(\ref{ineq1}) follows directly from~(\ref{inclustrap}).
\end{proof}
\begin{example}
We demonstrate the non--triviality of the inequalities in
Corollary~\ref{inewww}. Consider the four--dimensional Lie algebra
$\el=\slp (2,\Com)\oplus\g_1$ with non--zero commutation
relations: $[e_1,e_2]=e_1,\ [e_2,e_3]=e_3,\ [e_1,e_3]=2e_2$. The
invariant functions $\fa\el$ and $\fa^0\el$ have the following
form: \vspace{-8pt}
\begin{center}
\begin{tabular}[t]{|l||c|c|c|c|c|}
\hline \parbox[l][20pt][c]{0pt}{}   $\alpha$ & 1 & $0$ & $-1$& 2&
\\ \hline
\parbox[l][20pt][c]{0pt}{} $\fa\el(\alpha)$ & 4 & 4 &6& 2 & 1  \\ \hline
\end{tabular}\qquad
\begin{tabular}[t]{|l||c|c|c|c|c|}
\hline  \parbox[l][20pt][c]{0pt}{}  $\alpha$ & $1$& 0 &
\\ \hline
\parbox[l][20pt][c]{0pt}{} $\fa^0\el(\alpha)$ & 2  & 4& 1 \\
\hline
\end{tabular}
\end{center}
A blank space in the table of the function $\fa$ denotes a general
complex number, different from all previously listed values, i.~e.
it holds: $\fa\el(\al)=1$, $\al\neq 0,\pm 1,2$. It is clear that
$\dim C(\el)=1$ and $\dim \el^2=3$. Hence we have
$\operatorname{codim} \el^2=1$. We also calculate $$\dim
\gd{0,1,-1}{\el}=5$$ and obtain from Corollary~\ref{inewww} the
following inequalities
\begin{align*}
1&\leq  \fa^0\el\leq 5 \\ 4=\fa^0\el (0)&\leq\fa\el(0)=4
\\ 2=\fa^0\el (1)&\leq\fa\el(2)=2\\
1=\fa^0\el (\alpha)&\leq\fa\el(2\alpha)=1,\q \forall \al\in \Com,
\al\neq 0, 1, \pm 1/2
\\ 1&\leq \fa\el.
\end{align*}
\end{example}
%

\subsection{$(\alpha,\beta,\gamma)$--derivations of Low--dimensional Lie Algebras}\label{exder} In Section \ref{APAL} of
Appendix~\ref{APA} all $(\alpha,\beta,\gamma)$--derivations of all
two and three dimensional non--abelian Lie algebras are listed.
Here we present one typical example of
$(\alpha,\beta,\gamma)$--derivations for a four--dimensional Lie
algebra.
\begin{example}\label{4dimder}  The four--dimensional Lie algebra $\g_{4,2}(a)$ has
non--zero commutation relations $[e_1,e_4]=ae_1,\ [e_2,e_4]=e_2, \
[e_3,e_4]=e_2+e_3$, with a complex parameter $(a\neq 0,\pm 1 ,-2)
$. We present the complete set of its
$(\alpha,\beta,\gamma)$--derivations. Commutation relations of the
presented matrix Lie and Jordan algebras are placed in Appendix
\ref{IFLJ}. Note especially the form of the one--parametric
subspace $\gd{\delta,1,1}\g_{4,2}(a)$. We encounter here for the
first time an important phenomenon: the dimensionality of the
matrix subspace $\gd{\delta,1,1}\g_{4,2}(a)$ {\it depends} on the
value of the parameter $a\in\Com$. In the following formulas, we
abbreviate the notation and write $\gd{\alpha,\beta,\gamma}$
instead of $\gd{\alpha,\beta,\gamma}\g_{4,2}(a)$.

\begin{flushleft}
$\:\:\gd{1,1,1}= \Span_{\Com}{ \left\{ \left(
  \begin{smallmatrix}
    0 & 0 & 0 & 1 \\
    0 & 0 & 0 & 0 \\
    0 & 0 & 0 & 0 \\
    0 & 0 & 0 & 0
  \end{smallmatrix}
   \right),
   \left(
  \begin{smallmatrix}
    0 & 0 & 0 & 0 \\
    0 & 0 & 0 & 1 \\
    0 & 0 & 0 & 0 \\
    0 & 0 & 0 & 0
  \end{smallmatrix}
   \right),
   \left(
  \begin{smallmatrix}
    0 & 0 & 0 & 0 \\
    0 & 0 & 0 & 0 \\
    0 & 0 & 0 & 1 \\
    0 & 0 & 0 & 0
  \end{smallmatrix}
   \right),
   \left(
  \begin{smallmatrix}
    0 & 0 & 0 & 0 \\
    0 & 0 & 1 & 0 \\
    0 & 0 & 0 & 0 \\
    0 & 0 & 0 & 0
  \end{smallmatrix}
   \right),
   \left(
  \begin{smallmatrix}
    1 & 0 & 0 & 0 \\
    0 & 0 & 0 & 0 \\
    0 & 0 & 0 & 0 \\
    0 & 0 & 0 & 0
  \end{smallmatrix}
   \right),
   \left(
  \begin{smallmatrix}
    0 & 0 & 0 & 0 \\
    0 & 1 & 0 & 0 \\
    0 & 0 & 1 & 0 \\
    0 & 0 & 0 & 0
  \end{smallmatrix}
   \right)
   \right\}} $
\end{flushleft}
\begin{flushleft}
$\:\:\gd{0,1,1}= \Span_{\Com}{ \left\{ \left(
  \begin{smallmatrix}
    0 & 0 & 0 & 1 \\
    0 & 0 & 0 & 0 \\
    0 & 0 & 0 & 0 \\
    0 & 0 & 0 & 0
  \end{smallmatrix}
   \right),
   \left(
  \begin{smallmatrix}
    0 & 0 & 0 & 0 \\
    0 & 0 & 0 & 1 \\
    0 & 0 & 0 & 0 \\
    0 & 0 & 0 & 0
  \end{smallmatrix}
   \right),
   \left(
  \begin{smallmatrix}
    0 & 0 & 0 & 0 \\
    0 & 0 & 0 & 0 \\
    0 & 0 & 0 & 1 \\
    0 & 0 & 0 & 0
  \end{smallmatrix}
   \right),
   \left(
  \begin{smallmatrix}
    1 & 0 & 0 & 0 \\
    0 & 1 & 0 & 0 \\
    0 & 0 & 1 & 0 \\
    0 & 0 & 0 & -1
  \end{smallmatrix}
   \right)
   \right\} \cong \g_{4,5}(1,1) }$
\end{flushleft}
\begin{flushleft}
$\:\:\gd{1,1,0}= \Span_{\Com}{ \left\{ \left(
  \begin{smallmatrix}
    1 & 0 & 0 & 0 \\
    0 & 1 & 0 & 0 \\
    0 & 0 & 1 & 0 \\
    0 & 0 & 0 & 1
  \end{smallmatrix}
   \right)
   \right\} \cong \g_{1} }$
\end{flushleft}
\begin{flushleft}
$\:\:\gd{1,0,0}= \Span_{\Com}{ \left\{ \left(
  \begin{smallmatrix}
    0 & 0 & 0 & 1 \\
    0 & 0 & 0 & 0 \\
    0 & 0 & 0 & 0 \\
    0 & 0 & 0 & 0
  \end{smallmatrix}
   \right),
   \left(
  \begin{smallmatrix}
    0 & 0 & 0 & 0 \\
    0 & 0 & 0 & 1 \\
    0 & 0 & 0 & 0 \\
    0 & 0 & 0 & 0
  \end{smallmatrix}
   \right),
   \left(
  \begin{smallmatrix}
    0 & 0 & 0 & 0 \\
    0 & 0 & 0 & 0 \\
    0 & 0 & 0 & 1 \\
    0 & 0 & 0 & 0
  \end{smallmatrix}
   \right),
   \left(
  \begin{smallmatrix}
    0 & 0 & 0 & 0 \\
    0 & 0 & 0 & 0 \\
    0 & 0 & 0 & 0 \\
    0 & 0 & 0 & 1
  \end{smallmatrix}
   \right)
   \right\} \cong \g_{4,5}(1,1) }$
\end{flushleft}
\begin{flushleft}
 $\:\:\gd{0,1,0}=\{0\}$
\end{flushleft}
\begin{flushleft}
$\:\:\gd{1,1,1} \cap \gd{0,1,1}= \Span_{\Com}{ \left\{ \left(
  \begin{smallmatrix}
    0 & 0 & 0 & 1 \\
    0 & 0 & 0 & 0 \\
    0 & 0 & 0 & 0 \\
    0 & 0 & 0 & 0
  \end{smallmatrix}
   \right),
   \left(
  \begin{smallmatrix}
    0 & 0 & 0 & 0 \\
    0 & 0 & 0 & 1 \\
    0 & 0 & 0 & 0 \\
    0 & 0 & 0 & 0
  \end{smallmatrix}
   \right),
   \left(
  \begin{smallmatrix}
    0 & 0 & 0 & 0 \\
    0 & 0 & 0 & 0 \\
    0 & 0 & 0 & 1 \\
    0 & 0 & 0 & 0
  \end{smallmatrix}
   \right)
   \right\} \cong 3\g_{1} }$
\end{flushleft}
\begin{flushleft}
 $\:\:\gd{1,0,0}\cap \gd{0,1,0}=\{0\}$
\end{flushleft}
\begin{flushleft}
 $\:\:\gd{1,1,-1}=\{0\}$
\end{flushleft}
\begin{flushleft}
$\:\:\gd{0,1,-1}= \Span_{\Com}{ \left\{
   \left(
  \begin{smallmatrix}
    1 & 0 & 0 & 0 \\
    0 & 1 & 0 & 0 \\
    0 & 0 & 1 & 0 \\
    0 & 0 & 0 & 1
  \end{smallmatrix}
   \right)
   \right\} \cong \j_{1} }$
\end{flushleft}
\begin{flushleft}
$\:\:\gd{\delta,1,0}=\{0\}_{\delta\neq 1}$
\end{flushleft}
\begin{flushleft}
$\:\:\gd{\delta,1,1}= \Span_{\Com}{ \left\{ \left(
  \begin{smallmatrix}
    0 & 0 & 0 & 1 \\
    0 & 0 & 0 & 0 \\
    0 & 0 & 0 & 0 \\
    0 & 0 & 0 & 0
  \end{smallmatrix}
   \right),
   \left(
  \begin{smallmatrix}
    0 & 0 & 0 & 0 \\
    0 & 0 & 0 & 1 \\
    0 & 0 & 0 & 0 \\
    0 & 0 & 0 & 0
  \end{smallmatrix}
   \right),
   \left(
  \begin{smallmatrix}
    0 & 0 & 0 & 0 \\
    0 & 0 & 0 & 0 \\
    0 & 0 & 0 & 1 \\
    0 & 0 & 0 & 0
  \end{smallmatrix}
   \right),
   \left(
  \begin{smallmatrix}
    1 & 0 & 0 & 0 \\
    0 & 1 & 0 & 0 \\
    0 & 0 & 1 & 0 \\
    0 & 0 & 0 & -1+\delta
  \end{smallmatrix}
   \right)
   \right\}_{\delta\neq 1,a,1/a} }$
\end{flushleft}
\begin{flushleft}
$\:\:\gd{a,1,1}= \Span_{\Com}{ \left\{ \left(
  \begin{smallmatrix}
    0 & 0 & 0 & 1 \\
    0 & 0 & 0 & 0 \\
    0 & 0 & 0 & 0 \\
    0 & 0 & 0 & 0
  \end{smallmatrix}
   \right),
   \left(
  \begin{smallmatrix}
    0 & 0 & 0 & 0 \\
    0 & 0 & 0 & 1 \\
    0 & 0 & 0 & 0 \\
    0 & 0 & 0 & 0
  \end{smallmatrix}
   \right),
   \left(
  \begin{smallmatrix}
    0 & 0 & 0 & 0 \\
    0 & 0 & 0 & 0 \\
    0 & 0 & 0 & 1 \\
    0 & 0 & 0 & 0
  \end{smallmatrix}
   \right),\left(
  \begin{smallmatrix}
    0 & 0 & 1 & 0 \\
    0 & 0 & 0 & 0 \\
    0 & 0 & 0 & 0 \\
    0 & 0 & 0 & 0
  \end{smallmatrix}
   \right),
   \left(
  \begin{smallmatrix}
    1 & 0 & 0 & 0 \\
    0 & 1 & 0 & 0 \\
    0 & 0 & 1 & 0 \\
    0 & 0 & 0 & -1+a
  \end{smallmatrix}
   \right)
   \right\} }$
\end{flushleft}
\begin{flushleft}
$\:\:\gd{\frac{1}{a},1,1}= \Span_{\Com}{ \left\{ \left(
  \begin{smallmatrix}
    0 & 0 & 0 & 1 \\
    0 & 0 & 0 & 0 \\
    0 & 0 & 0 & 0 \\
    0 & 0 & 0 & 0
  \end{smallmatrix}
   \right),
   \left(
  \begin{smallmatrix}
    0 & 0 & 0 & 0 \\
    0 & 0 & 0 & 1 \\
    0 & 0 & 0 & 0 \\
    0 & 0 & 0 & 0
  \end{smallmatrix}
   \right),
   \left(
  \begin{smallmatrix}
    0 & 0 & 0 & 0 \\
    0 & 0 & 0 & 0 \\
    0 & 0 & 0 & 1 \\
    0 & 0 & 0 & 0
  \end{smallmatrix}
   \right),\left(
  \begin{smallmatrix}
    0 & 0 & 0 & 0 \\
    1 & 0 & 0 & 0 \\
    0 & 0 & 0 & 0 \\
    0 & 0 & 0 & 0
  \end{smallmatrix}
   \right),
   \left(
  \begin{smallmatrix}
    1 & 0 & 0 & 0 \\
    0 & 1 & 0 & 0 \\
    0 & 0 & 1 & 0 \\
    0 & 0 & 0 & -1+\frac{1}{a}
  \end{smallmatrix}
   \right)
   \right\} }$.
\end{flushleft}
\end{example}


\section{Associated Lie Algebras}
In \cite{HN2} we have investigated the problem when the subspace
$\gd{\alpha,\beta,\gamma}{\el}$ forms a Lie subalgebra of
$\gl\el$, i.~e. for which $\alpha,\beta,\gamma \in \Com$ is this
set closed with respect to the Lie product in $\gl\el$: $$ A,B \in
\gd{\alpha,\beta,\gamma}{\el} \Rightarrow AB-BA \in
    \gd{\alpha,\beta,\gamma}{\el}.$$
We found out the solutions which are now included in
Theorem~\ref{Lielist} as the cases 1. -- 5. and called them
\textbf{associated Lie algebras} of the Lie algebra $\el$. There
were two reasons for this investigation. Firstly, from
Corollaries~\ref{corgender} and~\ref{corgenderprun} we know that
for fixed $\alpha,\beta,\gamma,\alpha',\beta',\gamma'\in\Com$ the
map $d_{(\alpha,\beta,\gamma)}$, defined by
\begin{equation}\label{defdim}
 d_{(\alpha,\beta,\gamma)}{\el} = \dim
 \gd{\alpha,\beta,\gamma}{\el},
\end{equation}
as well as the map
$d_{(\alpha,\beta,\gamma)(\alpha',\beta',\gamma')}$, defined by
\begin{equation}\label{defdimprun}
 d_{(\alpha,\beta,\gamma)(\alpha',\beta',\gamma')}{\el} = \dim (\gd{\alpha,\beta,\gamma}{\el}\cap  \gd{\alpha',\beta',\gamma'}{\el}),
\end{equation}
constitute invariant characteristics of Lie algebras. Secondly,
the following Proposition allows us to consider not only the
dimensions but also the Lie structure of the associated Lie
algebras.
\begin{tvr}\label{invarLie} If two complex Lie algebras are isomorphic, $\el \cong
\wt\el$, then the associated Lie algebras and their intersections
are isomorphic as well, i.~e. it holds:
\begin{enumerate}
\item $\gd{1,1,1}\el \cong \gd{1,1,1}\wt\el$
\item $\gd{0,1,1}\el \cong \gd{0,1,1}\wt\el$
\item $\gd{1,1,0}\el \cong \gd{1,1,0}\wt\el$
\item $\gd{1,1,1}\el\cap\gd{0,1,1}\el \cong
\gd{1,1,1}\wt\el\cap\gd{0,1,1}\wt\el$.
\end{enumerate}
\end{tvr}
\begin{proof}
Since the map $\rho:\ \gl \el\map \gl \wt\el$, defined
by~(\ref{rho}), is a homomorphism, then its restriction on
subalgebras $\gd{1,1,1}\el\subset \gl \el$,
$\gd{0,1,1}\el\subset\gl \el, \dots$ is also a homomorphism. It
follows from~(\ref{invar}) and~(\ref{rhoabcprun}) that $\rho$ is
an isomorphism.
\end{proof}

Let us consider the following set of invariants of Lie algebras:
$$\inv =\{(d_k),(l_k),(c_k),\tau,d_{(1,1,1)}, d_{(0,1,1)},
d_{(1,1,0)},d_{(1,1,1)(0,1,1)} \}, $$ where the sequences $d_k,\,
l_k,\,c_k$ and the number $\tau$ were defined by relations
(\ref{seq1}) -- (\ref{seq3}) and~(\ref{formal}). We arrange the
values of the invariants in the set $\inv$ corresponding to some
coset of complex Lie algebras $[\el]$ into the following tuple:
\begin{equation}
\begin{array}{lll}
    \parbox[l][20pt][c]{0pt}{} \inv \el & = & (d_0(\el),d_1(\el),\dots )\ (l_0(\el),l_1(\el),\dots )\ (c_0(\el),c_1(\el),\dots )\quad
    \tau(\el) \\
    && [d_{(1,1,1)}{\el} ,\, d_{(0,1,1)}{\el} ,\, d_{(1,1,0)}{\el} ,\,d_{(1,1,1)(0,1,1)}\el ].
\end{array}
\end{equation}
In this notation the equations (\ref{fle1}), (\ref{fle2}) yield
$d_{(1,0,0)}{\el}=d_0(\el)(d_0(\el)-d_1(\el))$ and
$d_{(0,1,0)}{\el}=d_0(\el)c_0(\el)$. There is a natural question if
there exists a similar dependence among invariants in the set
$\inv$. Using the notion of independence defined by
(\ref{independent}), we give the answer in the following
proposition.

\begin{tvr}\label{indepass}
Let $\el$ be a complex Lie algebra. Then the invariant
\begin{enumerate}\item $d_{(1,1,1)}$ is independent on the set
$\:\inv \setminus \{d_{(1,1,1)}\}$ \item $d_{(0,1,1)}$ is
independent on the set $\:\inv \setminus\{ d_{(0,1,1)}\}$
\item $d_{(1,1,0)}$ is independent on the set
$\:\inv \setminus \{d_{(1,1,0)}\}$
 \item $d_{(1,1,1)(0,1,1)}$ is independent on the set
$\:\inv \setminus \{d_{(1,1,1)(0,1,1)}\}$.
\end{enumerate}
\end{tvr}
\begin{proof}
\begin{enumerate}\item $d_{(1,1,1)}$:
In order to satisfy the definition (\ref{independent}), we have to
construct two Lie algebras $\el$, $\wt{\el}$ whose invariants from
the set $\inv$ differ only in the values $d_{(1,1,1)}\el\neq
d_{(1,1,1)}\wt{\el} $. We may consider the following two
seven--dimensional Lie algebras $\el$, $\wt\el$: $$
\begin{array}{ll}
\el: \quad & [e_4,e_6]=e_1,\ [e_4,e_7]=e_2,\ [e_5,e_6]=e_2,\
[e_5,e_7]=e_3\\ \widetilde{\el}: \quad & [e_4,e_6]=e_1,\
[e_4,e_7]=e_2,\ [e_5,e_7]=e_3\\
\end{array}
$$ and find the corresponding values of the invariants from
$\inv$: $$
\begin{array}{llll}
\inv \el = \quad & (7,3,0)(7,3,0)(3,7) & 3 & [19,24,13,15] \\ \inv
\widetilde{\el}= & (7,3,0)(7,3,0)(3,7) & 3 & [20,24,13,15].
\\
\end{array}
$$ We observe that the only different values of $\inv\el$ and
$\inv\wt\el$ are $$19=d_{(1,1,1)}\el\neq d_{(1,1,1)}\wt{\el}=20.
$$ The proof of the remaining cases is analogous; we present the
pairs $\el, \wt\el$ and the values $\inv\el,\inv\wt\el$ for each
case.
\item $d_{(0,1,1)}$:
$$
\begin{array}{ll}
\el: \quad & [e_2,e_3]=e_4,\ [e_2,e_4]=e_5,\ [e_2,e_6]=-e_7,\
[e_2,e_8]=e_1,\\
            & [e_3,e_7]=e_1,\ [e_4,e_6]=e_1,\ [e_6,e_8]=e_5 \\
\widetilde{\el}: \quad & [e_2,e_3]=e_4,\ [e_2,e_4]=e_5,\
[e_2,e_8]=e_1,\ [e_3,e_6]=e_8,\\
            & [e_3,e_7]=e_1,\ [e_4,e_6]=e_1,\ [e_6,e_8]=e_5 \\
\end{array}
$$ $$
\begin{array}{llll}
\inv\el= \quad & (8,4,0)(8,4,2,0)(2,5,8) & 2 & [17,19,9,11] \\
\inv\widetilde{\el}= & (8,4,0)(8,4,2,0)(2,5,8) & 2 & [17,20,9,11]
\\
\end{array}
$$
\item $d_{(1,1,0)}$:
$$\begin{array}{ll} \el: \quad & [e_1,e_2]=e_4,\ [e_1,e_3]=e_5,\
[e_1,e_6]=e_1,\ [e_1,e_7]=e_3,\\
            & [e_2,e_6]=-e_2,\ [e_3,e_6]=e_3,\ [e_5,e_6]=2e_5 \\

\widetilde{\el}: \quad & [e_1,e_2]=e_4,\ [e_1,e_4]=e_5,\
[e_1,e_6]=e_1,\ [e_1,e_7]=e_3,\\
            & [e_2,e_6]=-2e_2,\ [e_3,e_6]=e_3,\ [e_4,e_6]=-e_4 \\
\end{array}
$$
$$
\begin{array}{llll}
\inv\el= \quad & (7,5,2,0)(7,5)(1) & 3 & [10,11,3,3] \\
\inv\widetilde{\el}= & (7,5,2,0)(7,5)(1) & 3 & [10,11,4,3] \\
\end{array}
$$

\item $d_{(1,1,1)(0,1,1)}$: $$
\begin{array}{ll}
\el: \quad & [e_1,e_3]=-e_3,\ [e_1,e_4]=e_4,\ [e_1,e_6]=2e_6,\
[e_1,e_7]=-e_7,\\
            & [e_1,e_8]=e_8,\ [e_3,e_6]=e_8,\ [e_4,e_5]=e_8,\ [e_4,e_7]=e_2 \\
\widetilde{\el}: \quad & [e_1,e_2]=-2e_2,\ [e_1,e_3]=-e_3,\
[e_1,e_4]=e_4,\ [e_1,e_6]=2e_6,\\
            & [e_1,e_8]=e_8,\ [e_2,e_6]=e_7,\ [e_3,e_6]=e_8,\ [e_4,e_5]=e_8 \\
\end{array}
$$ $$
\begin{array}{llll}
\inv\el= \quad & (8,6,2,0)(8,6)(1) & 2 & [12,13,4,3] \\
\inv\widetilde{\el}=  & (8,6,2,0)(8,6)(1) & 2 & [12,13,4,4] \\
\end{array}
$$
\end{enumerate}
\end{proof}

As the following example shows, analyzing the Lie structure of the
associated Lie algebras can be very useful.
\begin{example}\label{intersect}
 Let us present two 8--dimensional nilpotent Lie
 algebras as a list of their non--zero commutation relations in $\Z_3$--labeled
basis $( l_{01}, l_{02}, l_{10},  l_{20}, l_{11}, l_{22},
l_{12},l_{21}):$

$$
\begin{array}{ll}
\el_{17,9}: \quad & [l_{01},l_{10}]=l_{11},\
[l_{01},l_{20}]=l_{21},\ [l_{01},l_{11}]=l_{12},\
[l_{01},l_{22}]=l_{20},\\
            & [l_{02},l_{10}]=l_{12},\ [l_{10},l_{11}]=l_{21},\ [l_{20},l_{22}]=l_{12} \\
\\
\el_{17,12}: \quad & [l_{01},l_{10}]=l_{11},\
[l_{01},l_{20}]=l_{21},\ [l_{01},l_{22}]=l_{20},\
[l_{02},l_{10}]=l_{12},\\
            & [l_{02},l_{22}]=l_{21},\ [l_{10},l_{11}]=l_{21},\ [l_{20},l_{22}]=l_{12}. \\
\end{array}
$$   Because we have $$\inv \el_{17,9}= \inv\el_{17,12}= \quad
(8,4,0) (8,4,2,0)\,(2,5,8)\quad  2\quad [16,19,9,11], $$  the
algebras $\el_{17,9},\,\el_{17,9}$ cannot be distinguished using
the set of invariants $\inv \el$. We can advance to the higher
level by computing: $$\begin{array}{lllll} \inv
\gd{1,1,1}\el_{17,9} =\inv \gd{1,1,1}\el_{17,12} &
(16,15,6,0)(16,15)(0)& 6 & [16,15,1,6]
\\ \inv \gd{0,1,1}\el_{17,9} =\inv
\gd{0,1,1}\el_{17,12} & (19,15)(19,15)(0)& 5 & [32,0,1,0]\\ \inv
\gd{1,1,0}\el_{17,9} =\inv \gd{1,1,0}\el_{17,12} & (9,0)(9,0)(9)&
9& [81,81,81,81].\end{array}$$ We see that the algebras
$\el_{17,9},\,\el_{17,12}$ are still not decidedly
non--isomorphic. Surprisingly, the algebras
$\el_{17,9},\,\el_{17,12}$ have very different Lie structure of
the intersection of the operator algebras $\gd{1,1,1}{\el}\cap
\gd{0,1,1}{\el}$: $$\begin{array}{lllll}
\inv\gd{1,1,1}{\el_{17,9}}\cap \gd{0,1,1}{\el_{17,9}}&=
(11,6,0)\,(11,6,0)\,(6,11)& 7 & [43,67,31,31]\\
\inv\gd{1,1,1}{\el_{17,12}}\cap \gd{0,1,1}{\el_{17,12}}&=
  (11,4,0)\,(11,4,0)\,(7,11) & 7 & [57,78,50,50].
\end{array}$$
The conclusion $\el_{17,9}\ncong\el_{17,12}$ follows from
Proposition~\ref{invarLie}.
\end{example}
\subsection{Associated Lie Groups}
The fact that $\der\el$ is a Lie algebra of the group $\aut\el$
was stated in Theorem~\ref{groupI}. In order to answer the
question whether the associated Lie algebras $\gd{0,1,1}\el$ and
$\gd{1,1,0}\el$ and intersection $\gd{1,1,1}\el\cap\gd{0,1,1}\el$
are also Lie algebras of some linear groups, we firstly define the
sets:
\begin{align}\label{Aut011}
\ga{0,1,1}\el=&\set{f\in GL(\el)}{[f x,fy]=[x,y],\q \forall x,y
\in \el}\\ \ga{1,1,0}\el=&\set{f\in GL(\el)}{f[ x,y]=[f x,y],\q
\forall x,y \in \el}.
\end{align}
\begin{tvr}
Let $\el$ be a complex Lie algebra. The sets $\ga{0,1,1}\el$ and
$\ga{1,1,0}\el$ are subgroups of $GL(\el)$.
\end{tvr}
\begin{proof}
Both $\ga{0,1,1}\el$ and $\ga{1,1,0}\el$ contain the identity
operator and therefore are non--empty. Let $f,g\in \ga{0,1,1}\el$.
Substituting $x=g^{-1}z,y=g^{-1}w$ into the equation $[g x,g
y]=[x,y] $ one obtains $[g^{-1}z,g^{-1}w]=[z,w]$, i.~e.
$g^{-1}\in\ga{0,1,1}\el$. Then for all $x,y\in\el$
$$[fg^{-1}x,fg^{-1}y ]=[g^{-1}x,g^{-1}y ]=[x,y ],$$ i.~e.
$fg^{-1}\in\ga{0,1,1}\el.$

Similarly, let $f,g\in \ga{1,1,0}\el$. Substituting $x=g^{-1}z$
into the equation $g[ x, y]=[g x,y] $ one obtains
$g[g^{-1}z,y]=[z,y]$, or equivalently $[g^{-1}z,y]=g^{-1}[z,y]$,
i.~e. $g^{-1}\in\ga{1,1,0}\el$. Then for all $x,y\in\el$
$$fg^{-1}[x,y ]=f[g^{-1}x,y ]=[fg^{-1}x,y ].$$
\end{proof}
\begin{tvr}
Let $\el$ be a complex Lie algebra. Then $\ga{0,1,1}\el$ and
$\ga{1,1,0}\el$ are algebraic groups.
\end{tvr}
\begin{proof}
Let $c^k_{ij}$ denote the structural constants of $\el$,
$\dim\el=n$. Then $f=(f_{ij})\in\ga{0,1,1}\el$ if and only if
$$\sum_{p,q=1}^n f_{pi}f_{qj}c^k_{pq}=c^k_{ij},\
i,j,k\in\{1,\dots,n\} $$ and also $\wt f=(\wt
f_{ij})\in\ga{1,1,0}\el$ if and only if $$\sum_{p=1}^n \wt
f_{pi}c^k_{pj}-\wt f_{kp}c^p_{ij}=0,\ i,j,k\in\{1,\dots,n\}. $$
\end{proof}
\begin{lemma}
Let $\el$ be a complex Lie algebra and $A\in\gd{0,1,1}\el$. Then
for all $x,y\in \el$ and $m\in\N$
\begin{equation}\label{indukc}
  \sum_{i=0}^m \comb{m}{i}[A^i x,A^{m-i}y]=0.
\end{equation}
\end{lemma}
\begin{proof}
We prove the assertion by induction. For $m=1$ we obtain $[Ax,y]
+[x,Ay]=0$, i.~e. the definition of $\gd{0,1,1}\el$. Assuming that
(\ref{indukc}) is valid for $m$, we proceed to $m+1$ as follows:
\begin{align*}
 \sum_{i=0}^{m+1} \comb{m+1}{i}[A^i x,A^{m+1-i}y]=&[x,A^{m+1}y]+\sum_{i=0}^{m-1}
 \comb{m+1}{i+1}[A^{i+1}
 x,A^{m-i}y]+[A^{m+1}x,y]\\ =& [x,A^{m+1}y] + \sum_{i=0}^{m-1} \comb{m}{i+1}[A^{i+1} x,A^{m-i}y]+ \\
+&\sum_{i=0}^{m-1} \comb{m}{i}[A^{i+1} x,A^{m-i}y]+ [A^{m+1}x,y]\\
=&\sum_{i=0}^{m} \comb{m}{i}[A^i x,A^{m-i} (A y)]+\sum_{i=0}^{m}
\comb{m}{i}[A^i( Ax),A^{m-i}y].
\end{align*}
According to the assumption, each of the last two sums is equal to
zero.
\end{proof}
\begin{thm} Let $\el$ be a complex Lie algebra. Then it holds:
\begin{enumerate}
 \item $\gd{0,1,1}\el$ is the Lie algebra of $\ga{0,1,1}\el$,
 \item $\gd{1,1,0}\el$ is the Lie algebra of $\ga{1,1,0}\el$,
 \item $\gd{1,1,1}\el\cap\gd{0,1,1}\el$ is the Lie algebra of $\ga{0,1,1}\el\cap\aut\el$.
\end{enumerate}
\end{thm}
\begin{proof}

1. We prove the equality of the sets $$\gd{0,1,1}\el=\set{ A\in
\gl \el}{ \exp (\R A)\subset\ga{0,1,1}\el } $$ in two steps.

$\supset:$ Consider any $A\in \gl\el$ such that $\exp(tA)\in\
\ga{0,1,1}\el$, i.~e. for all $x,y\in\el$ and all $t \in \R$ it
holds:
\begin{equation}\label{grr1}
  [\exp
(tA)x,\exp(tA)y] =[x,y].
\end{equation}
Rewriting equation (\ref{grr1}) in the form

$$\left[\frac{1}{t}(\exp
(tA)-1)x,\exp(tA)y\right]+\left[x,\frac{1}{t}(\exp(tA)-1)y\right]=0
$$ and taking the limit $t\rightarrow 0$ we obtain $[Ax,y]+[x,A
y]=0$, i.~e. $A\in\gd{0,1,1}\el$.

$\subset:$ Let $A\in \gd{0,1,1}\el$. It is sufficient to prove
that for all $x,y\in\el$ the relation $$[\exp (A)x,\exp(A)y]
=[x,y]$$ holds. Then, since also $tA\in \gd{0,1,1}\el$ for all
$t\in \R$, the relation (\ref{grr1}) follows. Using
(\ref{indukc}), we calculate
\begin{align*}
\left[\sum_{i=0}^{\infty}\frac{1}{i!}A^i
x,\sum_{i=0}^{\infty}\frac{1}{i!}A^i y
\right]=&\sum_{m=0}^{\infty}\sum_{i=0}^{m}\left[\frac{1}{i!}A^i
x,\frac{1}{(m-i)!}A^{m-i} y\right]\\ =&
[x,y]+\sum_{m=1}^{\infty}\frac{1}{m!} \sum_{i=0}^m \comb{m}{i}[A^i
x,A^{m-i}y]=[x,y].
\end{align*}

2. $\supset:$ Consider any $A\in \gl\el$ such that $\exp(tA)\in\
\ga{1,1,0}\el$, i.~e. for all $x,y\in\el$ and all $t \in \R$ it
holds:
\begin{equation}\label{grr2}
  \exp(tA)[
x,y] =[\exp(tA)x,y].
\end{equation}
Rewriting equation (\ref{grr2}) in the form

$$\frac{1}{t}(\exp
(tA)-1)[x,y]=\left[\frac{1}{t}(\exp(tA)-1)x,y\right]$$ and taking
the limit $t\rightarrow 0$ we obtain $A[x,y]=[Ax, y]$, i.~e.
$A\in\gd{1,1,0}\el$.

$\subset:$ Let $A\in \gd{1,1,0}\el$. It is sufficient to prove
that for all $x,y\in\el$ $$\exp (A)[x,y] =[\exp (A)x,y].$$ Then,
since also $tA\in \gd{1,1,0}\el$ for all $t\in \R$, the relation
(\ref{grr2}) follows. Since for all $m\in\N$ is obviously the
relation $A^m[x,y]=[A^mx,y]$ satisfied, we have:

$$\sum_{m=0}^{\infty}\frac{1}{m!}A^m\left[ x, y
\right]=\sum_{m=0}^{\infty}\frac{1}{m!}\left[A^m x, y\right]
=\left[\sum_{m=0}^{\infty}\frac{1}{m!}A^m x, y\right]. $$

3. The statement follows directly from case 1.,
Theorem~\ref{groupI} and Proposition~\ref{prunikgrup}.
\end{proof}

\section{Invariant Function $\fa$ of Low--dimensional
Algebras}\label{invfoflow} Since among three--dimensional Lie
algebras infinite continuum appears already, it is clear that the
finite set $\inv\el$ of certain dimensions, though useful, can
never completely characterize Lie algebras of dimension higher or
equal to 3. On the contrary, it turns out that the invariant
function $\fa$ alone(!) forms a complete set of invariant(s) for
3--dimensional Lie algebras. We use the notation for
3--dimensional Lie algebras as in \cite{ccc} and begin with the
following lemma, which states that the invariant function~$\fa$
provides the classification of infinitely many Lie algebras in the
continuum $\g_{3,4}(a)$.
\begin{lemma}\label{examg34}
Let $\g_{3,4}(a)$ be a three--dimensional Lie algebra with
non--zero brackets $[e_1,e_3]=e_1,\ [e_2,e_3]=ae_2$, where
$a\in\Com$ and $a\neq 0,\pm 1$. If
$\fa\g_{3,4}(a)=\fa\g_{3,4}(a')$ then $\g_{3,4}(a)\cong
\g_{3,4}(a')$.
\end{lemma}
\begin{proof}
From the explicit form of matrices listed in Appendix \ref{APA} we
instantly see that the invariant function $\fa$ has the following
form
\begin{center}\vspace{-12pt}
\begin{tabular}[t]{|l||c|c|c|c|}
\hline \parbox[l][20pt][c]{0pt}{}   $\alpha$ & 1 & $a$
&$\frac{1}{a}$&
\\ \hline
\parbox[l][20pt][c]{0pt}{} $[\fa\g_{3,4}(a)](\alpha)$ & 4 & 4& 4&3   \\ \hline
\end{tabular}
\end{center}
A blank space in the table of the function $\fa$ denotes a general
complex number, different from all previously listed values, i.~e.
it holds: $[\fa\g_{3,4}(a)](\alpha)=3$, $\alpha\neq 1,a,1/a$. Let
us consider some $\g_{3,4}(a')$, $a'\neq 0,\pm 1$. If
$\fa\g_{3,4}(a)=\fa\g_{3,4}(a')$ then we need $a=a'$ or $a=1/a'$
-- otherwise $\fa\g_{3,4}(a)\neq\fa\g_{3,4}(a')$. The relation
\begin{equation}\label{izo3a}
 \g_{3,4}(a)\cong \g_{3,4}(1/a),
\end{equation}
can be verified directly.
\end{proof}

\begin{example}
The invariant function $\fa^0$ has for $\g_{3,1}$ the value
$\fa^0\g_{3,1}(\alpha) = 3$ and for the remaining algebras
$\g_{3,i},\,i=2,3, 4$ and $\slp (2,\Com)$ it holds:
$$\fa^0\g_{3,i}(\alpha) = \left\{
\begin{matrix}
    1, & \alpha = 1 \\
    0, & \alpha \neq 1 .
  \end{matrix} \right.
$$
\end{example}
\begin{thm}[Classification of three--dimensional complex Lie
algebras]\label{class3dim}$\,$\newline
  Two three--dimensional complex Lie
  algebras $\el$, $\wt\el$ are isomorphic if and only if
$\fa\el=\fa\wt\el$.
\end{thm}
\begin{proof}
$\Rightarrow:$ See Corollary~\ref{invarfunc}.

$\Leftarrow:$ According to Lemma~\ref{examg34} and observing the
tables in Appendix~\ref{IFLJL}, we conclude that all tables of the
invariant function~$\fa$ of non--isomorphic three-dimensional
complex Lie algebras are mutually different.
\end{proof}

\begin{example}\label{examg422}
In Example \ref{4dimder} explicit matrices of
$(\alpha,\beta,\gamma)$--derivations of four--dimensional
one--parametric Lie algebra $\g_{4,2}(a)$, $a\neq 0,\pm 1,-2$,
were presented. Thus we see that the invariant functions $\fa$ and
$\fa^0$ have the following form
\begin{center}\vspace{-12pt}
\begin{tabular}[t]{|l||c|c|c|c|}
\hline \parbox[l][20pt][c]{0pt}{}   $\alpha$ & 1 & $a$&
$\frac{1}{a}$& \\ \hline
\parbox[l][20pt][c]{0pt}{} $\fa\g_{4,2}(a)(\alpha)$ & 6 & 5 & 5 & 4  \\ \hline
\end{tabular}\qquad
\begin{tabular}[t]{|l||c|c|c|}
\hline  \parbox[l][20pt][c]{0pt}{}  $\alpha$ & 1  &
\\ \hline
\parbox[l][20pt][c]{0pt}{} $\fa^0\g_{4,2}(a)(\alpha)$ & 1 & 0\\
\hline
\end{tabular}\qquad
\end{center}
Let us consider some $\g_{4,2}(a')$, $a'\neq 0,\pm 1,-2$. If
$\fa\g_{4,2}(a)=\fa\g_{4,2}(a')$ then we need $a=a'$ or $a=1/a'$
-- otherwise $\fa\g_{4,2}(a)\neq\fa\g_{4,2}(a')$. However an
isomorphism relation similar to (\ref{izo3a}) does not hold. The
invariant function $\fa$ thus does not distinguish between pairs
$\g_{4,2}(a)$ and $\g_{4,2}(1/a)$ and consequently does not
provide a complete classification of Lie algebras in the continuum
$\g_{4,2}(a)$. In the next Chapter, we use additional invariant
functions which, together with $\fa$, allow the classification of
$\g_{4,2}(a)$ as well as all four--dimensional Lie algebras.
\end{example}


\chapter{Twisted Cocycles of Lie Algebras}\label{CHtwi}
\section{Two--dimensional Twisted Cocycles of the Adjoint Representation}
Recall that an $f$--cocycle of a Lie algebra $\el$ is defined as a
$q-$linear operator $z$ satisfying the relation $dz=0$, where the
map $d$ is defined by (\ref{koho1}). We generalize this definition
analogously to $(\alpha,\beta,\gamma)$--derivations. The content
of this chapter has not been previously published.

Let $\el$ be an arbitrary complex Lie algebra, $(V,f)$ its
representation and $\kappa=(\kappa_{ij})$ a $(q+1)\times(q+1)$
complex symmetric matrix. We call $c\in C^q(\el,V)$, $q\in\N$ for
which
\begin{align}\label{koho1t}
\nonumber 0 &= \sum_{i=1}^{q+1}
(-1)^{i+1}\kappa_{ii}f(x_i)c(x_1,\dots,\hat{x_i},\dots,x_{q+1})+\\
 &+ \sum_{\begin{smallmatrix}
  i,j=1 \\
  i<j
\end{smallmatrix}}^{q+1} (-1)^{i+j}\kappa_{ij} c([x_i,x_j],x_1,\dots,\hat{x_i},\dots,\hat{x_j},\dots,x_{q+1})
\end{align} a {\bf
$\kappa-$twisted cocycle} or shortly $\kappa-$cocycle of dimension
$q$ corresponding to $(V,f)$; the set of all $\kappa-$cocycles of
dimension $q$ is denoted by $Z^{q} (\el,f,\kappa)$. It is clear
that $Z^{q} (\el,f,\kappa)$ is a linear subspace of $C^q(\el,V)$.
Similarly to (\ref{vla1}) we observe that for any $\varepsilon \in
\Com \backslash\{0\}$ and $q\in \N$ it holds:
\begin{equation}\label{vlakoho1}
Z^{q} (\el,f,\kappa)=Z^{q} (\el,f,\varepsilon\kappa).
\end{equation}

In Chapter \ref{CHgen} we analyzed in detail
$(\alpha,\beta,\gamma)$--derivations; they are now included in the
definition of twisted cocycles. Considering the adjoint
representation, we immediately have the generalization of
Proposition \ref{kohoder}:
\begin{equation}\label{SPEC}
Z^{1} \left(\el,\ad_\el,\left(\begin{smallmatrix}
  \beta &  \alpha \\
   \alpha & \gamma
\end{smallmatrix}\right) \right)=\gd{\alpha,\beta,\gamma}\el.
\end{equation}
The most logical next step is to set $q=2$ and investigate in
detail the space $Z^{2} (\el,\ad_\el,\kappa)$ -- we devote this
chapter to this goal. For this purpose it may be more convenient
to use different notation, analogous to that of derivations,
defined by
\begin{equation}\label{kohonot}
\gc{\alpha_1,\alpha_2,\alpha_3,\beta_1,\beta_2,\beta_3}\el=Z^{2}
\left(\el,\ad_\el,\left(\begin{smallmatrix}
  \beta_1 &   \alpha_2 &  \alpha_3 \\
   \alpha_2 & \beta_3 &  \alpha_1 \\
   \alpha_3 &  \alpha_1 & \beta_2
\end{smallmatrix}\right) \right),
\end{equation}
i.~e. in the space
$\gc{\alpha_1,\alpha_2,\alpha_3,\beta_1,\beta_2,\beta_3}\el$ are
such $B \in C^2(\el,\el)$ which for all $x,y,z \in\el$ satisfy
\begin{align}\label{defb}
0&=\alpha_1 B(x, [y,z])+\alpha_2 B(z, [x,y])+\alpha_3 B(y, [z,x])
\nonumber\\ &+ \beta_1 [x, B ( y,z)]+ \beta_2 [z, B (
x,y)]+\beta_3 [y, B ( z,x)].
\end{align}
Six permutations of the variables $x,y,z \in\el$ in the defining
equation (\ref{defb}) give
\begin{lemma}
 Let $\el$ be a complex Lie algebra. Then for any $\alpha_1, \alpha_2,\alpha_3,\beta_1,\beta_2,\beta_3 \in \Com$
are all the following six spaces equal:
\begin{enumerate}
\item $\gc{\alpha_1,\alpha_2,\alpha_3,\beta_1,\beta_2,\beta_3}\el$
\item $\gc{\alpha_3,\alpha_1,\alpha_2,\beta_3,\beta_1,\beta_2}\el$
\item $\gc{\alpha_2,\alpha_3,\alpha_1,\beta_2,\beta_3,\beta_1}\el$
\item $\gc{\alpha_1,\alpha_3,\alpha_2,\beta_1,\beta_3,\beta_2}\el$
\item $\gc{\alpha_2,\alpha_1,\alpha_3,\beta_2,\beta_1,\beta_3}\el$
\item $\gc{\alpha_3,\alpha_2,\alpha_1,\beta_3,\beta_2,\beta_1}\el$
\end{enumerate}
\end{lemma}

\begin{lemma}\label{vlakoho} Let $\el$ be a complex Lie algebra. Then for any $\alpha_1, \alpha_2,\alpha_3,\beta_1,\beta_2,\beta_3 \in \Com$
is the space $
\gc{\alpha_1,\alpha_2,\alpha_3,\beta_1,\beta_2,\beta_3}\el$ equal
to all of the following:
\begin{enumerate}
\item$\gc{\alpha_1+\alpha_3,\alpha_2+\alpha_1,\alpha_3+\alpha_2,\beta_1+\beta_3,\beta_2+\beta_1,\beta_3+\beta_2}\el\,
\cap\,\gc{\alpha_1-\alpha_3,\alpha_2-\alpha_1,\alpha_3-\alpha_2,\beta_1-\beta_3,\beta_2-\beta_1,\beta_3-\beta_2}\el$
\item$\gc{0,\alpha_2-\alpha_3,\alpha_3-\alpha_2,0,\beta_2-\beta_3,\beta_3-\beta_2}\el\,
\cap\,\gc{2\alpha_1,\alpha_2+\alpha_3,\alpha_3+\alpha_2,2\beta_1,\beta_2+\beta_3,\beta_3+\beta_2}\el$
\item$\gc{0,\alpha_1-\alpha_2,\alpha_2-\alpha_1,0,\beta_1-\beta_2,\beta_2-\beta_1}\el\,
\cap\,\gc{2\alpha_3,\alpha_1+\alpha_2,\alpha_2+\alpha_1,2\beta_3,\beta_1+\beta_2,\beta_2+\beta_1}\el$
\item$\gc{0,\alpha_3-\alpha_1,\alpha_1-\alpha_3,0,\beta_3-\beta_1,\beta_1-\beta_3}\el\,
\cap\,\gc{2\alpha_2,\alpha_3+\alpha_1,\alpha_1+\alpha_3,2\beta_2,\beta_3+\beta_1,\beta_1+\beta_3}\el$

\end{enumerate}
\end{lemma}

\begin{proof} Suppose $\alpha_1,\alpha_2,\alpha_3,\beta_1,\beta_2,\beta_3 \in \Com$ are given. We demonstrate the proof on the case
1. -- the proof of the other cases is analogous. Let $B \in
\gc{\alpha_1,\alpha_2,\alpha_3,\beta_1,\beta_2,\beta_3}\el$; then
for arbitrary $x,y,z\in \el$ we have
    \begin{eqnarray}
    0&=\alpha_1 B(x, [y,z])+\alpha_2 B(z, [x,y])+\alpha_3 B(y, [z,x])
\nonumber\\ &+ \beta_1 [x, B ( y,z)]+ \beta_2 [z, B (
x,y)]+\beta_3 [y, B ( z,x)]\label{trik3} \\
    0&=\alpha_1 B(z, [x,y])+\alpha_2 B(y, [z,x])+\alpha_3 B(x, [y,z])
\nonumber\\ &+ \beta_1 [z, B ( x,y)]+ \beta_2 [y, B (
z,x)]+\beta_3 [x, B ( y,z)]. \label{trik4}
    \end{eqnarray}
    By adding and subtracting equations (\ref{trik3}) and (\ref{trik4}) we obtain
     \begin{eqnarray}
    0&=(\alpha_1+\alpha_3) B(x, [y,z])+(\alpha_2+\alpha_1) B(z, [x,y])+(\alpha_3 +\alpha_2 )B(y, [z,x])
\nonumber\\ &+ (\beta_1+\beta_3) [x, B ( y,z)]+ (\beta_2+\beta_1)
[z, B ( x,y)]+(\beta_3+\beta_2) [y, B ( z,x)]\label{trik5} \\
    0&=(\alpha_1-\alpha_3) B(x, [y,z])+(\alpha_2-\alpha_1) B(z, [x,y])+(\alpha_3 -\alpha_2 )B(y, [z,x])
\nonumber\\ &+ (\beta_1-\beta_3) [x, B ( y,z)]+ (\beta_2-\beta_1)
[z, B ( x,y)]+(\beta_3-\beta_2) [y, B ( z,x)]\label{trik6}
    \end{eqnarray}
    and thus $
\gc{\alpha_1,\alpha_2,\alpha_3,\beta_1,\beta_2,\beta_3}\el$ is the
subset of the intersection 1. Similarly, starting with equations
(\ref{trik5}), (\ref{trik6}) we obtain equations (\ref{trik3}),
(\ref{trik4}) and the remaining inclusion is proven.
\end{proof}
Further, we proceed to formulate an analogous theorem to
Theorem~\ref{klass}; the six original parameters are reduced to
four.
\begin{thm}\label{klass2} Let $\el$ be a Lie algebra. Then for any $\alpha_1,\alpha_2,\alpha_3,\beta_1,\beta_2,\beta_3
 \in\Com$ there exist $\alpha,\beta,\gamma,\delta\in \Com$ such that the subspace
$\gc{\alpha_1,\alpha_2,\alpha_3,\beta_1,\beta_2,\beta_3}\el
\subset C^2(\el,\el)$ is equal to some of the following sixteen
subspaces:
\begin{enumerate}
 \item $\gc{\alpha,0,0,\beta,0,0}{\el}$;
  $\gc{\alpha,0,0,\beta,1,-1}{\el}$;
  $\gc{\alpha,1,-1,\beta,0,0}{\el}$;
  $\gc{\alpha,\beta,-\beta,\gamma,1,-1}{\el}$
 \item $\gc{\alpha,0,0,\beta,1,0}{\el}$;
  $\gc{\alpha,0,0,\beta,1,1}{\el}$;
  $\gc{\alpha,\beta,-\beta,\gamma,1,0}{\el}$;
  $\gc{\alpha,1,-1,\beta,1,1}{\el}$
 \item $\gc{\alpha,1,0,\beta,0,0}{\el}$;
  $\gc{\alpha,1,1,\beta,0,0}{\el}$;
  $\gc{\alpha,1,0,\beta,\gamma,-\gamma}{\el}$;
  $\gc{\alpha,1,1,\beta,1,-1}{\el}$
 \item $\gc{\alpha,\beta,\gamma,\delta,1,0}{\el}$;
  $\gc{\alpha,\beta+1,\beta-1,\gamma,1,1}{\el}$;
  $\gc{\alpha,1,1,\beta,\gamma+1,\gamma-1}{\el}$;
  $\gc{\alpha,\beta,\beta,\gamma,1,1}{\el}$
\end{enumerate}
\end{thm}
\begin{proof}
\begin{enumerate}
    \item Suppose $\alpha_2 + \alpha_3 = 0 $ and $\beta_2 + \beta_3 = 0 $. Then the following four
    cases are possible:
    \begin{enumerate}
    \item $\alpha_2=-\alpha_3=0$ and $\beta_2=-\beta_3=0$. In this
    case we have
    $$\gc{\alpha_1,\alpha_2,\alpha_3,\beta_1,\beta_2,\beta_3}\el  =
    \gc{\alpha_1,0,0,\beta_1,0,0}\el.$$
    \item $\alpha_2=-\alpha_3=0$ and $\beta_2=-\beta_3\neq 0$. In this
    case we have:
\begin{align*}
\gc{\alpha_1,\alpha_2,\alpha_3,\beta_1,\beta_2,\beta_3}\el=&\gc{\alpha_1,0,0,\beta_1,\beta_2,\beta_3}\el\\
=&
\gc{0,0,0,0,\beta_2-\beta_3,\beta_3-\beta_2}\el\cap\gc{2\alpha_1,0,0,2\beta_1,0,0}\el\\
=&\gc{0,0,0,0,2,-2}\el\cap\gc{2\alpha_1,0,0,2\beta_1,0,0}\el\\
=&\gc{\alpha_1,0,0,\beta_1,1,-1}\el
\end{align*}
\item $\alpha_2=-\alpha_3\neq0$ and $\beta_2=-\beta_3= 0$. In this
    case we have:
\begin{align*}
\gc{\alpha_1,\alpha_2,\alpha_3,\beta_1,\beta_2,\beta_3}\el=&\gc{\alpha_1,\alpha_2,\alpha_3,\beta_1,0,0}\el\\
=&
\gc{0,\alpha_2-\alpha_3,\alpha_3-\alpha_2,0,0,0}\el\cap\gc{2\alpha_1,0,0,2\beta_1,0,0}\el\\
=&\gc{0,2,-2,0,0,0}\el\cap\gc{2\alpha_1,0,0,2\beta_1,0,0}\el\\
=&\gc{\alpha_1,1,-1,\beta_1,0,0}\el
\end{align*}
\item $\alpha_2=-\alpha_3\neq0$ and $\beta_2=-\beta_3\neq 0$. In this
    case we have:
\begin{align*}
\gc{\alpha_1,\alpha_2,\alpha_3,\beta_1,\beta_2,\beta_3}\el=&
\gc{0,\alpha_2-\alpha_3,\alpha_3-\alpha_2,0,\beta_2-\beta_3,\beta_3-\beta_2}\el\cap\gc{2\alpha_1,0,0,2\beta_1,0,0}\el\\
=&\gc{0,2\frac{\alpha_2-\alpha_3}{\beta_2-\beta_3},-2\frac{\alpha_2-\alpha_3}{\beta_2-\beta_3},0,2,-2}\el\cap\gc{2\alpha_1,0,0,2\beta_1,0,0}\el\\
=&\gc{\alpha_1,\frac{\alpha_2-\alpha_3}{\beta_2-\beta_3},-\frac{\alpha_2-\alpha_3}{\beta_2-\beta_3},\beta_1,1,-1}\el
\end{align*}
    \end{enumerate}
    \item Suppose $\alpha_2 + \alpha_3 = 0 $ and $\beta_2 + \beta_3 \neq 0 $. Then the following four
    cases are possible:
    \begin{enumerate}
        \item $\alpha_2=-\alpha_3=0$ and $\beta_2-\beta_3\neq 0$. In this
    case we have:
\begin{align*}
\gc{\alpha_1,\alpha_2,\alpha_3,\beta_1,\beta_2,\beta_3}\el=&\gc{\alpha_1,0,0,\beta_1,\beta_2,\beta_3}\el\\
=&
\gc{0,0,0,0,\beta_2-\beta_3,\beta_3-\beta_2}\el\cap\gc{2\alpha_1,0,0,2\beta_1,\beta_2+\beta_3,\beta_3+\beta_2}\el\\
=&\gc{0,0,0,0,1,-1}\el\cap\gc{2\frac{\alpha_1}{\beta_2+\beta_3},0,0,2\frac{\beta_1}{\beta_2+\beta_3},1,1}\el\\
=&\gc{\frac{\alpha_1}{\beta_2+\beta_3},0,0,\frac{\beta_1}{\beta_2+\beta_3},1,0}\el
\end{align*}
        \item $\alpha_2=-\alpha_3=0$ and $\beta_2=\beta_3\neq 0$. In this
    case we have
    $$\gc{\alpha_1,\alpha_2,\alpha_3,\beta_1,\beta_2,\beta_3}\el  =
    \gc{\alpha_1,0,0,\beta_1,\beta_2,\beta_2}\el=\gc{\frac{\alpha_1}{\beta_2},0,0,\frac{\beta_1}{\beta_2},1,1}\el.$$
\item $\alpha_2=-\alpha_3\neq0$ and $\beta_2-\beta_3\neq 0$. In this
    case we have:
\begin{align*}
\gc{\alpha_1,\alpha_2,\alpha_3,\beta_1,\beta_2,\beta_3}\el=&\gc{\alpha_1,\alpha_2,-\alpha_2,\beta_1,\beta_2,\beta_3}\el\\
=&
\gc{0,2\alpha_2,-2\alpha_2,0,\beta_2-\beta_3,\beta_3-\beta_2}\el\cap\gc{2\alpha_1,0,0,2\beta_1,\beta_2+\beta_3,\beta_3+\beta_2}\el\\
=&\gc{0,2\frac{\alpha_2}{\beta_2-\beta_3},-2\frac{\alpha_2}{\beta_2-\beta_3},0,1,-1}\el\cap\gc{2\frac{\alpha_1}{\beta_2+\beta_3},0,0,2\frac{\beta_1}{\beta_2+\beta_3},1,1}\el\\
=&\gc{\frac{\alpha_1}{\beta_2+\beta_3},\frac{\alpha_2}{\beta_2-\beta_3},-\frac{\alpha_2}{\beta_2-\beta_3},\frac{\beta_1}{\beta_2+\beta_3},1,0}\el
\end{align*}
\item $\alpha_2=-\alpha_3\neq0$ and $\beta_2=\beta_3\neq 0$. In this
    case we have:
\begin{align*}
\gc{\alpha_1,\alpha_2,\alpha_3,\beta_1,\beta_2,\beta_3}\el=&\gc{\alpha_1,\alpha_2,-\alpha_2,\beta_1,\beta_2,\beta_2}\el\\
=&
\gc{0,2\alpha_2,-2\alpha_2,0,0,0}\el\cap\gc{2\alpha_1,0,0,2\beta_1,2\beta_2,2\beta_2}\el\\
=&\gc{0,2,-2,0,0,0}\el\cap\gc{2\frac{\alpha_1}{\beta_2},0,0,2\frac{\beta_1}{\beta_2},2,2}\el\\
=&\gc{\frac{\alpha_1}{\beta_2},1,-1,\frac{\beta_1}{\beta_2},1,1}\el
\end{align*}
    \end{enumerate}
\item Suppose $\alpha_2 + \alpha_3 \neq 0 $ and $\beta_2 + \beta_3 = 0 $.
This case is a complete analogue of the previous one.
\item  Suppose $\alpha_2 + \alpha_3 \neq 0 $ and $\beta_2 + \beta_3 \neq 0 $. Then the following four
    cases are possible:
\begin{enumerate}
\item $\alpha_2-\alpha_3\neq0$ and $\beta_2-\beta_3\neq 0$. The
space $\gc{\alpha_1,\alpha_2,\alpha_3,\beta_1,\beta_2,\beta_3}\el$
is in this case equal to:
\begin{align*}
&
\gc{0,\alpha_2-\alpha_3,\alpha_3-\alpha_2,0,\beta_2-\beta_3,\beta_3-\beta_2}\el\cap\gc{2\alpha_1,\alpha_2+\alpha_3,\alpha_3+\alpha_2,2\beta_1,\beta_2+\beta_3,\beta_3+\beta_2}\el\\
=&\gc{0,\frac{\alpha_2-\alpha_3}{\beta_2-\beta_3},-\frac{\alpha_2-\alpha_3}{\beta_2-\beta_3},0,1,-1}\el\cap\gc{2\frac{\alpha_1}{\beta_2+\beta_3},\frac{\alpha_2+\alpha_3}{\beta_2+\beta_3},\frac{\alpha_2+\alpha_3}{\beta_2+\beta_3},2\frac{\beta_1}{\beta_2+\beta_3},1,1}\el\\
=&\gc{\frac{\alpha_1}{\beta_2+\beta_3},\frac{\alpha_2+\alpha_3}{2(\beta_2+\beta_3)}+\frac{\alpha_2-\alpha_3}{2(\beta_2-\beta_3)},\frac{\alpha_2+\alpha_3}{2(\beta_2+\beta_3)}-\frac{\alpha_2-\alpha_3}{2(\beta_2-\beta_3)},\frac{\beta_1}{\beta_2+\beta_3},1,0}\el
\end{align*}
\item $\alpha_2-\alpha_3\neq0$ and $\beta_2=\beta_3\neq 0$. The
space $\gc{\alpha_1,\alpha_2,\alpha_3,\beta_1,\beta_2,\beta_3}\el$
is in this case equal to:
\begin{align*}
&
\gc{0,\alpha_2-\alpha_3,\alpha_3-\alpha_2,0,0,0}\el\cap\gc{2\alpha_1,\alpha_2+\alpha_3,\alpha_3+\alpha_2,2\beta_1,2\beta_2,2\beta_2}\el\\
=&\gc{0,2,-2,0,0,0}\el\cap\gc{2\frac{\alpha_1}{\beta_2},\frac{\alpha_2+\alpha_3}{\beta_2},\frac{\alpha_2+\alpha_3}{\beta_2},2\frac{\beta_1}{\beta_2},2,2}\el\\
=&\gc{\frac{\alpha_1}{\beta_2},\frac{\alpha_2+\alpha_3}{2\beta_2}+1,\frac{\alpha_2+\alpha_3}{2\beta_2}-1,\frac{\beta_1}{\beta_2},1,1}\el
\end{align*}
\item $\alpha_2=\alpha_3\neq0$ and $\beta_2-\beta_3\neq 0$. The
space $\gc{\alpha_1,\alpha_2,\alpha_3,\beta_1,\beta_2,\beta_3}\el$
is in this case equal to:
\begin{align*}
&
\gc{0,0,0,0,\beta_2-\beta_3,\beta_3-\beta_2}\el\cap\gc{2\alpha_1,2\alpha_2,2\alpha_2,2\beta_1,\beta_2+\beta_3,\beta_3+\beta_2}\el\\
=&\gc{0,0,0,0,2,-2}\el\cap\gc{2\frac{\alpha_1}{\alpha_2},2,2,2\frac{\beta_1}{\alpha_2},\frac{\beta_2+\beta_3}{\alpha_2},\frac{\beta_2+\beta_3}{\alpha_2}}\el\\
=&\gc{\frac{\alpha_1}{\alpha_2},1,1,\frac{\beta_1}{\alpha_2},\frac{\beta_2+\beta_3}{2\alpha_2}+1,\frac{\beta_2+\beta_3}{2\alpha_2}-1}\el
\end{align*}
\item $\alpha_2=\alpha_3\neq 0$ and $\beta_2=\beta_3\neq 0$. In this
    case we have
    $$\gc{\alpha_1,\alpha_2,\alpha_3,\beta_1,\beta_2,\beta_3}\el  =
    \gc{\frac{\alpha_1}{\beta_2},\frac{\alpha_2}{\beta_2},\frac{\alpha_2}{\beta_2},\frac{\beta_1}{\beta_2},1,1}\el.$$
\end{enumerate}
\end{enumerate}
\end{proof}
\subsection{Twisted Cocycles of Low--dimensional Lie Algebras}
It may be more convenient, sometimes, to use different
distribution of the cocycle spaces
$\gc{\alpha_1,\alpha_2,\alpha_3,\beta_1,\beta_2,\beta_3}\el$ than
in Theorem~\ref{klass2}. Henceforth, we investigate mainly the
space $\gc{1,1,1,\la,\la,\la}\el$ which for $\la\neq 0$ fits in
the class $ \gc{\alpha,\beta,\beta,\gamma,1,1}\el$, with
$\alpha=\beta=1/\la,\,\gamma=1$. For $\la=0$, the space
$\gc{1,1,1,0,0,0}\el$ is a special case of the space
$\gc{\al,1,1,\beta,0,0}\el$, with $\al=1,\beta=0$. We also put
$\al=0,\,\beta=1,\,\gamma=\la$ into
$\gc{\alpha,\beta,\beta,\gamma,1,1}\el$ and investigate the space
$\gc{0,1,1,\la,1,1}\el$.
\begin{example}\label{examg342}
In Lemma~\ref{examg34} the Lie algebra $\g_{3,4}(a)$, $a\neq 0,\pm
1$ was introduced. We present the explicit form of the spaces
$\gc{1,1,1,\la,\la,\la}\g_{3,4}(a)$ and
$\gc{0,1,1,\la,1,1}\g_{3,4}(a)$. Let us define six
$\g_{3,4}(a)$--cochains
$t_1,\dots,t_4,t^\la_5,t^\la_6,\,\la\in\Com$ by listing their
non--zero commutation relations:
\begin{enumerate}[$t_1:$]
\item $t_1(e_1,e_3)=e_1$
\item $t_2(e_1,e_3)=e_2$
\item $t_3(e_2,e_3)=e_1$
\item $t_4(e_2,e_3)=e_2$
\item[$t^\la_5:$] $t^\la_5(e_1,e_2)=\la a e_2$, $t_5^\la(e_1,e_3)=(-\la a+a+1 )e_3$
\item[$t^\la_6:$] $t^\la_6(e_1,e_2)=-\la e_1$, $t_6^\la(e_2,e_3)=(-\la+ a +1) e_3$.
\end{enumerate}
Then we have:
\begin{itemize}
\item[$\cdot$]
$\gc{1,1,1,\la,\la,\la}\g_{3,4}(a)=
\Span_{\Com}\{t_1,\dots,t_4,t^\la_5,t^\la_6\}$
\item[$\cdot$]
$\gc{0,1,1,\la,1,1}\g_{3,4}(a)= \{0\}_{\la\neq
2,1+a,1+\frac{1}{a}}$
\item[$\cdot$]
$\gc{0,1,1,2,1,1}\g_{3,4}(a)= \Span_{\Com}\{t_1,t_4\}$
\item[$\cdot$]
$\gc{0,1,1,1+a,1,1}\g_{3,4}(a)= \Span_{\Com}\{t_3\}$
\item[$\cdot$]
$\gc{0,1,1,1+\frac{1}{a},1,1}\g_{3,4}(a)= \Span_{\Com}\{t_2\}$.
\end{itemize}
\end{example}

\begin{example}\label{excoc}
In Examples \ref{4dimder} and \ref{examg422} the explicit matrices
of $(\alpha,\beta,\gamma)$--derivations and the corresponding
invariant functions of $\g_{4,2}(a)$, $a\neq 0,\pm 1,-2$ were
presented. In this example, we calculate the explicit form of the
spaces $\gc{1,1,1,\la,\la,\la}\g_{4,2}(a)$ and
$\gc{0,1,1,\la,1,1}\g_{4,2}(a)$. We first define the following
nine $\g_{4,2}(a)$--cochains $b_1,\dots,b_9$ by listing their
non--zero commutation relations:
\begin{enumerate}[$b_1:$]
\item $b_1(e_1,e_4)=e_1$
\item $b_2(e_1,e_4)=e_2$
\item $b_3(e_1,e_4)=e_3$
\item $b_4(e_2,e_4)=e_1$
\item $b_5(e_2,e_4)=e_2$
\item $b_6(e_2,e_4)=e_3$
\item $b_7(e_3,e_4)=e_1$
\item $b_8(e_3,e_4)=e_2$
\item $b_9(e_3,e_4)=e_3$,
\end{enumerate}
and second three more $\g_{4,2}(a)$--cochains $b'_1,\dots,b'_3$:
\begin{enumerate}[$b'_1:$]
\item $b'_1(e_1,e_3)=2a e_1$, $b'_1(e_2,e_3)=(-1+a) e_3$, $b'_1(e_3,e_4)=(-1+a) e_4 $
\item $b'_2(e_1,e_2)=(1+a) e_2$, $b'_2(e_1,e_3)= e_3$
\item $b'_3(e_2,e_3)=e_1$.
\end{enumerate}
Moreover, for $\la\in\Com$, we define $\la$--parametric sets of
$\g_{4,2}(a)$--cochains:
\begin{enumerate}[$f^{\la}_1:$]
\item $f^\la_1(e_1,e_2)=\la a(\la a -a -1)(\la-2)^2 e_1,\ f^\la_1(e_1,e_3)=\la a(\la-2)^2 e_1$,\newline $f^\la_1(e_2,e_3)=2\la(\la a -a
-1)^2 e_2$, \newline  $f^\la_1(e_2,e_3)=\la (\la-2)(\la a -a
-1)^2e_3$,\newline$ f^\la_1(e_2,e_4)=(\la-2)^2 (\la a -a -1)^2 e_3
$
\item $f^\la_2(e_1,e_2)=\la (-\la +a +1)e_2,\ f^\la_2(e_1,e_3)=\la a e_2,
$\newline $f^\la_2(e_1,e_3)=\la (-\la +a +1)e_3,\
f^\la_2(e_1,e_4)=(-\la +a +1)^2e_4$
\item $f^\la_3(e_1,e_2)=-\la a(\la-2)^2 e_1,\ f^\la_3(e_2,e_3)=\la (2\la a
-2a-\la)e_2$,\newline $f^\la_3(e_2,e_3)=\la (\la-2)(-\la
a+a+1)e_3,\ f^\la_3(e_2,e_4)=\la (\la-2)^2(-\la
a+a+1)e_4$,\newline$ f^\la_3(e_2,e_3)= (\la-2)^2e_4
$
\end{enumerate}
Then it holds:
\begin{itemize}
\item[$\cdot$]
$\gc{1,1,1,\la,\la,\la}\g_{4,2}(a)=
\Span_{\Com}\{b_1,\dots,b_9,f^\la_1,f^\la_2,f^\la_3\}_{\la\neq
2,a+1,2/a }$
\item[$\cdot$]
$\gc{1,1,1,2,2,2}\g_{4,2}(a)=
\Span_{\Com}\{b_1,\dots,b_9,f^2_1,f^2_2,b'_1\}$
\item[$\cdot$]
$\gc{1,1,1,a+1,a+1,a+1}\g_{4,2}(a)=
\Span_{\Com}\{b_1,\dots,b_9,f^{a+1}_1,f^{a+1}_2,f^{a+1}_3,b'_2\}$
\item[$\cdot$]
$\gc{1,1,1,\frac{2}{a},\frac{2}{a},\frac{2}{a}}\g_{4,2}(a)=
\Span_{\Com}\{b_1,\dots,b_9,f^{2/a}_1,f^{2/a}_2,f^{2/a}_3,b'_3\}$.
\item[$\cdot$]
$\gc{0,1,1,\la,1,1}\g_{4,2}(a)= \{0\}_{\la\neq
2,1+a,1+\frac{1}{a}}$
\item[$\cdot$]
$\gc{0,1,1,2,1,1}\g_{4,2}(a)= \Span_{\Com}\{b_1,b_8,b_5+b_9\}$
\item[$\cdot$]
$\gc{0,1,1,1+a,1,1}\g_{4,2}(a)= \Span_{\Com}\{b_2\}$
\item[$\cdot$]
$\gc{0,1,1,1+\frac{1}{a},1,1}\g_{4,2}(a)= \Span_{\Com}\{b_7\}$

\end{itemize}
\end{example}
\section{Invariant Functions}
We directly generalize Theorem~\ref{tvr1}:
\begin{thm}\label{tvr1koho}
Let $g:\el \rightarrow \widetilde{\el}$ be an isomorphism of Lie
algebras $\el$ and $\widetilde{\el}$. Then the mapping $\rho
:C^q(\el,\el)\rightarrow C^q(\wt\el,\wt\el),\,q\in \N$ defined for
all $c\in C^q(\el,\el)$ and all $x_1,\dots,x_q\in \wt\el$ by $$
(\rho c)(x_1,\dots,x_q) = g c(g^{-1}x_1,\dots,g^{-1}x_q) $$ is an
isomorphism of vector spaces $C^q(\el,\el)$ and $
C^q(\wt\el,\wt\el)$. For any complex symmetric $(q+1)$--square
matrix $\kappa$
 $$\rho(Z^{q} (\el,\ad_\el,\kappa)) =
Z^{q} (\wt\el,\ad_{\wt\el},\kappa).$$
\end{thm}
\begin{proof}
Suppose we have $g:\el \rightarrow \wt\el$ such that for all
$x,y\in \widetilde{\el}$ $$[x,y]_{ \widetilde{\el}}=g
[g^{-1}x,g^{-1}y ]_{\el} $$ holds. It is clear that the map $\rho
:C^q(\el,\el)\rightarrow C^q(\wt\el,\wt\el),\,q\in \N$ is linear
and bijective, i.~e. it is an isomorphism of these vector spaces.
By putting $f=\ad_\el$ and rewriting definition (\ref{koho1t})  we
have for $c\in Z^{q} (\el,\ad_\el,\kappa)$ and all
$x_1,\dots,x_q\in \wt\el$
\begin{align}
\nonumber 0 &= \sum_{i=1}^{q+1}
(-1)^{i+1}\kappa_{ii}[g^{-1}x_i,c(g^{-1}x_1,\dots,\widehat{g^{-1}x_i},\dots,g^{-1}x_{q+1})]_{\el}+\\
\nonumber &+ \sum_{\begin{smallmatrix}
  i,j=1 \\
  i<j
\end{smallmatrix}}^{q+1} (-1)^{i+j}\kappa_{ij} c([g^{-1}x_i,g^{-1}x_j]_\el,g^{-1}x_1,\dots,\widehat{g^{-1}x_i},\dots,\widehat{g^{-1}x_j},\dots,g^{-1}x_{q+1})
\end{align}
Applying the mapping $g$ on this equation and taking into account
that $\kappa_{ij}\in \Com$ one has
\begin{align}\label{koho1adwt}
\nonumber 0 &= \sum_{i=1}^{q+1} (-1)^{i+1}\kappa_{ii}[x_i,(\rho
c)(x_1,\dots,\hat{x_i},\dots,x_{q+1})]_{\wt\el}+\\ \nonumber &+
\sum_{\begin{smallmatrix}
  i,j=1 \\
  i<j
\end{smallmatrix}}^{q+1} (-1)^{i+j}\kappa_{ij} (\rho c)([x_i,x_j]_{\wt\el},x_1,\dots,\hat{x_i},\dots,\hat{x_j},\dots,x_{q+1})
\end{align}
i.~e. $\rho c\in Z^{q} (\wt\el,\ad_{\wt\el},\kappa)$.
\end{proof}
\begin{cor}\label{invdim}
For any $q\in \N$ and any complex symmetric $(q+1)$--square matrix
$\kappa$ is the dimension of the vector space
$Z^{q}(\el,\ad_\el,\kappa)$ an invariant characteristic of Lie
algebras.
\end{cor}
Sixteen parametric spaces in Theorem~\ref{klass2} allow us to
define sixteen invariant functions of up to four variables.
However, a complete analysis of possible outcome is beyond the
scope of this work. Rather empirically, following calculations in
dimension four and eight, we pick up two one--parametric sets of
vector spaces to define two new invariant functions of a
$n$--dimensional Lie algebra $\el$. We call functions
$\fb\el,\fc\el:\Com \rightarrow \{0,1,\dots,n^2(n-1)/2\}$ defined
by the formulas
\begin{align}
(\fb\el)(\alpha) =& \dim\gc{1,1,1,\alpha,\alpha,\alpha}{\el}\\
(\fc\el)(\alpha) =& \dim\gc{0,1,1\alpha,1,1}{\el}
\end{align}
the {\bf invariant functions} corresponding to two--dimensional
twisted cocycles of the adjoint representation of a Lie algebra
$\el$.

From Theorem~\ref{tvr1koho} follows immediately:
\begin{cor}\label{invarfunc2} If two complex Lie algebras $\el,\wt\el$ are
isomorphic, $\el \cong \widetilde{\el}$, then it holds:
\begin{enumerate}
\item $\fb\el = \fb\wt\el$,
\item $\fc\el = \fc\wt\el$.
\end{enumerate}
\end{cor}
\subsection{Invariant Functions $\fb$ and $\fc$ of Low--dimensional Lie Algebras}
We now investigate the behaviour of the functions $\fb,\,\fc$ in
dimensions three and four.
\begin{example}\label{}
In Example \ref{examg342} we found the explicit form of the spaces
$\gc{1,1,1,\la,\la,\la}\g_{3,4}(a)$ and
$\gc{0,1,1,\la,1,1}\g_{3,4}(a)$, where $a\neq 0,\pm 1$. Observing
these spaces we immediately have: \vspace{0pt}
\begin{equation}\label{phi03}
\begin{tabular}[c]{|l||c|c|c|}
\hline  \parbox[l][20pt][c]{0pt}{}  $\alpha$ &   \\ \hline
\parbox[l][20pt][c]{0pt}{} $\fb\g_{3,4}(a)(\alpha)$ &6\\
\hline
\end{tabular}\qquad
\begin{tabular}[c]{|l||c|c|c|c|}
\hline  \parbox[l][20pt][c]{0pt}{}  $\alpha$ & 2&$1+a$ &
$1+\frac{1}{a}$&
\\ \hline
\parbox[l][20pt][c]{0pt}{} $\fc\g_{3,4}(a)(\alpha)$ &2&1&1&0\\
\hline
\end{tabular}
\end{equation}
\end{example}
Note that the function $\fc$ behaves like the function $\fa$ on
$\g_{3,4}(a)$. This fact allows us to derive a quite interesting
fact -- the function $\fc$ also classifies three--dimensional Lie
algebras:
\begin{lemma}\label{examg34phi}
For Lie algebra $\g_{3,4}(a),\,a\neq 0,\pm 1$ from
Lemma~\ref{examg34} it holds: if $\fc\g_{3,4}(a)=\fc\g_{3,4}(a')$
then $\g_{3,4}(a)\cong \g_{3,4}(a')$.
\end{lemma}
\begin{proof}
Observing (\ref{phi03}), we conclude that the proof is completely
analogous to the proof of Lemma \ref{examg34}.
\end{proof}

\begin{thm}[Classification of three--dimensional complex Lie
algebras II]\label{class3dim2}$\,$\newline
  Two three--dimensional complex Lie
  algebras $\el$ and $\wt\el$ are isomorphic if and only if
 \newline $\fc\el=\fc\wt\el$.
\end{thm}
\begin{proof}
$\Rightarrow:$ See Corollary~\ref{invarfunc2}.

$\Leftarrow:$ According to Lemma~\ref{examg34phi} and observing
the tables of $\fc$ in Appendix~\ref{IFLJL}, we conclude that all
tables of the invariant function~$\fc$ of non--isomorphic
three-dimensional complex Lie algebras are mutually different.
\end{proof}

\begin{example}\label{examg4223}
In Example \ref{excoc} we found the explicit form of the spaces
$\gc{1,1,1,\la,\la,\la}\g_{4,2}(a)$ and
$\gc{0,1,1,\la,1,1}\g_{4,2}(a)$, where $a\neq 0\pm1,-2$. Thus, the
invariant function $\fb$ of the algebra $\g_{4,2}(a)$ has the
following form:
\begin{equation*}
\begin{tabular}[c]{|l||c|c|c|}
\hline  \parbox[l][20pt][c]{0pt}{}  $\alpha$ & $1+a$ &
$\frac{2}{a}$ &
\\ \hline
\parbox[l][20pt][c]{0pt}{} $\fb\g_{4,2}(a)(\alpha)$ & 13  & 13 & 12\\
\hline
\end{tabular}
\end{equation*}
and the function $\fc$ has the form:
\begin{center}\vspace{-12pt}
\begin{tabular}[t]{|l||c|c|c|c|}
\hline \parbox[l][20pt][c]{0pt}{}   $\alpha$ & 2 & $1+a$ &
$1+\frac{1}{a}$ &
\\ \hline
\parbox[l][20pt][c]{0pt}{} $\fc\g_{4,2}(a)(\alpha)$ & 3 & 1 & 1 & 0 \\ \hline
\end{tabular}
\end{center}
\end{example}
Similar calculations as in the above examples lead us to the
tables of the functions $\fb$ and $\fc$ of all three and
four--dimensional Lie algebras. These tables are placed in
Appendix~\ref{IFLJL}.

\subsection{Classification of Four--dimensional Complex Lie
Algebras}\label{classfourdim} Theorem~\ref{class3dim} (or
\ref{class3dim2}) provided complete classification of
three--dimensional complex Lie algebras. We show in this section
that combined power of the functions $\fa$ and $\fb$ distinguishes
among all complex four--dimensional Lie algebras. The most
challenging problem in the identification process in any dimension
is to describe the parametric continua. The first natural question
is, whether the table of an invariant function of some parametric
continuum cannot 'degenerate' for some special value of the
parameter. For example, examining the table of $\fa
\g_{4,5}(a,-1-a)$ labeled by (g-\ref{g45am1ma}) in
Appendix~\ref{IFLJL}, we see that generally the value $5$ appears
six times there. Is it possible that for some special value of
parameter $a$ this table has different shape? We devote to this
problem the following part of this section. In order to be more
precise, let us firstly define the {\bf number of occurrences} of
$j\in\Com$ in a complex function~$f$. Let $j$ be in the range of
values of $f$. If there exist only finitely many mutually distinct
numbers $x_1,\dots,x_m\in \Com$ for which $f(x_1)=\dots=f(x_m)=j$
holds then we write $$f:j_m$$ and say that $j$ {\bf occurs} in $f$
$m$--times; otherwise we write $f:j$.
\begin{example}Consider the continuum
$ \g_{4,5}(a,-1-a)$, where $a\neq 0,\pm 1,-2
,-\frac{1}{2},-\frac{1}{2}\pm i\frac{\sqrt{3}}{2}  $, labeled by
(g-\ref{g45am1ma}) in Appendix~\ref{IFLJL}. In order to verify
that
\begin{align*} \fa\g_{4,5}(a,-1-a):\, & 6_1,5_6,4\\
\fb\g_{4,5}(a,-1-a):\, & 15_1,12
\end{align*}
we have to check for solutions each of 15 possible equalities
\begin{align*}
a=&\frac{1}{a}\\ a=&-1-a\\ a=&-\frac{a}{a+1}\\ &\vdots \\
-\frac{a}{1+a}=&-\frac{a+1}{a}.
\end{align*}
These equations have all solutions in the set $$\left\{0,\pm 1,-2
,-\frac{1}{2},-\frac{1}{2}\pm i\frac{\sqrt{3}}{2}\right\}$$ --
these values we excluded from the beginning.
\end{example}
Analogous calculations allow us to conclude:

\begin{lemma}\label{lemg441}
For the following complex four--dimensional Lie algebras defined
in Appendix \ref{IFLJL} it holds:
\begin{itemize}\item[]
\begin{itemize}
\item[(g-\ref{g34})] $\g_{3,4}(a)\oplus\g_1,\, a\neq 0,\pm 1$
\begin{align*} \fa\g_{3,4}(a)\oplus\g_1:\, & 7_1,6_3,5\\
\fb\g_{3,4}(a)\oplus\g_1:\, & 13_3,12
\end{align*}
\item[(g-\ref{g42})] $\g_{4,2}(a),\, a\neq 0,\pm 1,-2$
\begin{align*}\fa\g_{4,2}(a):\,& 6_1,5_2,4\\
\fb\g_{4,2}(a):\,&13_2,12
 \end{align*}
\item[(g-\ref{g45ab})] $\g_{4,5}(a,b),\, a\neq 0,\pm 1,\pm
b,\frac{1}{b},b^2,-1-b, $   $\ b\neq 0,\pm 1,\pm
a,\frac{1}{a},a^2,-1-a$
\begin{align*}\fa\g_{4,5}(a,b):\,& 6_1,5_6,4\\
\fb\g_{4,5}(a,b):\,&13_3,12
 \end{align*}
\item[(g-\ref{g45am1ma})] $\g_{4,5}(a,-1-a),\, a\neq 0,\pm 1,-2,-\frac{1}{2},-\frac{1}{2}\pm \frac{\sqrt{3}}{2}i$
\begin{align*} \fa\g_{4,5}(a,-1-a):\, & 6_1,5_6,4\\
\fb\g_{4,5}(a,-1-a):\, & 15_1,12
\end{align*}
\item[(g-\ref{g45aas})] $\g_{4,5}(a,a^2),\, a\neq 0,\pm 1,\pm i,-\frac{1}{2}\pm \frac{\sqrt{3}}{2}i$
\begin{align*} \fa\g_{4,5}(a,a^2):\, & 6_3,5_2,4\\
\fb\g_{4,5}(a,a^2):\, & 13_3,12
\end{align*}
\item[(g-\ref{g45a1})] $\g_{4,5}(a,1),\, a\neq 0,\pm 1,-2$
\begin{align*} \fa\g_{4,5}(a,1):\, & 8_1,6_2,4\\
\fb\g_{4,5}(a,1):\, & 15_1,13_1,12
\end{align*}
\item[(g-\ref{g45am1})] $\g_{4,5}(a,-1),\, a\neq 0,\pm 1,\pm i$
\begin{align*} \fa\g_{4,5}(a,-1):\, & 6_2,5_4,4\\
\fb\g_{4,5}(a,-1):\, & 16_1,13_2,12
\end{align*}
\item[(g-\ref{g48})] $\g_{4,8}(a),\, a\neq 0,\pm 1,\pm 2,\pm \frac{1}{2},-\frac{1}{2}\pm \frac{\sqrt{3}}{2}i$
\begin{align*} \fa\g_{4,8}(a):\, & 5_1,4_3,3\\
\fb\g_{4,8}(a):\, & 12_5,11
\end{align*}
\end{itemize}
\end{itemize}
\end{lemma}

\begin{example}\label{continuum2}
Suppose we have two--parametric family of Lie algebras
$\g_{4,5}(a,b)$ and the conditions $a\neq 0,\pm 1,\pm
b,\frac{1}{b},b^2,-1-b, $ $\ b\neq 0,\pm 1,\pm
a,\frac{1}{a},a^2,-1-a$ are satisfied. Then the function $\fb$ has
the form: \begin{equation}\label{tablegen}
\begin{tabular}[c]{|l||c|c|c|c|}
\hline \parbox[l][20pt][c]{0pt}{}   $\alpha$ &  $a+b$ &
$\frac{1+a}{b}$ & $\frac{1+b}{a}$ &
\\ \hline
\parbox[l][20pt][c]{0pt}{} $\fb\g_{4,5}(a,b)(\alpha)$ & 13 & 13 & 13& 12  \\ \hline
\end{tabular}
\end{equation}
We discuss now the following question: Is it possible to recover
the exact values of parameters $a,b$ if only the function $\fb$ of
some algebra in $\g_{4,5}(a,b)$ is known? If we have some member
of the family $\g_{4,5}(a',b')$, its function $\fb$ is of the
general form
\begin{equation}\label{tablez}
\begin{tabular}[c]{|l||c|c|c|c|}
\hline \parbox[l][20pt][c]{0pt}{}   $\alpha$ &  $z_1$ & $z_2$ &
$z_3$ &
\\ \hline
\parbox[l][20pt][c]{0pt}{} $\fb\g_{4,5}(a',b')(\alpha)$ & 13 & 13 & 13& 12  \\ \hline
\end{tabular}
\end{equation}
where $z_k\in\Com$ and $z_k\neq -1,\, k=1,2,3$. There are $3!=6$
possibilities of correspondence (permutations) between
(\ref{tablegen}) and (\ref{tablez}). One of them imply the system
of equations
\begin{equation}\label{sysss1}
\frac{1+a'}{b'}=z_1, \q  \frac{1+b'}{a'}=z_2, \q a'+b'=z_3 .
\end{equation}
with the solution
\begin{equation}\label{solsyss1}
a'=\frac{z_3+1}{z_2+1},\q b'=\frac{z_2 z_3 -1}{z_2 +1 }.
\end{equation}
We prove later that the other five permutations lead, in fact, to
the isomorphic realizations of the algebra $\g_{4,5}(a',b')$.
\end{example}
\begin{lemma}\label{lemg44}
For the four--dimensional Lie algebras from Lemma \ref{lemg441} it
holds:
\begin{itemize}\item[]
\begin{itemize}
\item[(g-\ref{g34})] If $\fa\g_{3,4}(a)\oplus \g_1=\fa\g_{3,4}(a')\oplus \g_1$ then $a'=a,\frac{1}{a}$.
\item[(g-\ref{g42})] If $\fb\g_{4,2}(a)=\fb\g_{4,2}(a')$ then $a'=a$.
\item[(g-\ref{g45ab})] If $\fb\g_{4,5}(a,b)=\fb\g_{4,5}(a',b')$ then $$(a',b')=(a,b),(b,a),\left(\frac{1}{a},\frac{b}{a}\right),\left(\frac{b}{a},\frac{1}{a}\right),
\left(\frac{1}{b},\frac{a}{b}\right),\left(\frac{a}{b},\frac{1}{b}\right)
.$$
\item[(g-\ref{g45am1ma})] If $\fa\g_{4,5}(a,-1-a)=\fa\g_{4,5}(a',-1-a')$ then
$$a'=a,\frac{1}{a},-\frac{a}{1+a},-1-\frac{1}{a},-1-a,-\frac{1}{1+a}.$$
\item[(g-\ref{g45aas})] If $\fa\g_{4,5}(a,a^2)=\fa\g_{4,5}(a',{a'}^2) $ then
$a'=a,\frac{1}{a}$.
\item[(g-\ref{g45a1})] If $\fb\g_{4,5}(a,1)=\fb\g_{4,5}(a',1)$ then $a'=a$.
\item[(g-\ref{g45am1})] If $\fb\g_{4,5}(a,-1)=\fb\g_{4,5}(a',-1)$ then $a'=a,-a$.
\item[(g-\ref{g48})] If $\fa\g_{4,8}(a)=\fa\g_{4,8}(a')$ then $a'=a,\frac{1}{a}$.
\end{itemize}
\end{itemize}
\end{lemma}
\begin{proof}
{\it Cases} (g-\ref{g34}), (g-\ref{g45aas}), (g-\ref{g45a1}),
(g-\ref{g45am1})$\,$ and$\,$ (g-\ref{g48}). In these cases is the
proof completely analogous to the proof of Lemma \ref{examg34}.

{\it Case} (g-\ref{g42}). The function $\fb$ of $\g_{4,2}(a')$ has
the form
\begin{equation}
\begin{tabular}[c]{|l||c|c|c|}
\hline  \parbox[l][20pt][c]{0pt}{}  $\alpha$ & $1+a'$ &
$\frac{2}{a'}$ &
\\ \hline
\parbox[l][20pt][c]{0pt}{} $\fb\g_{4,2}(a')(\alpha)$ & 13  & 13 & 12\\
\hline
\end{tabular}
\end{equation}
and there are two possibilities:
\begin{itemize}\item[]
\begin{itemize}
\item[(12)] If $a+1=a'+1, \, \frac{2}{a}=\frac{2}{a'}$ then
$a'=a$.
\item[(21)] If $a+1=\frac{2}{a'},\, a'+1=\frac{2}{a}$ then
$a=a'=1,-2$, which is not possible.
\end{itemize}
\end{itemize}

  {\it Case} (g-\ref{g45ab}). The function $\fb$ of
$\g_{4,5}(a',b')$ has the form
\begin{equation}\label{tablegen2}
\begin{tabular}[c]{|l||c|c|c|c|}
\hline \parbox[l][20pt][c]{0pt}{}   $\alpha$ &  $a'+b'$ &
$\frac{1+a'}{b'}$ & $\frac{1+b'}{a'}$ &
\\ \hline
\parbox[l][20pt][c]{0pt}{} $\fb\g_{4,5}(a',b')(\alpha)$ & 13 & 13 & 13& 12  \\ \hline
\end{tabular}
\end{equation}
and there are six possible correspondences between this table and
(\ref{tablegen}). Thus, substituting corresponding $z_k$'s into
(\ref{solsyss1}) we obtain:\begin{itemize}\item[]\begin{itemize}
\item[(123)] If $z_2=\frac{1+b}{a},\, z_3=a+b$ then $a'=a,b'=b$.
\item[(213)] If $z_2=\frac{1+a}{b},\, z_3=a+b$ then $a'=b,b'=a$.
\item[(132)] If $z_2=a+b,\, z_3=\frac{1+b}{a}$ then $a'=\frac{1}{a},b'=\frac{b}{a}$.
\item[(312)] If $z_2=\frac{1+a}{b},\, z_3=\frac{1+b}{a}$ then $a'=\frac{b}{a},b'=\frac{1}{a}$.
\item[(231)] If $z_2=a+1,\, z_3=\frac{1+a}{b}$ then $a'=\frac{1}{b},b'=\frac{a}{b}$.
\item[(312)] If $z_2=\frac{1+b}{a},\, z_3=\frac{1+a}{b}$ then $a'=\frac{a}{b},b'=\frac{1}{b}$.
\end{itemize}
\end{itemize}

{\it Case} 4. It is convenient to note that six values in the
table (g-18) can be arranged in the triple of pairs
$\{a,\frac{1}{a}\}$, $\{-1-a, -\frac{1}{1+a} \}$ and
$\{-\frac{1+a}{a}, -\frac{a}{1+a} \}$. Then one checks directly
only $6\cdot 2^3=48$ permutations and obtains the solutions like
in the previous case.
\end{proof}

\begin{cor}\label{corr44} For the four--dimensional Lie algebras from Lemma \ref{lemg441} it
holds:\begin{itemize}\item[]
\begin{itemize}
\item[(g-\ref{g34})] If $\fa\g_{3,4}(a)\oplus \g_1=\fa\g_{3,4}(a')\oplus \g_1$ then $\g_{3,4}(a)\oplus\g_1\cong\g_{3,4}(a')\oplus \g_1$.
\item[(g-\ref{g45ab})] If $\fb\g_{4,5}(a,b)=\fb\g_{4,5}(a',b')$ then
$\g_{4,5}(a,b)\cong\g_{4,5}(a',b')$.
\item[(g-\ref{g45am1ma})] If $\fa\g_{4,5}(a,-1-a)=\fa\g_{4,5}(a',-1-a')$ then
$\g_{4,5}(a,-1-a)\cong\g_{4,5}(a',-1-a')$.
\item[(g-\ref{g45aas})] If $\fa\g_{4,5}(a,a^2)=\fa\g_{4,5}(a',{a'}^2) $ then
$\g_{4,5}(a,a^2)\cong\g_{4,5}(a',{a'}^2) $.
\item[(g-\ref{g45am1})] If $\fb\g_{4,5}(a,-1)=\fb\g_{4,5}(a',-1)$ then $\g_{4,5}(a,-1)\cong\g_{4,5}(a',-1)$.
\item[(g-\ref{g48})] If $\fa\g_{4,8}(a)=\fa\g_{4,8}(a')$ then $\g_{4,8}(a)\cong\g_{4,8}(a')$.
\end{itemize}
\end{itemize}
\end{cor}
\begin{proof}
The statement follows from Lemma \ref{lemg44} and from the
relations
\begin{align}
\g_{3,4}(a)\oplus\g_1& \cong  \g_{3,4}(1/a)\oplus \g_1 \\
\g_{4,5}(a,b)\cong\g_{4,5}(b,a)
\cong\g_{4,5}\left(\frac{1}{a},\frac{b}{a}\right) &
\cong\g_{4,5}\left(\frac{b}{a},\frac{1}{a}\right)\cong\g_{4,5}
\left(\frac{1}{b},\frac{a}{b}\right)\cong\g_{4,5}\left(\frac{a}{b},\frac{1}{b}\right)\\
\g_{4,8}(a)&\cong\g_{4,8}(1/a),
\end{align}
which hold for all $a,b\neq 0$ and can be directly verified.
\end{proof}
\begin{thm}[Classification of four--dimensional complex Lie
algebras]\label{class4dim}  Two \newline four--dimensional complex
Lie algebras $\el$ and $\wt\el$ are isomorphic if and only if
$\fa\el=\fa\wt\el$ and $\fb\el=\fb\wt\el$.
\end{thm}
\begin{proof}
$\Rightarrow:$ See Corollaries~\ref{invarfunc} and
\ref{invarfunc2}.

$\Leftarrow:$ According to Lemmas~\ref{lemg441}, \ref{lemg44},
Corollary~\ref{corr44} and observing the tables in
Appendix~\ref{IFLJL}, we conclude that all non--isomorphic
four-dimensional complex Lie algebras differ at least in one of
the functions $\fa$ or $\fb$.
\end{proof}
\subsection{Identification of Four--dimensional Complex Lie Algebras}
An efficient algorithm for the identification of four--dimensional
Lie algebras was quite recently published in \cite{AY}. Using
Lemmas~\ref{lemg441}, \ref{lemg44}, Corollary~\ref{corr44} and
Theorem~\ref{class4dim}, we may now formulate an alternative {\it
algorithm}: take a four--dimensional complex Lie algebra~$\el$ and
\begin{enumerate}
\item Calculate $\fa \el$ and $\fb\el$.
\item The range of values of the functions $\fa$ and $\fb$ and the
number of their occurrences determines the label (g-$k$),
$k=1,\dots,34$ in Appendix~\ref{IFLJL}.
\item The algebra is now identified up to the exact value of
parameter(s) of the parametric continuum. These parameters are
determined in the following cases:
\begin{itemize}
\item[(g-\ref{g34})]
Pick any of the two values $z\in\Com, z\neq 1$, which satisfy
$\fa\el(z)=6$, and put $a=z$. Then $\el\cong
\g_{3,4}(a)\oplus\g_1$ holds.
\item[(g-\ref{g42})]
 There are two different complex numbers $z_1,z_2\neq 0$ which satisfy $\fb\el(z_1)=\fb\el(z_2)=13$.
If $z_1-1=2/z_2$ holds then put $a=z_1-1$, otherwise put
$a=z_2-1$. Then $\el\cong \g_{4,2}(a)$ holds.
\item[(g-\ref{g45ab})]
There are three mutually different complex numbers
$z_1,z_2,z_3\neq 0,-1$ which satisfy
$\fb\el(z_1)=\fb\el(z_2)=\fb\el(z_3)=13$. Put
$a=\frac{z_3+1}{z_2+1}$, $b=\frac{z_2z_3-1}{z_2+1}$. Then
$\el\cong \g_{4,5}(a,b)$ holds.
\item[(g-\ref{g45am1ma})]
Pick any of the six values $z\in\Com$, which satisfy
$\fa\el(z)=5$, and put $a=z$. Then $\el\cong \g_{4,5}(a,-1-a)$
holds.
\item[(g-\ref{g45aas})]
Pick any of the two values $z\in\Com, z\neq 1$, which satisfy
$\fa\el(z)=6$, and put $a=z$. Then $\el\cong \g_{4,5}(a,a^2)$
holds.
\item[(g-\ref{g45a1})]
Take the value $z\in\Com$, which satisfies $\fb\el(z)=15$, and put
$a=z-1$. Then $\el\cong \g_{4,5}(a,1)$ holds.
\item[(g-\ref{g45am1})] Pick any of the two values $z\in\Com$, which satisfy
$\fb\el(z)=13$, and put $a=z+1$. Then $\el\cong \g_{4,5}(a,-1)$
holds.
\item[(g-\ref{g48})] Pick any of the two values $z\in\Com$, $z\neq 2$ which satisfy
$\fa\el(z)=4$, and put $a=z$. Then $\el\cong \g_{4,8}(a)$ holds.
\end{itemize}
\end{enumerate}

We demonstrate the above algorithm of identification on the
following two examples.
\begin{example}\label{AYE1}
In \cite{AY}, a four--dimensional algebra $\el_1$ was introduced:
\begin{center}
\begin{tabular}{clll}
  $\el_1:$ & $[e_1,e_2]=-e_1-e_2+e_3,$ & $[e_1,e_3]=-6e_2+4e_3,$& $[e_1,e_4]=2e_1-e_2+e_4,$ \\
          & $[e_2,e_3]=3e_1-9e_2+5e_3,$  & $[e_2,e_4]=4e_1-2e_2+2e_4,$  & $[e_3,e_4]=6e_1-3e_2+3e_4.$\\
\end{tabular}
\end{center}
\begin{enumerate}\item
Computing the functions $\fa\el_1$ and $\fb\el_1$ one obtains:
\begin{center}\vspace{-12pt}
\begin{tabular}[t]{|l||c|c|c|c|}
\hline \parbox[l][20pt][c]{0pt}{}   $\alpha$ & 1 & 2&
$\frac{1}{2}$& \\ \hline
\parbox[l][20pt][c]{0pt}{} $\fa\el_1(\alpha)$ & 6 & 5 & 5 & 4  \\ \hline
\end{tabular}\qquad
\begin{tabular}[t]{|l||c|c|c|}
\hline  \parbox[l][20pt][c]{0pt}{}  $\alpha$ & 3 & 1 &
\\ \hline
\parbox[l][20pt][c]{0pt}{} $\fb\el_1(\alpha)$ & 13  & 13 & 12\\
\hline
\end{tabular}
\end{center}
\item The combination of occurrences $\fa\el_1: 6_1,5_2,4$ and $\fb\el_1:
13_2,12$ is unique for the case~(g-\ref{g42}).
\item Since for $z_1=3$, $z_2=1$ the equality $z_1-1=2/z_2$ holds,
one has $a=z_1-1=2$ and $\el_1\cong\g_{4,2}(2)$.
\end{enumerate}
\end{example}
\begin{example}\label{AYE2}
In \cite{AY}, a four--dimensional algebra $\el_2$ was also
introduced:
\begin{center}
\begin{tabular}{clll}
  $\el_2:$ & $[e_1,e_2]=4e_1+3e_2-6e_3+2e_4,$ & $[e_1,e_3]=15e_1+5e_2-15e_3+5e_4,$\\
  &$[e_1,e_4]=50e_1+15e_2-48e_3+16e_4,$ & $[e_2,e_3]=21e_1+2e_2-15e_3+5e_4,$\\
  & $[e_2,e_4]=93e_1+21e_2-81e_3+27e_4,$  & $[e_3,e_4]=90e_1+25e_2-84e_3+28e_4.$\\
\end{tabular}
\end{center}
\begin{enumerate}\item
Computing the functions $\fa\el_2$ and $\fb\el_2$ one obtains:
\begin{center}\vspace{-12pt}
\begin{tabular}[t]{|l||c|c|c|c|}
\hline \parbox[l][20pt][c]{0pt}{}   $\alpha$ &  $0$&
\\ \hline
\parbox[l][20pt][c]{0pt}{} $\fa\el_2(\alpha)$ & 6 & 4  \\ \hline
\end{tabular}\qquad
\begin{tabular}[t]{|l||c|c|c|}
\hline  \parbox[l][20pt][c]{0pt}{}  $\alpha$ & $0$  & 1&
\\ \hline
\parbox[l][20pt][c]{0pt}{} $\fb\el_2(\alpha)$ & 12  & 12 &10 \\
\hline
\end{tabular}\qquad
\end{center}

\item The combination of occurrences $\fa\el_2: 6_1,4$ and $\fb\el_2:
12_2,10$ is unique for the case~(g-\ref{g21og21}) and one has
$\el_2\cong\g_{2,1}\oplus\g_{2,1}$.
\end{enumerate}
\end{example}


\chapter{Contractions of Algebras}\label{CHcon}
\section{Continuous Contractions of Algebras}
Except where explicitly stated, Section 4.1 of this chapter
contains original and unpublished results. The content of Section
4.2 is based on \cite{HN6,HN2,HNx}.

Suppose we have an arbitrary algebra $\A=(V,\,\cdot\,)$ and a
continuous mapping $U: (0,1\rangle \map GL(V)$, i.~e. $U(\ep)\in
GL(V), \: 0<\ep\leq 1$.  If the limit
\begin{equation}\label{contr}
x\cdot_0  y=\lim_{\ep \map 0+}U(\ep )^{-1}(U(\ep )x\cdot U(\ep )y)
\end{equation}
exits for all $x,y\in V$ then we call the algebra
$\A_0=(V,\cdot_0)$ a {\bf one--parametric continuous contraction}
(or simply a {\bf contraction}) of the algebra $\A$ and write
$\A\map \A_0$. We call the contraction $\A\map \A_0$ {\bf proper}
if $\A\ncong \A_0$. Contraction to the Abelian algebra $\A_0$,
$$x\cdot_0 y=0, \,\forall x,y\in \A_0 $$ is always possible via
$U(\ep)= \ep\, 1$. We call all improper contractions and
contractions to the Abelian algebra {\bf trivial}.

It is well known that if $\el\map \el_0$ is any one--parametric
continuous contraction of a Lie algebra $\el$ then $\el_0$ is also
a Lie algebra. Invariant characteristics of Lie algebras change
after a contraction. The relation among these characteristics
before and after a contraction form useful necessary contraction
criteria. For example, such a set of these criteria, which
provided the complete classification of contractions of three and
four--dimensional Lie algebras, has been found in~\cite{Nes}. Our
aim is to state new necessary contraction criteria using
$(\alpha,\beta,\gamma)$--derivations and twisted cocycles. We
first give a criterion from \cite{Bur2,Campocoh}:
\begin{thm}
Let $\el $ be a complex Lie algebra and $\el\map \el_0$, and $q\in
\N$. Then it holds:
\begin{equation}\label{dimdercon}
  \dim Z^{q} (\el,\ad_\el) \leq \dim  Z^{q} (\el_0,\ad_{\el_0}).
\end{equation}
\end{thm}
There is a straightforward generalization of the above theorem:
\begin{thm}\label{dimderconthm}
Let $\el $ be a complex Lie algebra, $\el\map \el_0$ and $q\in
\N$. Then for any $(q+1)\times(q+1)$ complex symmetric matrix
$\kappa$
\begin{equation}\label{dimdercon2}
  \dim Z^{q} (\el,\ad_\el,\kappa) \leq \dim  Z^{q} (\el_0,\ad_{\el_0},\kappa)
\end{equation}
holds.
\end{thm}
\begin{proof}
Suppose that the contraction $\el\map \el_0$ is performed by the
mapping $U$, i.~e. $[x,y]_0=\lim_{\ep\map 0+}[x,y]_\ep ,$ where
$$[x,y]_\ep=U(\ep)^{-1}[U(\ep)x,U(\ep)y],\q \forall x,y\in\el.$$
Suppose $\el=(V,[\,,\,])$ and let us fix a basis
$\{x_1,\dots,x_n\}$ of $V$. We denote the structural constants of
the algebra $\el$ by~$c_{ij}^k$ and the structural constants of
the algebras $\el_\ep=(V,[\,,\,]_\ep)$ by~$c_{ij}^k(\ep)$. Then it
holds
\begin{equation}\label{limitstruc}
c_{ij}^k(0)=\lim_{\ep\map 0+}c_{ij}^k(\ep),
\end{equation}
where $c_{ij}^k(0)$ are the structural constants of $\el_0$.  The
dimension of the space $Z^{q} (\el,\ad_\el,\kappa)$ is determined
via the relation
\begin{equation}\label{rank}
  \dim Z^{q} (\el,\ad_\el,\kappa)= \dim
C^q(\el,\el)-\operatorname{rank} S^q(\el,\kappa),
\end{equation}
where $S^q(\el,\kappa)$ is the matrix corresponding to the linear
system of equations generated from (\ref{koho1t}). We write the
explicit form of this system for $q=1$. Then we obtain from
(\ref{koho1t}) that $D=(D_{ij})\in Z^{1}
\left(\el,\ad_\el,\left(\begin{smallmatrix}
  \beta &  \alpha \\
   \alpha & \gamma
\end{smallmatrix}\right) \right)$ if and only if the
linear system with the matrix $ S^{1}
\left(\el,\left(\begin{smallmatrix}
  \beta &  \alpha \\
   \alpha & \gamma
\end{smallmatrix}\right) \right) $ is satisfied
\begin{equation}\label{sys11}
 S^{1}
\left(\el,\left(\begin{smallmatrix}
  \beta &  \alpha \\
   \alpha & \gamma
\end{smallmatrix}\right) \right):\q \sum_{r=1}^n-\alpha c^r_{ij}D_{sr}+\beta
c^s_{rj}D_{ri}+\gamma c^s_{ir}D_{rj}=0, \q \forall i,j,s\in
\{1,\dots,n\},
\end{equation}
and similarly for $q>1$. Since $\el_\ep\cong\el$ holds for all
$0<\ep \leq 1 $, we see from Corollary~\ref{invdim} that
\begin{equation}\label{invzet}
\dim Z^{q} (\el,\ad_\el,\kappa)=\dim Z^{q}
(\el_\ep,\ad_{\el_\ep},\kappa), \q 0<\ep \leq 1,\,q\in\N.
\end{equation}
Since the relation
\begin{equation}\label{cql}
\dim C^q(\el,\el)= \dim C^q(\el_\ep,\el_\ep)= \dim
C^q(\el_0,\el_0),\q 0<\ep \leq 1,\,q\in\N
\end{equation}
holds, the relations~(\ref{rank}), (\ref{invzet}) then imply that
\begin{equation}\label{rank2}
\operatorname{rank} S^q(\el,\kappa)= \operatorname{rank}
S^q(\el_\ep,\kappa), \q 0<\ep \leq 1,\,q\in\N.
\end{equation}
The rank of the matrix $S^q(\el,\kappa)$ is equal to $r$ if and
only if there exists a non-zero minor of the order $r$ and every
minor of order higher than $r$ is zero. It follows from
(\ref{rank2}) that all minors of the orders higher than $r$ of the
matrices $S^q(\el_\ep,\kappa)$ are zeros. Since the equality
(\ref{limitstruc}) holds, all minors of the matrices
$S^q(\el_\ep,\kappa)$ converge to the minors of the matrix
$S^q(\el_0,\kappa)$. Thus, as the limits of zero functions, all
minors of order higher than $r$ of the matrix $S^q(\el_0,\kappa)$
are also zero. Therefore $\operatorname{rank}
S^q(\el_0,\kappa)\leq r$ and the statement of the theorem follows
from~(\ref{rank}) and~(\ref{cql}).
\end{proof}
There exist other necessary contraction criteria, similar to
(\ref{dimdercon}), (\ref{dimdercon2}) -- certain inequalities
between invariants. However, one very powerful criterion is quite
unique. This highly non-trivial theorem, very useful in
\cite{Bur2,Bur1,Nes}, was originally proved in \cite{Bor}.
\begin{thm}
If $\el_0$ is a proper contraction of a complex Lie algebra $\el$
then it holds:
\begin{equation}\label{dimdercon3}
  \dim \der \el < \dim \der \el_0.
\end{equation}
\end{thm}
\begin{cor}\label{concritmain}
If $\el_0$ is a proper contraction of a complex Lie algebra $\el$
then it holds: \begin{enumerate}\item $\fa\el\leq\fa\el_0 $ \item
$\fa\el(1)<\fa\el_0(1)$.
\end{enumerate}
\end{cor}
\begin{proof}
Since $\fa\el (\alpha)=\dim Z^{1}
\left(\el,\ad_\el,\left(\begin{smallmatrix}
  1 &  \alpha \\
   \alpha & 1
\end{smallmatrix}\right) \right)$ the first inequality follows from
(\ref{dimdercon2}) and the second from $\fa\el (1)= \dim \der \el$
and (\ref{dimdercon3}).
\end{proof}
\begin{cor}\label{concrit2}
If $\el_0$ is a contraction of a complex Lie algebra $\el$ then it
holds: \begin{enumerate}\item $\fb\el\leq\fb\el_0 $ \item
$\fc\el\leq\fc\el_0$.
\end{enumerate}
\end{cor}
\begin{proof}
Since $\fb\el (\alpha)=\dim Z^{2}
\left(\el,\ad_\el,\left(\begin{smallmatrix}
  \alpha &  1 & 1 \\
   1 & \alpha & 1 \\
   1 & 1 &\alpha
\end{smallmatrix}\right) \right)$ the first inequality follows from
(\ref{dimdercon2}); the proof of the second condition is
analogous.
\end{proof}

\begin{example}
Consider the contraction $\g_{3,2}\map \g_{3,3}$. In Table
\ref{tab1} the dimensions~(\ref{defdim}) of the associated Lie and
Jordan algebras of $\g_{3,2}$ and $\g_{3,3}$ are listed.
\begin{table}[!ht]
\centering\caption{\it Dimensions of the associated Lie and Jordan
algebras of~$\g_{3,2}$ and $\g_{3,3}$\hspace{55pt}
}\label{tab1}\vspace{-4pt}
\begin{tabular}[t]{|c||c|c|c|c|c|c|c|c|c|c|c}\hline
\hspace{80pt}\parbox[l][22pt][c]{0pt}{} & $d_{(1,1,1)}$ &
$d_{(0,1,1)}$ & $d_{(1,1,0)}$ &  $d_{(1,1,1)(0,1,1)}$&
$d_{(1,1,-1)}$ & $d_{(0,1,-1)}$\\ \hline\hline
\parbox[l][22pt][c]{0pt}{} $\g_{3,2}$& 4&3&1&2&0&1\\ \hline
\parbox[l][22pt][c]{0pt}{} $\g_{3,3}$& 6&3&1&2&0&1\\ \hline
\end{tabular}
\end{table}
Note that except for the dimensions of the algebras of derivations
$$4=\dim\der\g_{3,2}<\dim\der\g_{3,3}=6$$ none of these dimensions
grows after the contraction.
\end{example}

\subsection{Contractions of Low--dimensional Lie Algebras}

In Section~\ref{invfoflow} we have used the invariant function
$\fa$ to classify all three--dimensional Lie algebras. We now
employ the necessary contraction criterion of
Corollary~\ref{concritmain} to describe all possible contractions
among these algebras. The behaviour of the function $\fa$
determines the classification and contractions of
three--dimensional Lie algebras. Contractions of
three--dimensional algebras were the most recently classified in
\cite{Nes}:
\begin{thm}\label{contract3dim}
Only the following non--trivial contractions among
three--dimensional Lie algebras exist:
\begin{enumerate}
\item $\g_{3,4}(-1)$ is a contraction of $\slp(2,\Com)$,
\item $\g_{3,3}$ is a contraction of $\g_{3,2}$,
\item $\g_{3,1}$ is a contraction of $\g_{3,2}$, $\g_{3,4}(a)$,
$\g_{2,1}\oplus\g_1$ and $\slp(2,\Com)$.
\end{enumerate}
\end{thm}
We present two examples which show that both items in
Corollary~\ref{concritmain} are effective for the description of
contractions among three--dimensional Lie algebras.
\begin{example}\label{ExcL1}
Observing the corresponding tables of the functions $\fa$, we see
that the contraction of $\g_{3,2}$ to $\g_{2,1}\oplus\g_1$ is not
possible due to $\fa\g_{3,2}(1)=\fa\g_{2,1}\oplus\g_1(1)=4$ -- a
contradiction to the item 2. of the Corollary~\ref{concritmain}.
The necessary condition $\fa\g_{3,2}\leq\fa\g_{2,1}\oplus\g_1$ is
satisfied and, thus, does not exclude the existence of a
contraction.
\end{example}
\begin{example}\label{ExcL2}
On the contrary to the previous example, observing the
corresponding tables of the functions $\fa$, we see that the
contraction of $\slp(2,\Com)$ to $\g_{2,1}\oplus\g_1$ is not
possible due to
$5=\fa\slp(2,\Com)(-1)>\fa\g_{2,1}\oplus\g_1(-1)=4$ -- a
contradiction to the item 1. of the Corollary~\ref{concritmain}.
The necessary condition
$3=\fa\slp(2,\Com)(1)<\fa\g_{2,1}\oplus\g_1(1)=4$ is satisfied
and, thus, does not exclude the existence of a contraction.
\end{example}
A similar analysis of all possible pairs of three--dimensional Lie
algebras leads us to the following theorem.
\begin{thm}[Contractions of three--dimensional complex Lie
algebras]\label{conthmmain}$\,$\newline
  Let $\el$, $\el_0$ be two three--dimensional complex Lie
  algebras. Then there exists a proper one--parametric continuous
  contraction $\el\map\el_0$ if and only if
  \begin{equation*}
  \fa\el\leq\fa\el_0 \quad \text{and}\quad \fa\el(1)<\fa\el_0(1).
\end{equation*}
\end{thm}
\begin{proof}
$\Rightarrow$ : This implication is, in fact,
Corollary~\ref{concritmain}.

$ \Leftarrow$: This implication follows from a direct comparison
of the tables of the invariant functions $\fa$ of
three--dimensional Lie algebras in Appendix~\ref{IFLJL} and
Theorem~\ref{contract3dim}.
\end{proof}

In Section \ref{classfourdim} we used the invariant functions
$\fa$ and $\fb$ to classify all four--dimensional Lie algebras. In
order to obtain stronger contraction criteria, we also defined the
function $\fc$ -- a supplement to the functions $\fa$ and $\fb$
(see Example~\ref{fbcon}). The combined forces of the
Corollaries~\ref{concritmain} and \ref{concrit2}, though strong,
do not provide us with a complete classification of contractions
of four--dimensional Lie algebras. So far, we were unable to find
more suitable definitions of invariant functions, which offers the
concept of two--dimensional twisted cocycles, allowing such
classification. The complete description of the spaces of
two--dimensional twisted cocycles for four--dimensional Lie
algebras, similar to the classification of
$(\alpha,\beta,\gamma)$--derivations in Appendix \ref{APAL}, would
solve the existence of such functions explicitly. However, such a
complete description seems, at the moment, out of reach. We
discuss the application of the criteria of the
Corollaries~\ref{concritmain} and \ref{concrit2} to the
four--dimensional Lie algebras in the following examples.
\begin{example}
To demonstrate behaviour of the functions $\fa,\fb$ and $\fc$ in
dimension four, we consider the following sequence of contractions
\cite{Bur1,Nes} : $$\slp(2,\Com) \oplus\g_1\ \map\ \g_{4,8}(-1)\
\map\ \g_{3,4}(-1)\oplus\g_1\ \map\ \g_{4,1}\ \map\
\g_{3,1}\oplus\g_1\ \map\ 4\g_1. $$ Note in Table~\ref{tab2}, how
the value of each invariant function is greater or equal than the
value in the previous row. As expected, the strict inequality for
the values~$\fa(1)$ holds -- in this case the sequence of
dimensions: $4,\,5\,,6\,,7\,,10\,,16$. The strict increase of
values is also identified in the following cases: $\fa(2)$,
'generic values' of $\fa$, $\fb(1/2)$ and 'generic values'
of~$\fb$. These conjectures of strict inequalities are, however,
not valid for the general case of a contraction in dimension four.
\begin{table}[!ht]
\begin{center}
\centering \caption[l]{\it Invariant functions $\fa,\, \fb$ and
$\fc$ of the contraction sequence: $\slp(2,\Com) \oplus\g_1\ \map\
\g_{4,8}(-1)\ \map\ \g_{3,4}(-1)\oplus\g_1\ \map\ \g_{4,1}\ \map\
\g_{3,1}\oplus\g_1\ \map\ 4\g_1. $} \label{tab2}\vspace{-6pt}
\begin{tabular}[t]{|c||c|c|c|c|c||c|c|c|c|c||c|c|c|c|}
\cline{1-15} \parbox[l][22pt][c]{0pt}{} \hspace{108pt}
&\multicolumn{5}{|c||}{$\fa (\al)$}& \multicolumn{5}{|c||}{$\fb
(\al)$} & \multicolumn{4}{c|}{$\fc (\al)$} \\

\hline
\parbox[l][22pt][c]{0pt}{}  $\al$ & -1& 0& 1 &2&  & -1 & 0 & 1 & $\frac{1}{2}$ &  & 0 & 1 & 2 &\\ \hline\hline

\parbox[l][22pt][c]{0pt}{}  $\slp(2,\Com)\oplus\g_1$ & 6& 4& 4 &2& 1 & 14 & 12 & 12 & 10 & 9  &  0 & 0 & 1 & 0\\ \hline

\parbox[l][22pt][c]{0pt}{}  $\g_{4,8}(-1)$ & 6& 4& 5 &4& 4 & 14 & 12 & 13 &12 &12  &  0 & 0& 1 &0\\ \hline

\parbox[l][22pt][c]{0pt}{}  $\g_{3,4}(-1)\oplus\g_1$ & 7& 7& 6 &5& 5 & 16 & 16 & 15 & 14 & 14 &  3 & 3 & 3 & 1\\ \hline

\parbox[l][22pt][c]{0pt}{}  $\g_{4,1}$ & 7& 7& 7 &7& 7 & 16 & 16 & 15 & 15 & 15 & 3& 3 & 3 & 3\\ \hline

\parbox[l][22pt][c]{0pt}{}  $\g_{3,1}\oplus\g_1$ & 10& 11& 10 &10& 10 & 19 & 20 & 19 & 19  &19  & 8& 11 & 8 & 8\\ \hline

\parbox[l][22pt][c]{0pt}{}  $4\g_1$ & 16& 16& 16 &16& 16 & 24 & 24 & 24 &24 & 24 & 24& 24 & 24 &24\\ \hline

\end{tabular}\end{center}
\end{table}
\end{example}
\begin{example}
Consider the pair of four--dimensional Lie algebras
$\g_{4,2}(a),\,a\neq 0,\pm 1,-2$ and $\g_{4,5}(a',1),\,a'\neq
0,\pm 1,-2$. There are two possibilities, how the corresponding
tables of the invariant functions $\fb$, in Appendix~\ref{IFLJL}
cases~(g-\ref{g42}) and (g-\ref{g45a1}), can satisfy
$\fb\g_{4,2}(a)\leq \fb\g_{4,5}(a',1)$. The first possibility
leads to conditions $a'+1=2/a$ and $a+1=2/a'$ -- these have
solutions $a=a'=1,-2$ and we excluded them. The second possibility
implies $a=a'$. The necessary condition 1. of
Corollary~\ref{concrit2} therefore admits only the contraction
$\g_{4,2}(a)\map\g_{4,5}(a,1)$. This contraction indeed exists
\cite{Nes}. In Table~\ref{tab3} we summarize the behaviour of the
functions $\fa,\, \fb$ and $\fc$. Note that the function $\fa$
grows only at the points $1,\,a,\,\frac{1}{a}$ and the function
$\fb$ only at one(!) point $1+a$.
\begin{table}[h]
\begin{center}
\centering \caption[l]{\it Invariant functions $\fa,\, \fb$ and
$\fc$ of the contraction:
$\g_{4,2}(a)\map\g_{4,5}(a,1)$\hspace{50pt}}\label{tab3}\vspace{3pt}
\begin{tabular}[b]{|c||c|c|c|c||c|c|c||c|c|c|c|}
\hline\parbox[l][22pt][c]{0pt}{} \hspace{108pt}
&\multicolumn{4}{c||}{$\fa (\al)$}& \multicolumn{3}{|c||}{$\fb
(\al)$} & \multicolumn{4}{c|}{$\fc (\al)$} \\ \hline
\parbox[l][22pt][c]{0pt}{}  $\al$ & $\:\:1\:\:$ & $\:\:a\:\:$& $\:\:\frac{1}{a}\:\:$ &\hspace{16pt}  & $1+a$ & $\:\:\frac{2}{a}\:\:$ & \hspace{16pt}& $\:2\:$ & $1+a$ & $1+\frac{1}{a}$ &\hspace{12pt} \\ \hline\hline

\parbox[l][22pt][c]{0pt}{}  $\g_{4,2}(a)$ & 6& 5& 5 &4& 13 & 13 & 12 & 3 & 1 & 1  &  0 \\ \hline

\parbox[l][22pt][c]{0pt}{}  $\g_{4,5}(a,1)$ & 8& 6& 6 &4& 15 & 13 & 12 & 7 &2 &2  &  0 \\ \hline

\end{tabular}\end{center}
\end{table}
\end{example}
\begin{example}\label{fbcon}
Consider the pair of four--dimensional Lie algebras $\g_{4,7}$ and
$\g_{4,2}(1)$. The necessary conditions $\fa\g_{4,7}\leq
\fa\g_{4,2}(1), $ $[\fa\g_{4,7}](1)< [\fa\g_{4,2}(1)](1) $ and
$\fb\g_{4,7}\leq \fb\g_{4,2}(1)$ are satisfied. But since it holds
$$1=[\fc\g_{4,7}]\left(\frac{3}{2}\right)>
[\fc\g_{4,2}(1)]\left(\frac{3}{2}\right)=0, $$ a contraction is
not possible.
\end{example}

\subsection{Contractions of Two--dimensional Jordan Algebras}
The classification of Jordan algebras is a far more challenging
problem than the classification of Lie algebras. Indeed, even the
classification in dimension $2$ involves the solution of a system
of 13 cubic equations from Section~\ref{jordan} with 6 variables.
This problem was solved recently~\cite{jordanesp}. The resulting
five two--dimensional complex Jordan algebras, together with the
explicit form of the spaces $\gd{\alpha,\beta,\gamma}\A$ and the
function $\fa$, are listed in Appendices~\ref{APAJ}
and~\ref{IFLJJ}. We have:
\begin{thm}[Classification of two--dimensional complex Jordan algebras]
$\,$\newline  Two two--dimensional complex
  Jordan algebras $\mathcal{J}$ and $\wt{\mathcal{J}}$ are
  isomorphic if and only if
$\fa\mathcal{J}=\fa\wt{\mathcal{J}}.$
\end{thm}
\begin{proof}
The result follows from the direct comparison of the
classification of two-dimensional complex Jordan algebras and the
corresponding values of the invariant function~$\fa$ -- see
Appendix~\ref{IFLJJ}.
\end{proof}
It was also pointed out~\cite{jordanesp} that the necessary
contraction criterion (\ref{dimdercon3}) holds for
two--dimensional Jordan algebras. The modification of the proof of
Theorem~\ref{dimderconthm} for Jordan algebras and $q=1$ is
straightforward. Thus, the necessary criterion in
Corollary~\ref{concritmain} is also valid for Jordan algebras. The
behaviour of the function $\fa$ determines the classification and
contractions of two--dimensional complex Jordan algebras. The
contractions of two--dimensional complex Jordan algebras were
recently also classified~\cite{jordanesp}:

\begin{thm}\label{contract3dimJ}
Only the following non--trivial contractions among
two--dimensional complex Jordan algebras exist:
\begin{enumerate}
\item $\j_{2,1}$ is a contraction of $\j_{2,5}$,
\item $\j_{2,2}$ is a contraction of $\j_{2,5}$,
\item $\j_{2,3}$ is a contraction of $\j_{2,1}$, $\j_{2,2}$ and $\j_{2,5}$.
\end{enumerate}
\end{thm}
Comparing the tables of the invariant functions $\fa$ of
two--dimensional complex Jordan algebras in Appendix~\ref{IFLJJ}
and Theorem~\ref{contract3dimJ} we formulate an analogous theorem
to Theorem~\ref{conthmmain}.
\begin{thm}[Contractions of two--dimensional complex Jordan
algebras]\label{conthmmainJ} $\, $\newline
  Let $\mathcal{J}$ and $\mathcal{J}_0$ be two two--dimensional complex
  Jordan
  algebras. Then there exists a proper one--parametric continuous
  contraction $\mathcal{J}\map\mathcal{J}_0$ if and only if
  \begin{equation*}
  \fa\mathcal{J}\leq\fa\mathcal{J}_0 \quad \text{and}\quad \fa\mathcal{J}(1)<\fa\mathcal{J}_0(1).
\end{equation*}
\end{thm}
Similarly to Examples \ref{ExcL1}, \ref{ExcL2}, we show that both
conditions in Theorem~\ref{conthmmainJ} are necessary for the
description of contractions of two--dimensional Jordan algebras.
\begin{example}\label{ExcJ1}
Observing the corresponding tables of the functions $\fa$, we see
that the contraction of $\j_{2,1}$ to $\j_{2,2}$ is not possible
due to $\fa\j_{2,1}(1)=\fa\j_{2,2}(1)=1$. The necessary condition
$\fa\j_{2,1}\leq\fa\j_{2,2}$ is satisfied and, thus, does not
exclude the existence of a contraction.
\end{example}
\begin{example}\label{ExcJ2}
On the contrary to the previous example, observing the
corresponding tables of the functions $\fa$, we see that the
contraction of $\j_{2,1}$ to $\j_{2,4}$ is not possible due to
$2=\fa\j_{2,1}(2)>\fa\j_{2,4}(2)=1$. The necessary condition
$1=\fa\j_{2,1}(1)<\fa\j_{2,4}(1)=2$ is satisfied and does not
exclude the existence of a contraction.
\end{example}

\section{Graded Contractions of Lie Algebras} Suppose we have
some Abelian group $G$. We say that a complex Lie algebra
$\el=(V,[\,\,,\,\,])$ is $G${\bf --graded}, if there is a
decomposition into subspaces $L_i,\,i\in G$ -- called a {\bf
grading} $\Gamma
$
--
\begin{equation}\label{gra}\Gamma : \el= \bigoplus _{i \in G}
L_i
\end{equation} and the relation
\begin{equation}\label{grupa}
  [L_i,L_j]\subseteq L_{i+j},\q \forall i,j \in G
\end{equation}
holds. The decomposition $\el=\bigoplus _{i \in G} gL_i$, where
$g\in \aut \el$, is also a grading of $\el$ and is {\bf
equivalent} to $\Gamma$.

We define for all $x \in {L}_{i}, \ y\in {L}_{j},\, i,j\in G$ and
$\ep_{ij}\in \Com$ a new bilinear mapping on $V$ by the formula
\begin{equation}\label{kontr}
  [x,y]_{\ep} = \ep_{ij} [x,y].
\end{equation}
Since we claim the bilinearity of $[\,\,,\,\,]_{\ep}$, the
condition (\ref{kontr}) determines this mapping on the whole $V$.
If $\el^\ep:=(V,[\,\,,\,\,]_{\ep})$ is a Lie algebra, then it is
called a {\bf graded contraction} of the Lie algebra $\el$. Note
that the contraction preserves a grading because it is also true
that
\begin{equation}\label{gra2}
\el^\ep= \bigoplus _{i \in G} L_i
\end{equation}
is a grading of $\el^\ep$.

There are two conditions which the parameters $\ep_{ij}$ must
fulfill. Antisymmetry of $[\,\,,\,\,]_{\ep}$ immediately gives
\begin{equation}\label{antisd}
  \ep_{ij}=\ep_{ji}.
\end{equation}
The validity of the Jacobi identity requires: for all (unordered)
triples $i,j,k \in G$
\begin{equation}\label{eqpp} e(i\:j\:k):\, [x,[y,z]_\ep]_\ep + [z,[x,y]_\ep]_\ep +[y,[z,x]_\ep]_\ep
=0\q (\forall x \in L_i)(\forall y\in L_j)(\forall z \in L_k).
\end{equation}
is satisfied. Each set of $\ep_{ij}$'s which satisfies the above
conditions -- determines a Lie algebra -- can be written in the
form of a symmetric matrix $\ep =(\ep_{ij})$ which is called a
{\bf contraction matrix}.

\subsection{Invariant Functions and Graded Contractions of $\slp(3,\Com)$}\label{higher}
The Lie algebra $\slp(3,\Com)$ has four non--equivalent gradings
\cite{HPP6}. One of them is a $\Z_3\times\Z_3$--grading, called
the {\bf Pauli grading} \cite{PZ1}, and has the following explicit
form:
\begin{align}\label{zkos} \slp(3,\Com) 
&=\Com
\begin{pmatrix} 1 & 0 & 0
\\ 0 & \omega  & 0 \\ 0 & 0 & \omega ^{2}
\end{pmatrix}
\oplus \Com
\begin{pmatrix}
1 & 0 & 0 \\ 0 & \omega ^{2} & 0 \\ 0 & 0 & \omega
\end{pmatrix}
\oplus \Com
\begin{pmatrix}
0 & 1 & 0 \\ 0 & 0 & 1 \\ 1 & 0 & 0
\end{pmatrix}
 \oplus \Com
\begin{pmatrix}
0 & 0 & 1 \\ 1 & 0 & 0 \\ 0 & 1 & 0
\end{pmatrix}
 \oplus \nonumber\\& \oplus \Com
\begin{pmatrix}
0 & \omega  & 0 \\ 0 & 0 & \omega ^{2} \\ 1 & 0 & 0
\end{pmatrix}
 \oplus \Com
\begin{pmatrix}
0 & 0 & \omega  \\ 1 & 0 & 0 \\ 0 & \omega ^{2} & 0
\end{pmatrix}
 \oplus \Com
\begin{pmatrix}
0 & \omega ^{2} & 0 \\ 0 & 0 & \omega  \\ 1 & 0 & 0
\end{pmatrix}
 \oplus \Com
\begin{pmatrix}
0 & 0 & \omega ^{2} \\ 1 & 0 & 0 \\ 0 & \omega  & 0
\end{pmatrix} \nonumber
\end{align}
where $\omega=\exp (2\pi i/3)$.

In~\cite{HN2}, all graded contractions corresponding to the Pauli
grading of $\slp(3,\Com)$ have been found. Classifying the outcome
of the contraction, it turned out that the set of invariant
characteristics of Lie algebras, basically the set $\inv \el$
without the dimension of associated Lie algebras, is insufficient
to distinguish among the results. This difficulty was dealt with
by the explicit calculation of the isomorphism problem: complex
Lie algebras $\el$ and $\wt\el$ of the same dimension $n$
determined by structure constants $c_{ij}^k$ and $\wt c_{ij}^k$ in
are isomorphic if and only if there exists a regular matrix
$A=(A_{ij})\in GL(n,\Com)$, whose elements satisfy the following
system of quadratic equations
\begin{equation}\label{ire}
\sum_{r=1}^{n} c_{ij}^{r}A_{kr} = \sum_{s, t =1}^{n} A_{s i}A_{t
j}\wt c_{s t}^k,\quad \forall i,j,k\in\{1,\dots , n\}.
\end{equation}

This direct computation was used in \cite{HN2} to show that the
two graded contractions of the Pauli graded $\slp(3,\Com)$ denoted
by $\el_{17,9}$, $\el_{17,12}$ and defined already in
Example~\ref{intersect} are non--isomorphic. We have proved this
fact in Example~\ref{intersect} by analyzing the structure of
associated Lie algebra $\gd{1,1,1}\el\cap\gd{0,1,1}\el$. It was
pointed out in~\cite{HN6} that the set of invariants $\inv \el$
enriched by the sets $\inv\gd{1,1,1}$ and $\inv
[\gd{1,1,1}\el\,\cap\,\gd{0,1,1}\el]$ is sufficient for the
description of all results of the Pauli graded $\slp(3,\Com)$ --
parametric continua excluded.

Employing the invariant function $\fa$, one may also distinguish
between $\el_{17,9}$ and $\el_{17,12}$.

\begin{example}\label{intersect2}
Consider the graded contractions of the Pauli graded
$\slp(3,\Com)$ denoted by $\el_{17,9}$ and $\el_{17,12}$ and
defined in Example~\ref{intersect}. We have seen that calculating
the set of invariants $\inv \el$ one has: $$
\begin{array}{llll}
\inv\el_{17,9} =\inv\el_{17,12} & (8,4,0) (8,4,2,0)(2,5,8) & 2 &
[16,19,9,11]
\end{array}
$$ Calculating the dimensions of the associated Jordan algebras we
obtain:
\begin{align*}
\dim \gd{1,1,-1}\el_{17,9}&=\dim \gd{1,1,-1}\el_{17,12}=8
\\ \dim \gd{0,1,-1}\el_{17,9}&=\dim \gd{0,1,-1}\el_{17,12}=17
\end{align*}
We may add the values of the invariant function $\fa^0$
\begin{center}
\vspace{-12pt}
\begin{tabular}[t]{|l||c|c|c|}
\hline  \parbox[l][20pt][c]{0pt}{}  $\alpha$ & 0 & 1 & \\ \hline
\parbox[l][20pt][c]{0pt}{} $\fa^0\el_{17,9}(\alpha)$ & 16  & 9 & 8\\ \hline
\parbox[l][20pt][c]{0pt}{} $\fa^0\el_{17,12}(\alpha)$ & 16  & 9 & 8\\
\hline
\end{tabular}
\end{center}
and we see that the unique characterization is not attained.
However, calculating the invariant function $\fa$ yield:
\begin{center}\vspace{-12pt}
\begin{tabular}[t]{|l||c|c|c|}
\hline \parbox[l][20pt][c]{0pt}{}   $\alpha$ & 0 & -2 & \\ \hline
\parbox[l][20pt][c]{0pt}{} $\fa\el_{17,9}(\alpha)$ & 19 & 17 & 16\\ \hline
\end{tabular}\quad
\begin{tabular}[t]{|l||c|c|c|}
\hline \parbox[l][20pt][c]{0pt}{}   $\alpha$             & 0  &
$-\frac{1}{2}$ & \\ \hline \parbox[l][20pt][c]{0pt}{}
$\fa\el_{17,12}(\alpha)$ & 19 & 17            & 16\\ \hline
\end{tabular}
\end{center}
and since $\fa\el_{17,9} \neq \fa\el_{17,12}$, we see that the
conclusion $\el_{17,9} \ncong \el_{17,12}$ provides the function
$\fa$, as well as the set of invariants $\inv
[\gd{1,1,1}\el\,\cap\,\gd{0,1,1}\el]$; laborious direct
computation of the isomorphism problem is now unnecessary.
\end{example}

When solving the isomorphism problem for parametric continua of
Lie algebras, the situation is far more challenging.

\begin{example}\label{e1825}
Consider the following graded contraction of the Pauli graded
$\slp(3,\Com)$:
\begin{equation}\label{1825}
\ep_{18,25}(a)=\left(
\begin{smallmatrix}
0 & 0 & 1 & a & 0 & 0 & 0 & 0 \\ 0 & 0 & 1 & 1 & 0 & 0 & 0 & 0 \\
1 & 1 & 0 & 0 & 1 & 0 & 1 & 0 \\ a & 1 & 0 & 0 & 0 & 0 & 0 & 0 \\
0 & 0 & 1 & 0 & 0 & 0 & 0 & 0 \\ 0 & 0 & 0 & 0 & 0 & 0 & 0 & 0 \\
0 & 0 & 1 & 0 & 0 & 0 & 0 & 0 \\ 0 & 0 & 0 & 0 & 0 & 0 & 0 & 0
\end{smallmatrix}
 \right)
\end{equation}
We determine this graded contraction by listing its non--zero
commutation relations in $\Z_3$--labeled basis $( l_{01}, l_{02},
l_{10},  l_{20}, l_{11}, l_{22}, l_{12},l_{21})$: $$
\begin{array}{ll}
\el_{18,25}(a) \quad & [l_{01},l_{10}]=l_{11},\
[l_{01},l_{20}]=-al_{21},\ [l_{02},l_{10}]=l_{12},\
[l_{02},l_{20}]=l_{22},\\
            & [l_{10},l_{11}]=l_{21},\ [l_{10},l_{12}]=l_{22}, \q a\in \Com \\

\end{array}
$$

Infinitely many Lie algebras $\el_{18,25}(a)$ -- the parametric
continuum -- are all indecomposable and nilpotent. It is clear,
that the set $\inv \el$, or a similar finite set of certain
dimensions, can never completely characterize an infinite number
of algebras in $\el_{18,25}(a)$. Moreover, all invariants based on
the trace of the adjoint representation, such as $C_{pq}$
(\ref{Cpq}) and $\chi_i $(\ref{Chi}) are, due to
Theorem~\ref{Engel}, worthless in the case of the nilpotent
parametric continuum. The behaviour of the function $\fa$ turns
out to be quite valuable, otherwise we would be forced to try to
solve the isomorphism problem explicitly.

First, we achieve partial characterization by isolating two points
of $\el_{18,25}(a)$, namely $a=0,-1$, and obtain $$
\begin{array}{llll}
 \inv\el_{18,25}(0) \qquad & (8,4,0)
(8,4,2,0)(2,5,8) & 4 & [21,23,10,14] \\ \inv\el_{18,25}(-1) \qquad
& (8,4,0) (8,4,2,0)(2,5,8) & 4 & [22,22,10,13] \\
\inv\el_{18,25}(a)\quad a\neq 0,-1 \qquad & (8,4,0)
(8,4,2,0)(2,5,8) & 4 & [20,22,10,13] \\
\end{array}
$$ Calculating the dimensions of the associated Jordan algebras we
obtain for all $a\in \Com$:
\begin{align*}
\dim \gd{1,1,-1}\el_{18,25}(a)&=9
\\ \dim
\gd{0,1,-1}\el_{18,25}(a) &=18.
\end{align*}
The invariant function $\fa^0$ is not again of much use:
\begin{center}\vspace{-12pt}
\begin{tabular}[t]{|l||c|c|c|}
\hline \parbox[l][20pt][c]{0pt}{}   $\alpha$  & 0  & 1 & \\ \hline
\parbox[l][20pt][c]{0pt}{} $\fa^0\el_{18,25}(a)(\alpha)$ & 16 & 10 & 9\\ \hline
\end{tabular}
\end{center}
However, the calculation of the invariant function $\fa$ yield:
\begin{center}\vspace{-12pt}
\begin{tabular}[t]{|l||c|c|}
\multicolumn{3}{l}{$a = 0$} \\ \hline \parbox[l][20pt][c]{0pt}{}
$\alpha$  & 0  &  \\ \hline \parbox[l][20pt][c]{0pt}{}
$\fa\el_{18,25}(0)(\alpha)$ & 23 & 21 \\ \hline
\end{tabular}\quad
\begin{tabular}[t]{|l||c|c|c|c|}
\multicolumn{5}{l}{$a=1$} \\ \hline \parbox[l][20pt][c]{0pt}{}
$\alpha$  & 0  & -1 & 1 &\\ \hline \parbox[l][20pt][c]{0pt}{}
$\fa\el_{18,25}(1)(\alpha)$ & 22 & 21 & 20 & 19\\ \hline
\end{tabular}

\vspace{12pt}

\begin{tabular}[t]{|l||c|c|c|}
\multicolumn{4}{l}{$a=-1$} \\ \hline \parbox[l][20pt][c]{0pt}{}
$\alpha$  & 0  & 1 & \\ \hline \parbox[l][20pt][c]{0pt}{}
$\fa\el_{18,25}(-1)(\alpha)$ & 22 & 22 & 19\\ \hline
\end{tabular}

\vspace{12pt}

\begin{tabular}[t]{|l||c|c|c|c|c|}
\multicolumn{6}{l}{$a \neq 0, \pm 1$} \\ \hline
\parbox[l][20pt][c]{0pt}{}   $\alpha$  & 0  & 1 & $-a$ &
$-\frac{1}{a}$ &\\ \hline \parbox[l][20pt][c]{0pt}{}
$\fa\el_{18,25}(a)(\alpha)$ & 22 & 20 & 20 & 20 & 19\\ \hline
\end{tabular}
\end{center}
In that last table ($a \neq 0, \pm 1$) we observe that the
function $\fa$ has the same form for pairs $\el_{18,25}(a)$ and
$\el_{18,25}(1/a)$. This indicates possible isomorphism between
these pairs and, indeed, we verify directly that
$$\el_{18,25}(a)\cong \el_{18,25}(1/a),\q a \neq 0, \pm 1. $$
Finally since for $a, a' \neq 0,\pm1, \ a\neq a',\ a \neq 1/a'$
the inequality
 $$\fa\el_{18,25}(a)(-a)= 20 \neq \fa\el_{18,25}(a')(-a) = 19$$
holds, the relation $\el(a) \ncong \el(a')$ is thus guaranteed.
The invariant function $\fa$ provides us therefore with a complete
characterization of the continuum $\el_{18,25}(a)$.
\end{example}

\begin{example}\label{e1713} Similarly to the previous example, consider the following graded
contraction of the Pauli graded $\slp(3,\Com)$:
\begin{equation}\label{1713}
\ep^{17,13}(a)= \left(
\begin{smallmatrix}
0 & 0 & a & 1 & 0 & 1 & 0 & 0 \\ 0 & 0 & 1 & 0 & 0 & 0 & 0 & 0 \\
a & 1 & 0 & 0 & 1 & 1 & 0 & 0 \\ 1 & 0 & 0 & 0 & 0 & 1 & 0 & 0 \\
0 & 0 & 1 & 0 & 0 & 0 & 0 & 0 \\ 1 & 0 & 1 & 1 & 0 & 0 & 0 & 0 \\
0 & 0 & 0 & 0 & 0 & 0 & 0 & 0 \\ 0 & 0 & 0 & 0 & 0 & 0 & 0 & 0
\end{smallmatrix}
 \right)
\end{equation}
We determine this graded contraction by listing its non--zero
commutation relations in $\Z_3$--labeled basis $( l_{01}, l_{02},
l_{10},  l_{20}, l_{11}, l_{22}, l_{12},l_{21})$:
 $$
\begin{array}{ll}
\el_{17,13}(a) \quad & [l_{01},l_{10}]=al_{11},\
[l_{01},l_{20}]=l_{21},\ [l_{01},l_{22}]=l_{20},\
[l_{02},l_{10}]=l_{12},\\
  & [l_{10},l_{11}]=l_{21},\ [l_{10},l_{22}]=l_{02},\ [l_{20},l_{22}]=l_{12},\q a\neq 0 \\
\end{array}
$$ Lie algebras $\el_{17,13}(a)$ are all indecomposable and
nilpotent. Isolating one point, $a=-1$ we obtain $$
\begin{array}{llll}
\inv\el_{17,13}(-1) \qquad & (8,5,0) (8,5,2,0)(2,5,8) & 4 &
[19,19,8,9]
\\ \inv\el_{17,13}(a)\q a\neq 0,-1 \qquad & (8,5,0) (8,5,2,0)(2,5,8) & 4 &
[17,19,8,9] \\
\end{array}
$$ Calculating the dimensions of the associated Jordan algebras we
obtain for all $a\in \Com$:
\begin{align*}
\dim \gd{1,1,-1}\el_{17,13}(a)&=7\\
\dim\gd{0,1,-1}\el_{17,13}(a)&=18.
\end{align*}
The invariant function $\fa^0$ has the form
\begin{center}\vspace{-12pt}
\begin{tabular}[t]{|l||c|c|c|}
\hline \parbox[l][20pt][c]{0pt}{}   $\alpha$  & 0  & 1 & \\ \hline
\parbox[l][20pt][c]{0pt}{} $\fa^0\el_{17,13}(a)(\alpha)$ & 16 & 8 & 7\\ \hline
\end{tabular}
\end{center}
and the invariant function $\fa$ has the form

\begin{center}\vspace{-12pt}
\begin{tabular}[t]{|l||c|c|c|c|}
\multicolumn{5}{l}{$a=1$} \\ \hline \parbox[l][20pt][c]{0pt}{}
$\alpha$  & -1  & 0 & 1 &\\ \hline \parbox[l][20pt][c]{0pt}{}
$\fa\el_{17,13}(1)(\alpha)$ & 19 & 19 & 17 & 16\\ \hline
\end{tabular}\qquad
\begin{tabular}[t]{|l||c|c|c|c|}
\multicolumn{5}{l}{$a=-1$} \\ \hline \parbox[l][20pt][c]{0pt}{}
$\alpha$  & 0 & 1 & -1 &\\ \hline \parbox[l][20pt][c]{0pt}{}
$\fa\el_{17,13}(-1)(\alpha)$ & 19 & 19 & 17 & 16\\ \hline
\end{tabular}

\vspace{12pt}

\begin{tabular}[t]{|l||c|c|c|c|c|c|}
\multicolumn{7}{l}{$a \neq \pm1$} \\ \hline
\parbox[l][20pt][c]{0pt}{}   $\alpha$  & 0 & 1 & $-1$ & $-a$ &
$-\frac{1}{a}$ & \\ \hline \parbox[l][20pt][c]{0pt}{}
$\fa\el_{17,13}(a)(\alpha)$ & 19 & 17 & 17 & 17 & 17 & 16\\ \hline
\end{tabular}
\end{center}
After verifying $\el_{17,13}(a)\cong\el_{17,13}(1/a),\,a\neq 0,\pm
1$, we conclude that the function $\fa$ again represents a
priceless instrument providing complete description of presented
parametric continuum of Lie algebras.
\end{example}

\begin{example}\label{e177} Consider the following graded
contraction of the Pauli graded $\slp(3,\Com)$:
\begin{equation}\label{177}
\ep^{17,7}(a)=
 \left(
\begin{smallmatrix}
0 & 0 & a & 1 & 1 & 1 & 0 & 0 \\ 0 & 0 & 1 & 0 & 0 & 1 & 0 & 0 \\
a & 1 & 0 & 0 & 1 & 0 & 0 & 0 \\ 1 & 0 & 0 & 0 & 0 & 0 & 0 & 0 \\
1 & 0 & 1 & 0 & 0 & 0 & 0 & 0 \\ 1 & 1 & 0 & 0 & 0 & 0 & 0 & 0 \\
0 & 0 & 0 & 0 & 0 & 0 & 0 & 0 \\ 0 & 0 & 0 & 0 & 0 & 0 & 0 & 0
\end{smallmatrix}
 \right),
\end{equation}
where $a\neq 0$. We may also determine this graded contraction by
listing its non--zero commutation relations in $\Z_3$--labeled
basis $( l_{01}, l_{02}, l_{10},  l_{20}, l_{11}, l_{22},
l_{12},l_{21})$: $$
\begin{array}{ll}
\el_{17,7}(a): \quad & [l_{01},l_{10}]=-al_{11},\
[l_{01},l_{20}]=l_{21},\ [l_{01},l_{11}]=l_{12},\
[l_{01},l_{22}]=l_{20},\\
  & [l_{02},l_{10}]=l_{12},\ [l_{02},l_{22}]=l_{21},\ [l_{10},l_{11}]=l_{21},\q a\neq 0 .\\
\end{array} $$ Lie algebras $\el_{17,7}(a)$ are all
indecomposable and nilpotent. We obtain for all $a\in \Com,\,
a\neq 0$: $$
\begin{array}{llll}
\inv\el_{17,7}(a) \qquad & (8,4,0) (8,4,2,0)(2,5,8) & 2 &
[19,20,9,12]
\end{array}
$$ Computing the dimensions of the associated Jordan algebras we
obtain for all $a\in \Com$:
\begin{align*}
\dim \gd{1,1,-1}\el_{17,7}(a)&=8\\
\dim\gd{0,1,-1}\el_{17,7}(a)&=17.
\end{align*}
The invariant functions $\fa^0$ and $\fa$ have the following form:
\begin{center}\vspace{-12pt}
\begin{tabular}[t]{|l||c|c|c|}
\hline \parbox[l][20pt][c]{0pt}{}   $\alpha$  & 0  & 1 & \\ \hline
\parbox[l][20pt][c]{0pt}{} $\fa^0\el_{17,7}(a)(\alpha)$ & 16 & 9 & 8\\ \hline
\end{tabular}
\qquad
\begin{tabular}[t]{|l||c|c|c|}
\hline \parbox[l][20pt][c]{0pt}{}   $\alpha$  & 0  & 1 & \\ \hline
\parbox[l][20pt][c]{0pt}{} $\fa\el_{17,7}(a)(\alpha)$ & 20 & 19 & 18\\ \hline
\end{tabular}
\end{center}
In this case, the function $\fa$ completely fails  -- does not
depend on $a\neq 0$. We are able, however, to advance by
calculation of the function $\fb$:
\begin{center}\vspace{-12pt}
\begin{tabular}[t]{|l||c|c|c|c|}
\multicolumn{5}{l}{$a=1$} \\ \hline \parbox[l][20pt][c]{0pt}{}
$\alpha$  & 0 & 1 & -1 &\\ \hline \parbox[l][20pt][c]{0pt}{}
$\fb\el_{17,7}(1)(\alpha)$ & 112 & 83 & 81 & 80\\ \hline
\end{tabular}\qquad
\begin{tabular}[t]{|l||c|c|c|c|}
\multicolumn{5}{l}{$a=-1$} \\ \hline \parbox[l][20pt][c]{0pt}{}
$\alpha$  & 0 & 1 & -1 &\\ \hline \parbox[l][20pt][c]{0pt}{}
$\fb\el_{17,7}(-1)(\alpha)$ & 104 & 83 & 81 & 80\\ \hline
\end{tabular}

\vspace{12pt}

\begin{tabular}[t]{|l||c|c|c|c|}
\multicolumn{5}{l}{$a=\frac{1}{4}+\frac{\sqrt{7}}{4}i$}\parbox[l][20pt][c]{0pt}{}
\\ \hline
\parbox[l][26pt][c]{0pt}{} $\alpha$  & 0 & 1 & $-\frac{1}{4}-\frac{\sqrt{7}}{4}i$ &\\ \hline
\parbox[l][26pt][c]{0pt}{} $\fb\el_{17,7}\left(\frac{1}{4}+\frac{\sqrt{7}}{4}i\right)(\alpha)$ & 104 & 82 &
82 & 80\\ \hline
\end{tabular}

\vspace{12pt}

\begin{tabular}[t]{|l||c|c|c|c|}
\multicolumn{5}{l}{$a=\frac{1}{4}-\frac{\sqrt{7}}{4}i$}\parbox[l][20pt][c]{0pt}{}
\\ \hline
\parbox[l][26pt][c]{0pt}{} $\alpha$  & 0 & 1 & $-\frac{1}{4}+\frac{\sqrt{7}}{4}i$ &\\ \hline
\parbox[l][26pt][c]{0pt}{} $\fb\el_{17,7}\left(\frac{1}{4}-\frac{\sqrt{7}}{4}i\right)(\alpha)$ & 104 & 82 &
82 & 80\\ \hline
\end{tabular}

\vspace{12pt}

\begin{tabular}[t]{|l||c|c|c|c|}
\multicolumn{5}{l}{$a=\frac{1}{3}$}\parbox[l][20pt][c]{0pt}{}
\\ \hline
\parbox[l][26pt][c]{0pt}{} $\alpha$  & 0 & 1 & $-\frac{1}{3}$ &\\ \hline
\parbox[l][26pt][c]{0pt}{} $\fb\el_{17,7}\left(\frac{1}{3}\right)(\alpha)$ & 104 & 83 &
81 & 80\\ \hline
\end{tabular}

\vspace{12pt}

\begin{tabular}[t]{|l||c|c|c|c|c|}
\multicolumn{6}{l}{$a \neq
0,\pm1,\frac{1}{3},\frac{1}{4}\pm\frac{\sqrt{7}}{4}i
$}\parbox[l][20pt][c]{0pt}{} \\ \hline
\parbox[l][20pt][c]{0pt}{}   $\alpha$  & 0 & 1 & $-a$ & $-\frac{1}{2}+\frac{1}{2a}$ & \\
  \hline \parbox[l][20pt][c]{0pt}{}
$\fb\el_{17,7}(a)(\alpha)$ & 104 & 82 & 81 & 81 & 80 \\ \hline
\end{tabular}
\end{center}
In order to verify that
\begin{equation}\label{equ177}
\fb\el_{17,7}(a):\,  104_1,82_1,81_2,80, \q a \neq
0,\pm1,\frac{1}{3},\frac{1}{4}\pm\frac{\sqrt{7}}{4}i
\end{equation}
we have to check the equality
\begin{equation*}
-a=-\frac{1}{2}+\frac{1}{2a}
\end{equation*}
which has the solutions $\frac{1}{4}\pm\frac{\sqrt{7}}{4}i$. Thus,
(\ref{equ177}) is verified.

We proceed to solve the relation
$$\fb\el_{17,7}(a)=\fb\el_{17,7}(a'), \q a,a' \neq
0,\pm1,\frac{1}{3},\frac{1}{4}\pm\frac{\sqrt{7}}{4}i $$ and we
obtain
\begin{enumerate}
\item If $-a=-a',\,
-\frac{1}{2}+\frac{1}{2a}=\frac{1}{2}+\frac{1}{2a'}$ then $a=a'$.
\item If $-a=-\frac{1}{2}+\frac{1}{2a'},\,-a'=-\frac{1}{2}+\frac{1}{2a}
$ then $a=a'=\frac{1}{4}\pm\frac{\sqrt{7}}{4}i$.
\end{enumerate}
The second case is not possible. Observing that all other tables
of the function $$\fb\el_{17,7}(a),\, a=
\pm1,\frac{1}{3},\frac{1}{4}\pm\frac{\sqrt{7}}{4}i $$ are mutually
different, we have: if $\fb\el_{17,7}(a)=\fb\el_{17,7}(a'), \,
a,a' \neq 0$ then $a=a'$. We conclude that even though the
function $\fa$ did not distinguish the algebras in nilpotent
parametric continuum, the function $\fb$ provided their complete
description.
\end{example}

A further analysis of the parametric continua resulting from the
Pauli graded $\slp(3,\Com)$ as well as the analysis of all graded
contractions of $\slp(3,\Com)$ will be done elsewhere~\cite{Nov}.
We summarize the results from Examples \ref{e1825}, \ref{e1713}
and \ref{e177} into the following theorem.
\begin{thm}\label{zvast}
Let  $\el_{18,25}(a)$, $\el_{17,13}(a)$ and $\el_{17,7}(a)$,
$a\neq 0$ be graded contractions of the Pauli graded
$\slp(3,\Com)$, defined by the contraction matrices (\ref{1825}),
(\ref{1713}) and (\ref{177}), respectively. Then it holds:
\begin{enumerate}
\item $\el_{18,25}(a)\cong \el_{18,25}(a') $ if and only if
$a'=a,\,1/a$.
\item $\el_{17,13}(a)\cong \el_{17,13}(a') $ if and only if
$a'=a,\,1/a$.
\item $\el_{17,7}(a)\cong \el_{17,7}(a') $ if and only if $a'=a$.
\end{enumerate}
\end{thm}

\nonumchapter{Conclusion}

In Chapter~\ref{CHgen} we have introduced a new concept -- the
$(\alpha,\beta,\gamma)$--derivations and presented results related
to their structure and significance as invariants. There exist,
however, several non-equivalent ways of generalizing the
derivation of a Lie algebra. For example in \cite{Hartwig}, a
linear operator $A \in \End\el$ is called a
$(\sigma,\tau)$--derivation of $\el$ if for some $\sigma, \tau \in
\End\el$ the property $ A[x,y] = [Ax,\tau y] + [\sigma x,Ay]$
holds for all $x,y \in \el$. This generalization for $\sigma,\tau$
homomorphisms appears already in \cite{Jacobson}. If there exists
$B\in \der\el$ such that for all $x,y \in \el$ the condition $
A[x,y] = [Ax,y] + [x,By]$ holds, then the operator $A$ forms
another generalization \cite{Bresar}. A more general definition
emerged in \cite{Leger} and runs as follows: $A \in \End\el$ is
called a {\it generalized derivation} of $\el$ if there exist
$B,C\in \End\el$ such that for all $x,y \in \el$ the property $
C[x,y] = [Ax,y] + [x,By]$ holds.

We see that various definitions of generalized derivations were
formulated with the aid of some operators inserted into the
equation $D[x,y]=[D x,y]+[x,D y]$. But for the 'invariance' of the
definition (\ref{gd}), i.~e. the validity of the equation
(\ref{invar}), it is essential that these operators commute with
an arbitrary isomorphism of the Lie algebra. Thus, we have chosen
in (\ref{gd}), from this point of view, a 'maximal' set of
operators, i.~e. multiples of the identity operator.

Among the results of the concept of
$(\alpha,\beta,\gamma)$--derivations are 'associated' Lie and
Jordan algebras. Jordan algebra $\gd{1,1,-1}\el$ and Lie algebras
$\gd{0,1,1}\el$, $ \gd{1,1,1}\el \cap \gd{0,1,1}\el$ seem to be
unnoticed in the existing literature. We demonstrated in
Proposition~\ref{indepass} and Example~\ref{intersect} that their
dimension and structure are a really useful addition to the
knowledge about given Lie algebra. The sets $\gd{1,1,0}{\el}$ and
$\gd{0,1,-1}{\el}$ appeared in \cite{Leger} and are called {\it
centroid} and {\it quasicentroid}, respectively. The inquiry under
which conditions these sets coincide has also been discussed
in~\cite{Leger}.

In Chapter~\ref{CHtwi} we have defined the set of
$\kappa$--twisted cocycles $Z^q(\el,f,\kappa)$ for an arbitrary
representation. Then we investigated in detail the case
$f=\ad_\el$. This study provided a complete description of
four--dimensional complex Lie algebras and the algorithm for their
identification. There are, however, two other obvious choices for
the representation $f$: either the representation $\ad_\el^*$, or
the trivial representation $f:\el\map\Com$. Another option is to
define generalized prederivations (see e. g. \cite{Bur4}) -- i.~e.
insert four parameters into the equation $$P[x,[y,z]]=[P
x,[y,z]]+[x,[P y, z]]+ [x,[y,P z]],\q \forall x,y,z\in \el. $$
None of these possibilities turned out to be very fruitful,
therefore we did not investigate such generalizations in detail.
Note that according to Proposition~\ref{Albertprop} it is
pointless to modify by parameters the operator Lie multiplication
$XY-YX$, where $X,Y\in \End\el$.

In Chapter~\ref{CHcon} we have applied the concept of twisted
cocycles to formulate a necessary criterion for the existence of
continuous contraction. We demonstrated that
Corollary~\ref{concritmain} decides about the existence of a
continuous contraction for three--dimensional Lie and
two--dimensional Jordan algebras. In Examples~\ref{e1825},
\ref{e1713} and \ref{e177}, we have successfully applied the
invariant functions to the eight--dimensional graded contractions
of the Pauli graded~$\slp(3,\Com)$.

Resolving parametric continua of Lie algebras belongs to the most
challenging parts of their classification. Except for the explicit
calculation, all known tools for such resolving were based on the
adjoint representation. 'Trace' invariants $\chi_i$ and $C_{pq}$,
defined by relations~(\ref{Chi}), (\ref{Cpq}) and able to resolve
a parametric continuum, have been successfully used in
\cite{AY,Bur1,Nes}. Advantages of these invariants are easy
calculation and invariance under continuous contractions. One may
also consider action of the adjoint representation on the
nilradical of a Lie algebra -- considering eigenvalues of this
action one may also resolve a parametric continuum. All these
approaches, however, fail in the case of a nilpotent continuum  --
due to Theorem~\ref{Engel} none of the invariants $\chi_i$ nor
$C_{pq}$ exists, eigenvalues of the adjoint representation are all
zeros. We have seen in Examples~\ref{e1825}, \ref{e1713} and
\ref{e177} that such cases may resolve our new invariant functions
$\fa,\fb$ -- their form is non--trivial and depends on the
parameter of a nilpotent continuum. The main idea behind the
present work -- the classification of all graded contractions of
$\slp (3,\Com)$ -- can now be easier to achieve.

It is computationally advantageous that determining associated Lie
algebras is a linear problem: in order to determine the space
$\gd{\alpha,\beta,\gamma}\el$ for fixed
$\alpha,\beta,\gamma\in\Com$, one has to solve homogeneous system
of linear equations~(\ref{sys11}). The investigation how the
dimension of the vector space $\gd{\delta,1,1}\el$ depends on
$\delta$ -- computing the function $\fa$ -- is more challenging.
One has to analyze the rank of the $\delta$--dependent matrix
corresponding to the linear system~(\ref{sys11}). Determining the
rank of parametric matrices is called a {\it specialization
problem}. There has been some progress concerning the application
of various computational algorithms to parametric matrices. The
results, however, are unsatisfactory and computation of the rank
of parametric matrices is not implemented in standard symbolic
mathematical tools -- MAPLE, Mathematica. It was therefore
necessary to develop a new algorithm. This new algorithm is based
on the standard Gaussian elimination with row pivoting. The given
column is at first searched for a non--zero complex number. If
this number is not found, then the $\delta$--dependent element
with minimal length is assumed to be non--zero and this assumption
is added to some set -- the set of assumptions. The resulting set
of assumptions is analyzed and solved for the equality of its
elements to zero. For these solutions, special values of $\delta$,
the standard Gaussian elimination is performed again. This
algorithm, implemented in MAPLE VIII, turned out to be sufficient
for our purpose. However, it is in some cases quite
computationally demanding. Namely, computing Examples~\ref{AYE1},
\ref{AYE2} from~\cite{AY} took a significant amount of
computational time. Even more demanding is the computing the
functions $\fb,\fc$ for eight--dimensional Lie algebras -- the
$\delta$--dependent matrices have $224$ columns and contain
approximately two thousand rows.

Thus, allowing a compact formulation in Theorem~\ref{class4dim},
the significance of the functions $\fa, \fb$ in dimension four is
more theoretical than computationally advantageous. The invariants
$\chi_i$ provide sufficient description of parametric
four--dimensional Lie algebras. On the other hand, the knowledge
about the dimensions of highly non--trivial structures, interlaced
with given Lie algebra, may also be very valuable. Considering
parametric continua of nilpotent algebras the invariant functions
$\fa,\fb,\fc$ seem to be even more important. In these cases, the
behaviour of the invariant functions
$\fa,\fb,\fc$ is quite unique and irreplaceable.

\appendix
\chapter{Classification of $(\alpha,\beta,\gamma)$--derivations of two and three--dimensional
complex Lie and Jordan Algebras}\label{APA} Appendix \ref{APA} is
divided into two sections. Section~\ref{APAL} contains matrices of
$(\alpha,\beta,\gamma)$--derivations of two and three--dimensional
non--abelian complex Lie algebras. Section \ref{APAJ} contains
matrices of $(\alpha,\beta,\gamma)$--derivations of
two--dimensional non--abelian complex Jordan algebras. Instead of
symbols $\gd{\alpha,\beta,\gamma}\A$, abbreviated symbols
$\gd{\alpha,\beta,\gamma}$ are used.
\section{Lie Algebras}\label{APAL}
\subsection*{Two--dimensional Complex Lie Algebras}

%
\right) \right\}} $ for $\delta \neq 0$. \end{itemize}
\subsection*{Three-dimensional Complex Lie Algebras}
\begin{center}%
\right)
 \right\}}$
\end{itemize}
\section{Jordan Algebras}\label{APAJ}

\subsection*{Two--dimensional Complex Jordan Algebras}
\begin{center}
%
%
\right) \right\}} $
\end{itemize}

\chapter{Invariant Functions of complex Lie and Jordan Algebras of dimensions $2,\,3$ and $4$}\label{IFLJ}
\vspace{-12pt} Appendix \ref{IFLJ} is divided into two sections.
Section \ref{IFLJL} contains the classification of complex Lie
algebras up to dimension four and the invariant functions
$\fa,\fb,\fc$. We basically follow the notation of~\cite{Nes} and
present a list connecting it to other notations.
Section~\ref{IFLJJ} contains the classification of one and
two--dimensional complex Jordan algebras \cite{jordanesp} and the
tables of the invariant function $\fa$. Instead of the symbols
$\fa\A,\fb\A,\fc\A$, abbreviated symbols $\fa,\fb,\fc$ are used.
Blank spaces in the tables of the functions $\fa,\fb,\fc$ denote
general complex numbers, different from all previously listed
values, e. g. it holds: $$\fa\g_{3,4}(-1)(\al)=3,\q\alpha\in
\Com,\,\al\neq\pm 1.$$

\section{Lie Algebras}\label{IFLJL}
\subsection*{Two--dimensional Complex Lie Algebras}
\hspace{16pt}
\begin{tabular}{ll}
$2\g_{1}: \quad$ & Abelian
\end{tabular}
\begin{center}
\begin{tabular}[t]{|l||c|c|}
\hline \parbox[l][20pt][c]{0pt}{}   $\alpha$ & 1 &  \\ \hline
\parbox[l][20pt][c]{0pt}{} $\fa(\alpha)$ & 4 & 4  \\ \hline
\end{tabular}\qquad
\begin{tabular}[t]{|l||c|c|c|}
\hline  \parbox[l][20pt][c]{0pt}{}  $\alpha$ &   \\ \hline
\parbox[l][20pt][c]{0pt}{} $\fb(\alpha)$ & 2\\
\hline
\end{tabular}\qquad
\begin{tabular}[t]{|l||c|c|c|}
\hline  \parbox[l][20pt][c]{0pt}{}  $\alpha$ &   \\ \hline
\parbox[l][20pt][c]{0pt}{} $\fc(\alpha)$ & 2\\
\hline
\end{tabular}
\end{center}

\begin{tabular}{ll}
$\g_{2,1}: \quad$ & $ [e_1,e_2]=e_1$
\end{tabular}

\begin{center}
\begin{tabular}[t]{|l||c|c|c|}
\hline \parbox[l][20pt][c]{0pt}{}   $\alpha$ & 1 & 0 &  \\ \hline
\parbox[l][20pt][c]{0pt}{} $\fa(\alpha)$ & 2 & 3& 2  \\ \hline
\end{tabular}\qquad
\begin{tabular}[t]{|l||c|c|c|}
\hline  \parbox[l][20pt][c]{0pt}{}  $\alpha$ &   \\ \hline
\parbox[l][20pt][c]{0pt}{} $\fb(\alpha)$ & 2\\
\hline
\end{tabular}\qquad
\begin{tabular}[t]{|l||c|c|c|}
\hline  \parbox[l][20pt][c]{0pt}{}  $\alpha$ & 2 &  \\ \hline
\parbox[l][20pt][c]{0pt}{} $\fc(\alpha)$ & 1& 0\\
\hline
\end{tabular}
\end{center}

\newpage
\subsection*{Three--dimensional Complex Lie Algebras}
\hspace{16pt}
\begin{tabular}{ll}
$3\g_{1}: \quad$ & Abelian
\end{tabular}
\begin{center}
\begin{tabular}[t]{|l||c|c|}
\hline \parbox[l][20pt][c]{0pt}{}   $\alpha$ & 1 &  \\ \hline
\parbox[l][20pt][c]{0pt}{} $\fa(\alpha)$ & 9 & 9  \\ \hline
\end{tabular}\qquad
\begin{tabular}[t]{|l||c|c|c|}
\hline  \parbox[l][20pt][c]{0pt}{}  $\alpha$ &   \\ \hline
\parbox[l][20pt][c]{0pt}{} $\fb(\alpha)$ & 9\\
\hline
\end{tabular}\qquad
\begin{tabular}[t]{|l||c|c|c|}
\hline  \parbox[l][20pt][c]{0pt}{}  $\alpha$ &   \\ \hline
\parbox[l][20pt][c]{0pt}{} $\fc(\alpha)$ & 9\\
\hline
\end{tabular}
\end{center}

\begin{tabular}{ll}
$\g_{2,1}\oplus \g_1: \quad$ & $ [e_1,e_2]=e_2$
\end{tabular}

\begin{center}
\begin{tabular}[t]{|l||c|c|c|}
\hline \parbox[l][20pt][c]{0pt}{}   $\alpha$ & 1 & 0 &  \\ \hline
\parbox[l][20pt][c]{0pt}{} $\fa(\alpha)$ & 4 & 6 & 4  \\ \hline
\end{tabular}\qquad
\begin{tabular}[t]{|l||c|c|c|}
\hline  \parbox[l][20pt][c]{0pt}{}  $\alpha$ &   \\ \hline
\parbox[l][20pt][c]{0pt}{} $\fb(\alpha)$ & 6\\
\hline
\end{tabular}\qquad
\begin{tabular}[t]{|l||c|c|c|c|}
\hline  \parbox[l][20pt][c]{0pt}{}  $\alpha$ & 1& 2&  \\ \hline
\parbox[l][20pt][c]{0pt}{} $\fc(\alpha)$ &2&2&1\\
\hline
\end{tabular}
\end{center}

\begin{tabular}{ll} $\g_{3,1}: \quad$ &$[e_2,e_3]=e_1$
\end{tabular}
\begin{center}
\begin{tabular}[t]{|l||c|c|c|}
\hline \parbox[l][20pt][c]{0pt}{}   $\alpha$ & 1 &   \\ \hline
\parbox[l][20pt][c]{0pt}{} $\fa(\alpha)$ & 6 & 6   \\ \hline
\end{tabular}\qquad
\begin{tabular}[t]{|l||c|c|c|}
\hline  \parbox[l][20pt][c]{0pt}{}  $\alpha$ & 0&  \\ \hline
\parbox[l][20pt][c]{0pt}{} $\fb(\alpha)$ &9& 8\\
\hline
\end{tabular}\qquad
\begin{tabular}[t]{|l||c|c|c|c|}
\hline  \parbox[l][20pt][c]{0pt}{}  $\alpha$ &   \\ \hline
\parbox[l][20pt][c]{0pt}{} $\fc(\alpha)$ &3\\
\hline
\end{tabular}
\end{center}

\begin{tabular}{ll} $\g_{3,2}: \quad$ &$[e_1,e_3]=e_1,\
[e_2,e_3]=e_1+e_2$
\end{tabular}
\begin{center}
\begin{tabular}[t]{|l||c|c|c|}
\hline \parbox[l][20pt][c]{0pt}{}   $\alpha$ & 1 &   \\ \hline
\parbox[l][20pt][c]{0pt}{} $\fa(\alpha)$ & 4 & 3   \\ \hline
\end{tabular}\qquad
\begin{tabular}[t]{|l||c|c|c|}
\hline  \parbox[l][20pt][c]{0pt}{}  $\alpha$ &  \\ \hline
\parbox[l][20pt][c]{0pt}{} $\fb(\alpha)$ &6\\
\hline
\end{tabular}\qquad
\begin{tabular}[t]{|l||c|c|c|c|}
\hline  \parbox[l][20pt][c]{0pt}{}  $\alpha$ & 2 &  \\ \hline
\parbox[l][20pt][c]{0pt}{} $\fc(\alpha)$ &2&0\\
\hline
\end{tabular}
\end{center}

\begin{tabular}{ll} $\g_{3,3}: \quad$ &$[e_1,e_3]=e_1,\
[e_2,e_3]=e_2$
\end{tabular}
\begin{center}
\begin{tabular}[t]{|l||c|c|c|}
\hline \parbox[l][20pt][c]{0pt}{}   $\alpha$ & 1 &   \\ \hline
\parbox[l][20pt][c]{0pt}{} $\fa(\alpha)$ & 6 & 3   \\ \hline
\end{tabular}\qquad
\begin{tabular}[t]{|l||c|c|c|}
\hline  \parbox[l][20pt][c]{0pt}{}  $\alpha$ &   \\ \hline
\parbox[l][20pt][c]{0pt}{} $\fb(\alpha)$ &6\\
\hline
\end{tabular}\qquad
\begin{tabular}[t]{|l||c|c|c|c|}
\hline  \parbox[l][20pt][c]{0pt}{}  $\alpha$ & 2&   \\ \hline
\parbox[l][20pt][c]{0pt}{} $\fc(\alpha)$ &6& 0\\
\hline
\end{tabular}
\end{center}

\begin{tabular}{ll} $\g_{3,4}(-1): \quad$ &$[e_1,e_3]=e_1,\
[e_2,e_3]=-e_2$
\end{tabular}
\begin{center}
\begin{tabular}[t]{|l||c|c|c|}
\hline \parbox[l][20pt][c]{0pt}{}   $\alpha$ & 1 & -1 &   \\
\hline
\parbox[l][20pt][c]{0pt}{} $\fa(\alpha)$ & 4 & 5 &3   \\ \hline
\end{tabular}\qquad
\begin{tabular}[t]{|l||c|c|c|}
\hline  \parbox[l][20pt][c]{0pt}{}  $\alpha$ & 0&  \\ \hline
\parbox[l][20pt][c]{0pt}{} $\fb(\alpha)$ &9& 7\\
\hline
\end{tabular}\qquad
\begin{tabular}[t]{|l||c|c|c|c|}
\hline  \parbox[l][20pt][c]{0pt}{}  $\alpha$ & 0&2&  \\ \hline
\parbox[l][20pt][c]{0pt}{} $\fc(\alpha)$ &2&2&0\\
\hline
\end{tabular}
\end{center}

\begin{tabular}{ll} $\g_{3,4}(a): \quad$ &$[e_1,e_3]=e_1,\
[e_2,e_3]=ae_2,\ a\neq 0,\pm 1$
\end{tabular}
\begin{center}
\begin{tabular}[t]{|l||c|c|c|c|}
\hline \parbox[l][20pt][c]{0pt}{}   $\alpha$ & 1 & $a$
&$\frac{1}{a}$&
\\ \hline
\parbox[l][20pt][c]{0pt}{} $\fa(\alpha)$ & 4 & 4& 4&3   \\ \hline
\end{tabular}\qquad
\begin{tabular}[t]{|l||c|c|c|}
\hline  \parbox[l][20pt][c]{0pt}{}  $\alpha$ &   \\ \hline
\parbox[l][20pt][c]{0pt}{} $\fb(\alpha)$ &6\\
\hline
\end{tabular}\qquad
\begin{tabular}[t]{|l||c|c|c|c|}
\hline  \parbox[l][20pt][c]{0pt}{}  $\alpha$ & 2&$1+a$ &
$1+\frac{1}{a}$&
\\ \hline
\parbox[l][20pt][c]{0pt}{} $\fc(\alpha)$ &2&1&1&0\\
\hline
\end{tabular}
\end{center}

\begin{tabular}{ll}$\slp (2,\Com): \quad$ &$[e_1,e_2]=e_1,\
[e_2,e_3]=e_3,\ [e_1,e_3]=2e_2$
\end{tabular}
\begin{center}
\begin{tabular}[t]{|l||c|c|c|c|}
\hline \parbox[l][20pt][c]{0pt}{}   $\alpha$ & 1 & $-1$& 2&   \\
\hline
\parbox[l][20pt][c]{0pt}{} $\fa(\alpha)$ & 3 & 5& 1&0   \\ \hline
\end{tabular}\qquad
\begin{tabular}[t]{|l||c|c|c|}
\hline  \parbox[l][20pt][c]{0pt}{}  $\alpha$ & 0&  \\ \hline
\parbox[l][20pt][c]{0pt}{} $\fb(\alpha)$ &9& 6\\
\hline
\end{tabular}\qquad
\begin{tabular}[t]{|l||c|c|c|c|}
\hline  \parbox[l][20pt][c]{0pt}{}  $\alpha$ & 2&   \\ \hline
\parbox[l][20pt][c]{0pt}{} $\fc(\alpha)$ &1&0\\
\hline
\end{tabular}
\end{center}
\subsection*{Four--dimensional Complex Lie Algebras}
\begin{enumerate}[(g-1)]
\item \begin{tabular}{ll} $4\g_{1}: \quad$ & Abelian
\end{tabular}
\vspace{-16pt}
\begin{center}
\begin{tabular}[t]{|l||c|c|}
\hline \parbox[l][20pt][c]{0pt}{}   $\alpha$ & 1 &  \\ \hline
\parbox[l][20pt][c]{0pt}{} $\fa(\alpha)$ & 16 & 16  \\ \hline
\end{tabular}\qquad
\begin{tabular}[t]{|l||c|c|c|}
\hline  \parbox[l][20pt][c]{0pt}{}  $\alpha$ &   \\ \hline
\parbox[l][20pt][c]{0pt}{} $\fb(\alpha)$ & 24\\
\hline
\end{tabular}\qquad
\begin{tabular}[t]{|l||c|c|c|}
\hline  \parbox[l][20pt][c]{0pt}{}  $\alpha$ &   \\ \hline
\parbox[l][20pt][c]{0pt}{} $\fc(\alpha)$ & 24\\
\hline
\end{tabular}
\end{center}
\item \begin{tabular}{ll} $\g_{2,1}\oplus 2\g_1: \quad$ & $
[e_1,e_2]=e_1$
\end{tabular}
\vspace{-16pt}\begin{center}
\begin{tabular}[t]{|l||c|c|c|c|}
\hline \parbox[l][20pt][c]{0pt}{}   $\alpha$ & 1 & $0$&
\\ \hline
\parbox[l][20pt][c]{0pt}{} $\fa(\alpha)$ & 8 & 11 & 8  \\ \hline
\end{tabular}\qquad
\begin{tabular}[t]{|l||c|c|c|}
\hline  \parbox[l][20pt][c]{0pt}{}  $\alpha$ & $0$  &
\\ \hline
\parbox[l][20pt][c]{0pt}{} $\fb(\alpha)$ & 16  & 14 \\
\hline
\end{tabular}\qquad
\begin{tabular}[t]{|l||c|c|c|c|}
\hline \parbox[l][20pt][c]{0pt}{}   $\alpha$ & 1 & $2$ &
\\ \hline
\parbox[l][20pt][c]{0pt}{} $\fc(\alpha)$ & 8 & 7 & 6 \\ \hline
\end{tabular}
\end{center}
\item\label{g21og21}\begin{tabular}{ll} $\g_{2,1}\oplus \g_{2,1}: \quad$ & $
[e_1,e_2]=e_1,\ [e_3,e_4]=e_3$
\end{tabular}
\vspace{-16pt}\begin{center}
\begin{tabular}[t]{|l||c|c|c|c|}
\hline \parbox[l][20pt][c]{0pt}{}   $\alpha$ & 1 & $0$&
\\ \hline
\parbox[l][20pt][c]{0pt}{} $\fa(\alpha)$ & 4 & 6 & 4  \\ \hline
\end{tabular}\qquad
\begin{tabular}[t]{|l||c|c|c|}
\hline  \parbox[l][20pt][c]{0pt}{}  $\alpha$ & $0$  & 1&
\\ \hline
\parbox[l][20pt][c]{0pt}{} $\fb(\alpha)$ & 12  & 12 &10 \\
\hline
\end{tabular}\qquad
\begin{tabular}[t]{|l||c|c|c|c|}
\hline \parbox[l][20pt][c]{0pt}{}   $\alpha$ & 1 & $2$ &
\\ \hline
\parbox[l][20pt][c]{0pt}{} $\fc(\alpha)$ & 2 & 2 & 0 \\ \hline
\end{tabular}
\end{center}

\item\begin{tabular}{ll} $\g_{3,1}\oplus \g_1: \quad$ &$[e_2,e_3]=e_1$
\end{tabular}
\vspace{-16pt}\begin{center}
\begin{tabular}[t]{|l||c|c|c|c|}
\hline \parbox[l][20pt][c]{0pt}{}   $\alpha$ & 1 & $0$&
\\ \hline
\parbox[l][20pt][c]{0pt}{} $\fa(\alpha)$ & 10 & 11 & 10  \\ \hline
\end{tabular}\qquad
\begin{tabular}[t]{|l||c|c|c|}
\hline  \parbox[l][20pt][c]{0pt}{}  $\alpha$ & $0$  &
\\ \hline
\parbox[l][20pt][c]{0pt}{} $\fb(\alpha)$ & 20  & 19 \\
\hline
\end{tabular}\qquad
\begin{tabular}[t]{|l||c|c|c|c|}
\hline \parbox[l][20pt][c]{0pt}{}   $\alpha$ & 1 &
\\ \hline
\parbox[l][20pt][c]{0pt}{} $\fc(\alpha)$ & 11 & 8 \\ \hline
\end{tabular}
\end{center}

\item\begin{tabular}{ll} $\g_{3,2}\oplus \g_1: \quad$
&$[e_1,e_3]=e_1,\ [e_2,e_3]=e_1+e_2$
\end{tabular}
\vspace{-16pt}\begin{center}
\begin{tabular}[t]{|l||c|c|c|c|}
\hline \parbox[l][20pt][c]{0pt}{}   $\alpha$ & 1 & $0$&
\\ \hline
\parbox[l][20pt][c]{0pt}{} $\fa(\alpha)$ & 6 & 7 & 5  \\ \hline
\end{tabular}\qquad
\begin{tabular}[t]{|l||c|c|c|}
\hline  \parbox[l][20pt][c]{0pt}{}  $\alpha$ & $1$  &
\\ \hline
\parbox[l][20pt][c]{0pt}{} $\fb(\alpha)$ & 13  & 12 \\
\hline
\end{tabular}\qquad
\begin{tabular}[t]{|l||c|c|c|c|}
\hline \parbox[l][20pt][c]{0pt}{}   $\alpha$ & 1 & $2$ &
\\ \hline
\parbox[l][20pt][c]{0pt}{} $\fc(\alpha)$ & 3 & 3 & 1 \\ \hline
\end{tabular}
\end{center}

\item\begin{tabular}{ll} $\g_{3,3}\oplus\g_1: \quad$
&$[e_1,e_3]=e_1,\ [e_2,e_3]=e_2$
\end{tabular}
\vspace{-16pt}\begin{center}
\begin{tabular}[t]{|l||c|c|c|c|}
\hline \parbox[l][20pt][c]{0pt}{}   $\alpha$ & 1 & $0$&
\\ \hline
\parbox[l][20pt][c]{0pt}{} $\fa(\alpha)$ & 8 & 7 & 5  \\ \hline
\end{tabular}\qquad
\begin{tabular}[t]{|l||c|c|c|}
\hline  \parbox[l][20pt][c]{0pt}{}  $\alpha$ & $1$  &
\\ \hline
\parbox[l][20pt][c]{0pt}{} $\fb(\alpha)$ & 15  & 12 \\
\hline
\end{tabular}\qquad
\begin{tabular}[t]{|l||c|c|c|c|}
\hline \parbox[l][20pt][c]{0pt}{}   $\alpha$ & 1 & $2$ &
\\ \hline
\parbox[l][20pt][c]{0pt}{} $\fc(\alpha)$ & 3 & 7 & 1 \\ \hline
\end{tabular}
\end{center}

\newpage
\item\begin{tabular}{ll} $\g_{3,4}(-1)\oplus\g_1: \quad$
&$[e_1,e_3]=e_1,\ [e_2,e_3]=-e_2$
\end{tabular}
\vspace{-16pt}\begin{center}
\begin{tabular}[t]{|l||c|c|c|c|c|}
\hline \parbox[l][20pt][c]{0pt}{}   $\alpha$ & 1 & $0$ & $-1$&
\\ \hline
\parbox[l][20pt][c]{0pt}{} $\fa(\alpha)$ & 6 & 7 &7& 5  \\ \hline
\end{tabular}\qquad
\begin{tabular}[t]{|l||c|c|c|c|}
\hline  \parbox[l][20pt][c]{0pt}{}  $\alpha$ & $1$& 0 & $-1$&
\\ \hline
\parbox[l][20pt][c]{0pt}{} $\fb(\alpha)$ & 15  & 16&16&14 \\
\hline
\end{tabular}\qquad
\begin{tabular}[t]{|l||c|c|c|c|c|}
\hline \parbox[l][20pt][c]{0pt}{}   $\alpha$ & 1 & $2$ & 0&
\\ \hline
\parbox[l][20pt][c]{0pt}{} $\fc(\alpha)$ & 3 & 3 & 3&  1 \\ \hline
\end{tabular}
\end{center}

\item\label{g34}\begin{tabular}{ll} $\g_{3,4}(a)\oplus \g_1: \quad$
&$[e_1,e_3]=e_1,\ [e_2,e_3]=ae_2,\ a\neq 0,\pm 1$
\end{tabular}
\vspace{-16pt}\begin{center}
\begin{tabular}[t]{|l||c|c|c|c|c|c|}
\hline \parbox[l][20pt][c]{0pt}{}   $\alpha$ & 1 & $0$ &
$a$&$\frac{1}{a}$ &
\\ \hline
\parbox[l][20pt][c]{0pt}{} $\fa(\alpha)$ & 6 & 7 &6& 6& 5  \\ \hline
\end{tabular}\qquad
\begin{tabular}[t]{|l||c|c|c|c|}
\hline  \parbox[l][20pt][c]{0pt}{}  $\alpha$ & $1$& $a$ &
$\frac{1}{a}$&
\\ \hline
\parbox[l][20pt][c]{0pt}{} $\fb(\alpha)$ & 13  & 13&13&12 \\
\hline
\end{tabular}\qquad

\vspace{12pt}

\begin{tabular}[t]{|l||c|c|c|c|c|}
\hline \parbox[l][20pt][c]{0pt}{}   $\alpha$ & 1 & $2$ &
$1+a$&$1+\frac{1}{a}$ &
\\ \hline
\parbox[l][20pt][c]{0pt}{} $\fc(\alpha)$ & 3 & 3 & 2 &2&  1 \\ \hline
\end{tabular}
\end{center}

\item\begin{tabular}{ll}$\slp (2,\Com)\oplus\g_1: \quad$
&$[e_1,e_2]=e_1,\ [e_2,e_3]=e_3,\ [e_1,e_3]=2e_2$
\end{tabular}
\vspace{-16pt}\begin{center}
\begin{tabular}[t]{|l||c|c|c|c|c|}
\hline \parbox[l][20pt][c]{0pt}{}   $\alpha$ & 1 & $0$ & $-1$& 2&
\\ \hline
\parbox[l][20pt][c]{0pt}{} $\fa(\alpha)$ & 4 & 4 &6& 2 & 1  \\ \hline
\end{tabular}\qquad
\begin{tabular}[t]{|l||c|c|c|c|c|}
\hline  \parbox[l][20pt][c]{0pt}{}  $\alpha$ & $1$& 0 & $-1$&
$\frac{1}{2}$ &
\\ \hline
\parbox[l][20pt][c]{0pt}{} $\fb(\alpha)$ & 12  & 12&14&10&9 \\
\hline
\end{tabular}\qquad
\begin{tabular}[t]{|l||c|c|c|c|c|}
\hline \parbox[l][20pt][c]{0pt}{}   $\alpha$ & 2 &
\\ \hline
\parbox[l][20pt][c]{0pt}{} $\fc(\alpha)$ & 1 & 0 \\ \hline
\end{tabular}
\end{center}

\item\begin{tabular}{ll}$\g_{4,1}: \quad$ &$[e_2,e_4]=e_1,\
[e_3,e_4]=e_2$
\end{tabular}
\vspace{-16pt}\begin{center}
\begin{tabular}[t]{|l||c|c|}
\hline \parbox[l][20pt][c]{0pt}{}   $\alpha$ & 1 &  \\ \hline
\parbox[l][20pt][c]{0pt}{} $\fa(\alpha)$ & 7 & 7  \\ \hline
\end{tabular}\qquad
\begin{tabular}[t]{|l||c|c|c|}
\hline  \parbox[l][20pt][c]{0pt}{}  $\alpha$ & -1 & 0 & \\ \hline
\parbox[l][20pt][c]{0pt}{} $\fb(\alpha)$ & 16  & 16 & 15\\
\hline
\end{tabular}\qquad
\begin{tabular}[t]{|l||c|c|c|}
\hline \parbox[l][20pt][c]{0pt}{}   $\alpha$       &  \\ \hline
\parbox[l][20pt][c]{0pt}{} $\fc(\alpha)$ & 3 \\ \hline
\end{tabular}
\end{center}

\item\label{g42}\begin{tabular}{ll}$\g_{4,2}(a): \quad$ &$[e_1,e_4]=ae_1,\
[e_2,e_4]=e_2, \ [e_3,e_4]=e_2+e_3,\ a\neq 0,\pm 1, -2 $
\end{tabular}
\vspace{-16pt}\begin{center}
\begin{tabular}[t]{|l||c|c|c|c|}
\hline \parbox[l][20pt][c]{0pt}{}   $\alpha$ & 1 & $a$&
$\frac{1}{a}$& \\ \hline
\parbox[l][20pt][c]{0pt}{} $\fa(\alpha)$ & 6 & 5 & 5 & 4  \\ \hline
\end{tabular}\qquad
\begin{tabular}[t]{|l||c|c|c|}
\hline  \parbox[l][20pt][c]{0pt}{}  $\alpha$ & $1+a$ &
$\frac{2}{a}$ &
\\ \hline
\parbox[l][20pt][c]{0pt}{} $\fb(\alpha)$ & 13  & 13 & 12\\
\hline
\end{tabular}

\vspace{12pt}

\begin{tabular}[t]{|l||c|c|c|c|}
\hline \parbox[l][20pt][c]{0pt}{}   $\alpha$ & 2 & $1+a$ &
$1+\frac{1}{a}$ &
\\ \hline
\parbox[l][20pt][c]{0pt}{} $\fc(\alpha)$ & 3 & 1 & 1 & 0 \\ \hline
\end{tabular}
\end{center}

\item\begin{tabular}{ll}$\g_{4,2}(1): \quad$ &$[e_1,e_4]=e_1,\
[e_2,e_4]=e_2, \ [e_3,e_4]=e_2+e_3 $
\end{tabular}
\vspace{-16pt}\begin{center}
\begin{tabular}[t]{|l||c|c|c|c|}
\hline \parbox[l][20pt][c]{0pt}{}   $\alpha$ & 1 &
\\ \hline
\parbox[l][20pt][c]{0pt}{} $\fa(\alpha)$ & 8 & 4  \\ \hline
\end{tabular}\qquad
\begin{tabular}[t]{|l||c|c|c|c|}
\hline \parbox[l][20pt][c]{0pt}{}   $\alpha$ & 2 &
\\ \hline
\parbox[l][20pt][c]{0pt}{} $\fb(\alpha)$ & 15 & 12  \\ \hline
\end{tabular}\qquad
\begin{tabular}[t]{|l||c|c|c|c|}
\hline \parbox[l][20pt][c]{0pt}{}   $\alpha$ & 2 &
\\ \hline
\parbox[l][20pt][c]{0pt}{} $\fc(\alpha)$ & 7 & 0  \\ \hline
\end{tabular}
\end{center}

\newpage
\item\begin{tabular}{ll}$\g_{4,2}(-2): \quad$ &$[e_1,e_4]=-2e_1,\
[e_2,e_4]=e_2, \ [e_3,e_4]=e_2+e_3 $
\end{tabular}
\vspace{-16pt}\begin{center}
\begin{tabular}[t]{|l||c|c|c|c|}
\hline \parbox[l][20pt][c]{0pt}{}   $\alpha$ & 1 & $-2$&
$-\frac{1}{2}$& \\ \hline
\parbox[l][20pt][c]{0pt}{} $\fa(\alpha)$ & 6 & 5 & 5 & 4  \\ \hline
\end{tabular}\qquad
\begin{tabular}[t]{|l||c|c|c|c|}
\hline \parbox[l][20pt][c]{0pt}{}   $\alpha$ & -1 &
\\ \hline
\parbox[l][20pt][c]{0pt}{} $\fb(\alpha)$ & 15 & 12  \\ \hline
\end{tabular}\qquad
\begin{tabular}[t]{|l||c|c|c|c|}
\hline \parbox[l][20pt][c]{0pt}{}   $\alpha$ & 2 & $-1$ &
$\frac{1}{2}$ &
\\ \hline
\parbox[l][20pt][c]{0pt}{} $\fc(\alpha)$ & 3 & 1 & 1 & 0 \\ \hline
\end{tabular}
\end{center}

\item\begin{tabular}{ll}$\g_{4,2}(-1): \quad$ &$[e_1,e_4]=-e_1,\
[e_2,e_4]=e_2, \ [e_3,e_4]=e_2+e_3$
\end{tabular}
\vspace{-16pt}\begin{center}
\begin{tabular}[t]{|l||c|c|c|c|}
\hline \parbox[l][20pt][c]{0pt}{}   $\alpha$ & 1 & $-1$&
\\ \hline
\parbox[l][20pt][c]{0pt}{} $\fa(\alpha)$ & 6 & 6 &  4  \\ \hline
\end{tabular}\qquad
\begin{tabular}[t]{|l||c|c|c|}
\hline  \parbox[l][20pt][c]{0pt}{}  $\alpha$ & $-2$ & $0$ &
\\ \hline
\parbox[l][20pt][c]{0pt}{} $\fb(\alpha)$ & 13  & 16 & 12\\
\hline
\end{tabular}\qquad
\begin{tabular}[t]{|l||c|c|c|c|}
\hline \parbox[l][20pt][c]{0pt}{}   $\alpha$ & $0$ & $2$ &
\\ \hline
\parbox[l][20pt][c]{0pt}{} $\fc(\alpha)$ & 2 & 3 & 0  \\ \hline
\end{tabular}
\end{center}

\item\begin{tabular}{ll}$\g_{4,3}: \quad$ &$[e_1,e_4]=e_1,\
 [e_3,e_4]=e_2$
\end{tabular}
\vspace{-16pt}\begin{center}
\begin{tabular}[t]{|l||c|c|c|c|}
\hline \parbox[l][20pt][c]{0pt}{}   $\alpha$ & 1 & $0$&
\\ \hline
\parbox[l][20pt][c]{0pt}{} $\fa(\alpha)$ & 6 & 7 &  6  \\ \hline
\end{tabular}\qquad
\begin{tabular}[t]{|l||c|c|c|c|}
\hline \parbox[l][20pt][c]{0pt}{}   $\alpha$ & 0 &
\\ \hline
\parbox[l][20pt][c]{0pt}{} $\fb(\alpha)$ & 16 & 13  \\ \hline
\end{tabular}\qquad
\begin{tabular}[t]{|l||c|c|c|c|}
\hline \parbox[l][20pt][c]{0pt}{}   $\alpha$ & 1 & 2 &
\\ \hline
\parbox[l][20pt][c]{0pt}{} $\fc(\alpha)$ & 3 & 3& 2  \\ \hline
\end{tabular}
\end{center}

\item\begin{tabular}{ll}$\g_{4,4}: \quad$ &$[e_1,e_4]=e_1,\
[e_2,e_4]=e_1+e_2, \ [e_3,e_4]=e_2+e_3$
\end{tabular}
\vspace{-16pt}\begin{center}
\begin{tabular}[t]{|l||c|c|c|c|}
\hline \parbox[l][20pt][c]{0pt}{}   $\alpha$ & 1 &
\\ \hline
\parbox[l][20pt][c]{0pt}{} $\fa(\alpha)$ & 6 & 4 \\ \hline
\end{tabular}\qquad
\begin{tabular}[t]{|l||c|c|c|c|}
\hline \parbox[l][20pt][c]{0pt}{}   $\alpha$ & 2 &
\\ \hline
\parbox[l][20pt][c]{0pt}{} $\fb(\alpha)$ & 13 & 12  \\ \hline
\end{tabular}\qquad
\begin{tabular}[t]{|l||c|c|c|c|}
\hline \parbox[l][20pt][c]{0pt}{}   $\alpha$ & 1 & 2 &
\\ \hline
\parbox[l][20pt][c]{0pt}{} $\fc(\alpha)$ & 0 & 3& 0  \\ \hline
\end{tabular}
\end{center}

\item\label{g45ab}\begin{tabular}[t]{ll}$\g_{4,5}(a,b): \quad$ &$[e_1,e_4]=ae_1,\
[e_2,e_4]=be_2, \ [e_3,e_4]=e_3, $ \\ & $a\neq 0,\pm 1,\pm
b,1/b,b^2,-1-b, $   $\ b\neq 0,\pm 1,\pm a,1/a,a^2,-1-a$
\end{tabular}
\vspace{-8pt}\begin{center}
\begin{tabular}[t]{|l||c|c|c|c|c|c|c|c|}
\hline \parbox[l][20pt][c]{0pt}{}   $\alpha$ & 1 & $a$&
$\frac{1}{a}$& $b$& $\frac{1}{b}$ &  $\frac{a}{b}$ & $\frac{b}{a}$
&
\\ \hline
\parbox[l][20pt][c]{0pt}{} $\fa(\alpha)$ & 6 & 5 &  5& 5& 5&5&5&4  \\ \hline
\end{tabular}\qquad
\begin{tabular}[t]{|l||c|c|c|c|}
\hline \parbox[l][20pt][c]{0pt}{}   $\alpha$ &  $a+b$ &
$\frac{1+a}{b}$ & $\frac{1+b}{a}$ &
\\ \hline
\parbox[l][20pt][c]{0pt}{} $\fb(\alpha)$ & 13 & 13 & 13& 12  \\ \hline
\end{tabular}

\vspace{12pt}

\begin{tabular}[t]{|l||c|c|c|c|c|c|c|c|}
\hline \parbox[l][20pt][c]{0pt}{}   $\alpha$ & 2 & $1+a$& $1+b$ &
$1+\frac{1}{a}$&  $1+\frac{1}{b}$ & 1+ $\frac{a}{b}$ &
$1+\frac{b}{a}$ &
\\ \hline
\parbox[l][20pt][c]{0pt}{} $\fc(\alpha)$ & 3 & 1 &  1& 1& 1&1&1&0  \\ \hline
\end{tabular}
\end{center}

\item\label{g45am1ma}\begin{tabular}[t]{ll}$\g_{4,5}(a,-1-a): \quad$ &$[e_1,e_4]=ae_1,\
[e_2,e_4]=(-1-a)e_2, \ [e_3,e_4]=e_3, $ \\ & $a\neq 0,\pm 1,-2
,-1/2,-1/2\pm i\sqrt{3}/2  $
\end{tabular}
\vspace{-8pt}\begin{center}
\begin{tabular}[t]{|l||c|c|c|c|c|c|c|c|}
\hline \parbox[l][20pt][c]{0pt}{}   $\alpha$ & 1 & $a$&
$\frac{1}{a}$& $-1-a$& $-\frac{1}{1+a}$ &  $-\frac{a}{a+1}$ &
$-\frac{a+1}{a}$ &
\\ \hline
\parbox[l][20pt][c]{0pt}{} $\fa(\alpha)$ & 6 & 5 &  5& 5& 5&5&5&4  \\ \hline
\end{tabular}\qquad
\begin{tabular}[t]{|l||c|c|c|c|}
\hline \parbox[l][20pt][c]{0pt}{}   $\alpha$ &  -1 &
\\ \hline
\parbox[l][20pt][c]{0pt}{} $\fb(\alpha)$ & 15 &  12  \\ \hline
\end{tabular}

\vspace{12pt}

\begin{tabular}[t]{|l||c|c|c|c|c|c|c|c|}
\hline \parbox[l][20pt][c]{0pt}{}   $\alpha$ & 2 & $1+a$& $-a$ &
$1+\frac{1}{a}$&  $\frac{a}{a+1}$ & $\frac{1}{a+1}$ &
$-\frac{1}{a}$ &
\\ \hline
\parbox[l][20pt][c]{0pt}{} $\fc(\alpha)$ & 3 & 1 &  1& 1& 1&1&1&0  \\ \hline
\end{tabular}
\end{center}

\item\label{g45aas}\begin{tabular}[t]{ll}$\g_{4,5}(a,a^2): \quad$ &$[e_1,e_4]=ae_1,\
[e_2,e_4]=a^2e_2, \ [e_3,e_4]=e_3, $ \\ & $a\neq 0,\pm 1,\pm
i,-1/2\pm i\sqrt{3}/2 $
\end{tabular}
\vspace{-8pt}\begin{center}
\begin{tabular}[t]{|l||c|c|c|c|c|c|c|c|}
\hline \parbox[l][20pt][c]{0pt}{}   $\alpha$ & 1 & $a$&
$\frac{1}{a}$& $a^2$& $\frac{1}{a^2}$ &
\\ \hline
\parbox[l][20pt][c]{0pt}{} $\fa(\alpha)$ & 6 & 6 &  6& 5& 5&4  \\ \hline
\end{tabular}\qquad
\begin{tabular}[t]{|l||c|c|c|c|}
\hline \parbox[l][20pt][c]{0pt}{}   $\alpha$ &  $a+a^2$ &
$\frac{a+1}{a^2}$ & $\frac{a^2+1}{a}$ &
\\ \hline
\parbox[l][20pt][c]{0pt}{} $\fb(\alpha)$ & 13 & 13 & 13& 12  \\ \hline
\end{tabular}

\vspace{12pt}

\begin{tabular}[t]{|l||c|c|c|c|c|c|c|c|}
\hline \parbox[l][20pt][c]{0pt}{}   $\alpha$ & 2 & $1+a$& $1+a^2$
& $1+\frac{1}{a}$&  $1+\frac{1}{a^2}$ &
\\ \hline
\parbox[l][20pt][c]{0pt}{} $\fc(\alpha)$ & 3 & 2 &  1& 2& 1&0  \\ \hline
\end{tabular}
\end{center}

\item\label{g45a1}\begin{tabular}[t]{ll}$\g_{4,5}(a,1): \quad$ &$[e_1,e_4]=ae_1,\
[e_2,e_4]=e_2, \ [e_3,e_4]=e_3, $ \\ & $a\neq 0,\pm 1,-2$
\end{tabular}
\vspace{-8pt}\begin{center}
\begin{tabular}[t]{|l||c|c|c|c|}
\hline \parbox[l][20pt][c]{0pt}{}   $\alpha$ & 1 & $a$&
$\frac{1}{a}$& \\ \hline
\parbox[l][20pt][c]{0pt}{} $\fa(\alpha)$ & 8 & 6 & 6 & 4  \\ \hline
\end{tabular}\qquad
\begin{tabular}[t]{|l||c|c|c|}
\hline  \parbox[l][20pt][c]{0pt}{}  $\alpha$ & $1+a$ &
$\frac{2}{a}$ &
\\ \hline
\parbox[l][20pt][c]{0pt}{} $\fb(\alpha)$ & 15  & 13 & 12\\
\hline
\end{tabular}

\vspace{12pt}

\begin{tabular}[t]{|l||c|c|c|c|}
\hline \parbox[l][20pt][c]{0pt}{}   $\alpha$ & 2 & $1+a$ &
$1+\frac{1}{a}$ &
\\ \hline
\parbox[l][20pt][c]{0pt}{} $\fc(\alpha)$ & 7 & 2 & 2 & 0 \\ \hline
\end{tabular}
\end{center}

\item\label{g45am1}\begin{tabular}[t]{ll}$\g_{4,5}(a,-1): \quad$ &$[e_1,e_4]=ae_1,\
[e_2,e_4]=-e_2, \ [e_3,e_4]=e_3, $ \\ & $a\neq 0,\pm 1,\pm i$
\end{tabular}
\vspace{-8pt}\begin{center}
\begin{tabular}[t]{|l||c|c|c|c|c|c|c|c|}
\hline \parbox[l][20pt][c]{0pt}{}   $\alpha$ & 1 & $a$&
$\frac{1}{a}$& $-1$&   $-a$ & $-\frac{1}{a}$ &
\\ \hline
\parbox[l][20pt][c]{0pt}{} $\fa(\alpha)$ & 6 & 5 &  5& 6& 5&5&4  \\ \hline
\end{tabular}\qquad
\begin{tabular}[t]{|l||c|c|c|c|}
\hline \parbox[l][20pt][c]{0pt}{}   $\alpha$ &  $-1+a$ & $-1-a$ &
$0$ &
\\ \hline
\parbox[l][20pt][c]{0pt}{} $\fb(\alpha)$ & 13 & 13 & 16& 12  \\ \hline
\end{tabular}

\vspace{12pt}

\begin{tabular}[t]{|l||c|c|c|c|c|c|c|c|}
\hline \parbox[l][20pt][c]{0pt}{}   $\alpha$ & 2 & $1+a$& $0$ &
$1+\frac{1}{a}$ & $1-a$ & $1-\frac{1}{a}$ &
\\ \hline
\parbox[l][20pt][c]{0pt}{} $\fc(\alpha)$ & 3 & 1 &  2& 1&1&1&0  \\ \hline
\end{tabular}
\end{center}

\item\begin{tabular}[t]{ll}$\g_{4,5}(1,1): \quad$ &$[e_1,e_4]=e_1,\
[e_2,e_4]=e_2, \ [e_3,e_4]=e_3 $
\end{tabular}
\vspace{-16pt}\begin{center}
\begin{tabular}[t]{|l||c|c|c|c|}
\hline \parbox[l][20pt][c]{0pt}{}   $\alpha$ & 1 &
\\ \hline
\parbox[l][20pt][c]{0pt}{} $\fa(\alpha)$ & 12 & 4 \\ \hline
\end{tabular}\qquad
\begin{tabular}[t]{|l||c|c|c|c|}
\hline \parbox[l][20pt][c]{0pt}{}   $\alpha$ & 2 &
\\ \hline
\parbox[l][20pt][c]{0pt}{} $\fb(\alpha)$ & 18 & 12  \\ \hline
\end{tabular}\qquad
\begin{tabular}[t]{|l||c|c|c|c|}
\hline \parbox[l][20pt][c]{0pt}{}   $\alpha$ &  2 &
\\ \hline
\parbox[l][20pt][c]{0pt}{} $\fc(\alpha)$ & 18 &  0  \\ \hline
\end{tabular}
\end{center}

\item\begin{tabular}[t]{ll}$\g_{4,5}(-1,1): \quad$ &$[e_1,e_4]=-e_1,\
[e_2,e_4]=e_2, \ [e_3,e_4]=e_3 $
\end{tabular}
\vspace{-16pt}\begin{center}
\begin{tabular}[t]{|l||c|c|c|c|}
\hline \parbox[l][20pt][c]{0pt}{}   $\alpha$ & 1 & $-1$&
\\ \hline
\parbox[l][20pt][c]{0pt}{} $\fa(\alpha)$ & 8 & 8 & 4  \\ \hline
\end{tabular}\qquad
\begin{tabular}[t]{|l||c|c|c|}
\hline  \parbox[l][20pt][c]{0pt}{}  $\alpha$ & $0$ & $-2$ &
\\ \hline
\parbox[l][20pt][c]{0pt}{} $\fb(\alpha)$ & 20  & 13 & 12\\
\hline
\end{tabular}\qquad
\begin{tabular}[t]{|l||c|c|c|c|}
\hline \parbox[l][20pt][c]{0pt}{}   $\alpha$ & 2 & $0$ &
\\ \hline
\parbox[l][20pt][c]{0pt}{} $\fc(\alpha)$ & 7 & 4 & 0 \\ \hline
\end{tabular}
\end{center}

\newpage
\item\begin{tabular}[t]{ll}$\g_{4,5}(-2,1): \quad$ &$[e_1,e_4]=-2e_1,\
[e_2,e_4]=e_2, \ [e_3,e_4]=e_3 $
\end{tabular}
\vspace{-16pt}\begin{center}
\begin{tabular}[t]{|l||c|c|c|c|}
\hline \parbox[l][20pt][c]{0pt}{}   $\alpha$ & 1 & $-2$&
$-\frac{1}{2}$& \\ \hline
\parbox[l][20pt][c]{0pt}{} $\fa(\alpha)$ & 8 & 6 & 6 & 4  \\ \hline
\end{tabular}\qquad
\begin{tabular}[t]{|l||c|c|c|}
\hline  \parbox[l][20pt][c]{0pt}{}  $\alpha$ & $-1$ &
\\ \hline
\parbox[l][20pt][c]{0pt}{} $\fb(\alpha)$ & 16   & 12\\
\hline
\end{tabular}\qquad
\begin{tabular}[t]{|l||c|c|c|c|}
\hline \parbox[l][20pt][c]{0pt}{}   $\alpha$ & 2 & $-1$ &
$\frac{1}{2}$ &
\\ \hline
\parbox[l][20pt][c]{0pt}{} $\fc(\alpha)$ & 7 & 2 & 2 & 0 \\ \hline
\end{tabular}
\end{center}

\item\begin{tabular}[t]{ll}$\g_{4,5}(-\frac{1}{2}+\frac{\sqrt{3}}{2}i,-\frac{1}{2}-\frac{\sqrt{3}}{2}i): \quad$ &$[e_1,e_4]=(-\frac{1}{2}+\frac{\sqrt{3}}{2}i)e_1,\
[e_2,e_4]=(-\frac{1}{2}-\frac{\sqrt{3}}{2})ie_2,$\\ &
$[e_3,e_4]=e_3
$
\end{tabular}
\vspace{-8pt}\begin{center}
\begin{tabular}[t]{|l||c|c|c|c|}
\hline \parbox[l][20pt][c]{0pt}{}   $\alpha$ & 1 &
$-\frac{1}{2}+\frac{\sqrt{3}}{2}i$&
$-\frac{1}{2}-\frac{\sqrt{3}}{2}i$& \\ \hline
\parbox[l][20pt][c]{0pt}{} $\fa(\alpha)$ & 6 & 7 & 7 & 4  \\ \hline
\end{tabular}\qquad
\begin{tabular}[t]{|l||c|c|c|}
\hline  \parbox[l][20pt][c]{0pt}{}  $\alpha$ & $-1$ &
\\ \hline
\parbox[l][20pt][c]{0pt}{} $\fb(\alpha)$ & 15   & 12\\
\hline
\end{tabular}\qquad

\vspace{12pt}

\begin{tabular}[t]{|l||c|c|c|c|}
\hline \parbox[l][20pt][c]{0pt}{}   $\alpha$ & 2 &
$\frac{1}{2}-\frac{\sqrt{3}}{2}i$ &
$\frac{1}{2}+\frac{\sqrt{3}}{2}i$ &
\\ \hline
\parbox[l][20pt][c]{0pt}{} $\fc(\alpha)$ & 3 & 3 & 3 & 0 \\ \hline
\end{tabular}
\end{center}

\item\begin{tabular}[t]{ll}$\g_{4,5}(i,-1): \quad$ &$[e_1,e_4]=ie_1,\
[e_2,e_4]=-e_2, \ [e_3,e_4]=e_3 $
\end{tabular}
\vspace{-16pt}\begin{center}
\begin{tabular}[t]{|l||c|c|c|c|c|c|c|c|}
\hline \parbox[l][20pt][c]{0pt}{}   $\alpha$ & 1 & $i$& $-i$&
$-1$&
\\ \hline
\parbox[l][20pt][c]{0pt}{} $\fa(\alpha)$ & 6 & 6 &  6& 6&4  \\ \hline
\end{tabular}\qquad
\begin{tabular}[t]{|l||c|c|c|c|}
\hline \parbox[l][20pt][c]{0pt}{}   $\alpha$ &  $-1+i$ & $-1-i$ &
$0$ &
\\ \hline
\parbox[l][20pt][c]{0pt}{} $\fb(\alpha)$ & 13 & 13 & 16& 12  \\ \hline
\end{tabular}

\vspace{12pt}

\begin{tabular}[t]{|l||c|c|c|c|c|c|c|c|}
\hline \parbox[l][20pt][c]{0pt}{}   $\alpha$ & 2 & $1+i$& $0$ &
$1-i$  &
\\ \hline
\parbox[l][20pt][c]{0pt}{} $\fc(\alpha)$ & 3 & 2 &  2& 2&0  \\ \hline
\end{tabular}
\end{center}

\item\begin{tabular}[t]{ll}$\g_{4,7}: \quad$ &$[e_2,e_3]=e_1,\
[e_1,e_4]=2e_1, \ [e_2,e_4]=e_2, \ [e_3,e_4]=e_2+e_3$
\end{tabular}
\vspace{-16pt}\begin{center}
\begin{tabular}[t]{|l||c|c|c|c|}
\hline \parbox[l][20pt][c]{0pt}{}   $\alpha$ & 1 & $2$&
\\ \hline
\parbox[l][20pt][c]{0pt}{} $\fa(\alpha)$ & 5 & 4 & 3  \\ \hline
\end{tabular}\qquad
\begin{tabular}[t]{|l||c|c|c|c|c|c|c|c|}
\hline \parbox[l][20pt][c]{0pt}{}   $\alpha$ & 0 & $1$& $3$  &
\\ \hline
\parbox[l][20pt][c]{0pt}{} $\fc(\alpha)$ & 12 & 12 &  12& 11  \\ \hline
\end{tabular}\qquad
\begin{tabular}[t]{|l||c|c|c|c|}
\hline \parbox[l][20pt][c]{0pt}{}   $\alpha$ & $\frac{3}{2}$
 & $2$&
\\ \hline
\parbox[l][20pt][c]{0pt}{} $\fc(\alpha)$ & 1 & 1 & 0 \\ \hline
\end{tabular}
\end{center}

\item\label{g48}\begin{tabular}[t]{ll}$\g_{4,8}(a): \quad$ &$[e_2,e_3]=e_1,\
[e_1,e_4]=(1+a)e_1, \ [e_2,e_4]=e_2, \ [e_3,e_4]=ae_3$ \\ & $a\neq
0,\pm 1,\pm 2,\pm 1/2,-1/2\pm\sqrt{3}i/2 $
\end{tabular}
\vspace{-4pt}\begin{center}
\begin{tabular}[t]{|l||c|c|c|c|c|c|c|c|}
\hline \parbox[l][20pt][c]{0pt}{}   $\alpha$ & 1 & $2$&   $a$&
$\frac{1}{a}$&
\\ \hline
\parbox[l][20pt][c]{0pt}{} $\fa(\alpha)$ & 5 & 4 &  4& 4& 3  \\ \hline
\end{tabular}\qquad
\begin{tabular}[t]{|l||c|c|c|c|c|c|c|c|}
\hline \parbox[l][20pt][c]{0pt}{}   $\alpha$ & 0 & $1$& $1+2a$ &
$1+\frac{2}{a}$ & $\frac{5}{3}+\frac{2}{3}(a+\frac{1}{a})$ &
\\ \hline
\parbox[l][20pt][c]{0pt}{} $\fb(\alpha)$ & 12 & 12 &  12& 12&12&11  \\ \hline
\end{tabular}
\vspace{-4pt}

\begin{tabular}[t]{|l||c|c|c|c|c|c|c|c|}
\hline \parbox[l][20pt][c]{0pt}{}   $\alpha$ & 2 &
$\frac{1+2a}{1+a}$& $\frac{2+a}{1+a}$&
\\ \hline
\parbox[l][20pt][c]{0pt}{} $\fc(\alpha)$ & 1 & 1 &  1&0  \\ \hline
\end{tabular}
\end{center}

\newpage
\item\begin{tabular}[t]{ll}$\g_{4,8}(1): \quad$ &$[e_2,e_3]=e_1,\
[e_1,e_4]=2e_1, \ [e_2,e_4]=e_2, \ [e_3,e_4]=e_3$
\end{tabular}
\vspace{-16pt}\begin{center}
\begin{tabular}[t]{|l||c|c|c|c|c|c|c|c|}
\hline \parbox[l][20pt][c]{0pt}{}   $\alpha$ & 1 & $2$&
\\ \hline
\parbox[l][20pt][c]{0pt}{} $\fa(\alpha)$ & 7 & 4 &  3  \\ \hline
\end{tabular}\qquad
\begin{tabular}[t]{|l||c|c|c|c|c|c|c|c|}
\hline \parbox[l][20pt][c]{0pt}{}   $\alpha$ & 0 & $1$& $3$ &
\\ \hline
\parbox[l][20pt][c]{0pt}{} $\fb(\alpha)$ & 12 & 12 &  14&11  \\ \hline
\end{tabular}\qquad
\begin{tabular}[t]{|l||c|c|c|c|c|c|c|c|}
\hline \parbox[l][20pt][c]{0pt}{}   $\alpha$ & 2 & $\frac{3}{2}$&
\\ \hline
\parbox[l][20pt][c]{0pt}{} $\fc(\alpha)$ & 1 & 2 & 0  \\ \hline
\end{tabular}
\end{center}

\item\begin{tabular}[t]{ll}$\g_{4,8}(2): \quad$ &$[e_2,e_3]=e_1,\
[e_1,e_4]=3e_1, \ [e_2,e_4]=e_2, \ [e_3,e_4]=2e_3$
\end{tabular}
\vspace{-16pt}\begin{center}
\begin{tabular}[t]{|l||c|c|c|c|c|c|c|c|}
\hline \parbox[l][20pt][c]{0pt}{}   $\alpha$ & 1 & $2$&
$\frac{1}{2}$&
\\ \hline
\parbox[l][20pt][c]{0pt}{} $\fa(\alpha)$ & 5 & 5 &  4&  3  \\ \hline
\end{tabular}\qquad
\begin{tabular}[t]{|l||c|c|c|c|c|c|c|c|}
\hline \parbox[l][20pt][c]{0pt}{}   $\alpha$ & 0 & $1$& $5$ & $2$
& $\frac{10}{3}$ &
\\ \hline
\parbox[l][20pt][c]{0pt}{} $\fb(\alpha)$ & 12 & 12 &  12& 12&12&11  \\ \hline
\end{tabular}
\vspace{12pt}

\begin{tabular}[t]{|l||c|c|c|c|c|c|c|c|}
\hline \parbox[l][20pt][c]{0pt}{}   $\alpha$ & 2 & $\frac{5}{3}$&
$\frac{4}{3}$&
\\ \hline
\parbox[l][20pt][c]{0pt}{} $\fc(\alpha)$ & 1 & 1 &  1&0  \\ \hline
\end{tabular}
\end{center}

\item\begin{tabular}[t]{ll}$\g_{4,8}(0): \quad$ &$[e_2,e_3]=e_1,\
[e_1,e_4]=e_1, \ [e_2,e_4]=e_2$

\end{tabular}
\vspace{-16pt}\begin{center}
\begin{tabular}[t]{|l||c|c|c|c|c|c|c|c|}
\hline \parbox[l][20pt][c]{0pt}{}   $\alpha$ & 1 & $0$&
\\ \hline
\parbox[l][20pt][c]{0pt}{} $\fa(\alpha)$ & 5 & 6 & 4  \\ \hline
\end{tabular}\qquad
\begin{tabular}[t]{|l||c|c|c|c|c|c|c|c|}
\hline \parbox[l][20pt][c]{0pt}{}   $\alpha$ & 0 & $1$&
\\ \hline
\parbox[l][20pt][c]{0pt}{} $\fb(\alpha)$ & 12 & 13 & 11  \\ \hline
\end{tabular}\qquad
\begin{tabular}[t]{|l||c|c|c|c|c|c|c|c|}
\hline \parbox[l][20pt][c]{0pt}{}   $\alpha$ & 1 & 2&
\\ \hline
\parbox[l][20pt][c]{0pt}{} $\fc(\alpha)$ & 2 & 2 &0  \\ \hline
\end{tabular}
\end{center}

\item\begin{tabular}[t]{ll}$\g_{4,8}(-1): \quad$ &$[e_2,e_3]=e_1,\
\ [e_2,e_4]=e_2, \ [e_3,e_4]=-e_3$
\end{tabular}
\vspace{-16pt}\begin{center}
\begin{tabular}[t]{|l||c|c|c|c|c|c|c|c|}
\hline \parbox[l][20pt][c]{0pt}{}   $\alpha$ & 1 & $-1$&
\\ \hline
\parbox[l][20pt][c]{0pt}{} $\fa(\alpha)$ & 5 & 6 & 4  \\ \hline
\end{tabular}\qquad
\begin{tabular}[t]{|l||c|c|c|c|c|c|c|c|}
\hline \parbox[l][20pt][c]{0pt}{}   $\alpha$ & 1 & $-1$&
\\ \hline
\parbox[l][20pt][c]{0pt}{} $\fb(\alpha)$ & 13 & 14 & 12  \\ \hline
\end{tabular}\qquad
\begin{tabular}[t]{|l||c|c|c|c|c|c|c|c|}
\hline \parbox[l][20pt][c]{0pt}{}   $\alpha$ & 2 &
\\ \hline
\parbox[l][20pt][c]{0pt}{} $\fc(\alpha)$ & 1 & 0  \\ \hline
\end{tabular}
\end{center}

\item\begin{tabular}[t]{ll}$\g_{4,8}(-2): \quad$ &$[e_2,e_3]=e_1,\
[e_1,e_4]=-e_1, \ [e_2,e_4]=e_2, \ [e_3,e_4]=-2e_3$
\end{tabular}
\vspace{-16pt}\begin{center}
\begin{tabular}[t]{|l||c|c|c|c|c|c|c|c|}
\hline \parbox[l][20pt][c]{0pt}{}   $\alpha$ & 1 & $2$&   $-2$&
$-\frac{1}{2}$&
\\ \hline
\parbox[l][20pt][c]{0pt}{} $\fa(\alpha)$ & 5 & 4 &  4& 4& 3  \\ \hline
\end{tabular}\qquad
\begin{tabular}[t]{|l||c|c|c|c|c|c|c|c|}
\hline \parbox[l][20pt][c]{0pt}{}   $\alpha$ & 0 & $1$& $-3$ &
\\ \hline
\parbox[l][20pt][c]{0pt}{} $\fb(\alpha)$ & 16 & 12 &  12& 11  \\ \hline
\end{tabular}
\vspace{12pt}

\begin{tabular}[t]{|l||c|c|c|c|c|c|c|c|}
\hline \parbox[l][20pt][c]{0pt}{}   $\alpha$ & 2 & $3$& $0$&
\\ \hline
\parbox[l][20pt][c]{0pt}{} $\fc(\alpha)$ & 1 & 1 &  1&0  \\ \hline
\end{tabular}
\end{center}

\item\begin{tabular}[t]{ll}$\g_{4,8}(-\frac{1}{2}+\frac{\sqrt{3}}{2}i): \quad$ &$[e_2,e_3]=e_1,\
[e_1,e_4]=(\frac{1}{2}+\frac{\sqrt{3}}{2}i)e_1, \ [e_2,e_4]=e_2$,
\\ & $[e_3,e_4]=(-\frac{1}{2}+\frac{\sqrt{3}}{2}i)e_3$
\end{tabular}
\vspace{-0pt}\begin{center}
\begin{tabular}[t]{|l||c|c|c|c|c|c|c|c|}
\hline \parbox[l][20pt][c]{0pt}{}   $\alpha$ & 1 & $2$&
$-\frac{1}{2}+\frac{\sqrt{3}}{2}i$&
$-\frac{1}{2}-\frac{\sqrt{3}}{2}i$&
\\ \hline
\parbox[l][20pt][c]{0pt}{} $\fa(\alpha)$ & 5 & 4 &  4& 4& 3  \\ \hline
\end{tabular}\qquad
\begin{tabular}[t]{|l||c|c|c|c|c|c|c|c|}
\hline \parbox[l][20pt][c]{0pt}{}   $\alpha$ & 0 & $1$&
$\sqrt{3}i$ & $-\sqrt{3}i$ &
\\ \hline
\parbox[l][20pt][c]{0pt}{} $\fb(\alpha)$ & 12 &  12& 12&12&11  \\ \hline
\end{tabular}
\vspace{12pt}

\begin{tabular}[t]{|l||c|c|c|c|c|c|c|c|}
\hline \parbox[l][20pt][c]{0pt}{}   $\alpha$ & 2 &
$\frac{3}{2}+\frac{\sqrt{3}}{2}i$&
$\frac{3}{2}-\frac{\sqrt{3}}{2}i$&
\\ \hline
\parbox[l][20pt][c]{0pt}{} $\fc(\alpha)$ & 1 & 1 &  1&0  \\ \hline
\end{tabular}
\end{center}
\end{enumerate}
\subsection*{Comparison of Notations of Lie Algebras}
We reproduce the list from \cite{Nes} of isomorphisms of Lie
algebras which compares notation based on \cite{Nes} to those in
\cite{Bur1} and \cite{AY}. The ordering of items is the same as
above. We label four--dimensional Lie algebras by corresponding
labels (g--1),$\dots$, (g--34).
\begin{center}
%
$
\end{flushleft}

\newpage
\section{Jordan Algebras}\label{IFLJJ}

\subsection*{One--dimensional Complex Jordan Algebras}
\hspace{16pt}
\begin{tabular}{ll}
$\j_{1}: \quad$ & Abelian
\end{tabular}
\begin{center}
\begin{tabular}[t]{|l||c|c|}
\hline \parbox[l][20pt][c]{0pt}{}   $\alpha$ & 1 &  \\ \hline
\parbox[l][20pt][c]{0pt}{} $\fa(\alpha)$ & 1 & 1  \\ \hline
\end{tabular}\qquad
\end{center}

\begin{tabular}{ll}
$\j_{1,1}: \quad$ & $e_1 \circ e_1 = e_1  $
\end{tabular}
\begin{center}
\begin{tabular}[t]{|l||c|c|c|}
\hline \parbox[l][20pt][c]{0pt}{}   $\alpha$ & 1 & 2 &  \\ \hline
\parbox[l][20pt][c]{0pt}{} $\fa(\alpha)$ & 0 & 1 & 0 \\ \hline
\end{tabular}\qquad
\end{center}

\subsection*{Two--dimensional Complex Jordan Algebras}
\hspace{16pt}
\begin{tabular}{ll}
$2\j_{1}: \quad$ & Abelian
\end{tabular}
\begin{center}
\begin{tabular}[t]{|l||c|c|}
\hline \parbox[l][20pt][c]{0pt}{}   $\alpha$ & 1 &  \\ \hline
\parbox[l][20pt][c]{0pt}{} $\fa(\alpha)$ & 4 & 4  \\ \hline
\end{tabular}\qquad
\end{center}

\begin{tabular}{ll}
$\j_{2,1}: \quad$ & $ e_1 \circ e_1 = e_1,\ e_1\circ e_2 =e_2$
\end{tabular}
\begin{center}
\begin{tabular}[t]{|l||c|c|c|}
\hline \parbox[l][20pt][c]{0pt}{}   $\alpha$ & 1 & 2 &  \\ \hline
\parbox[l][20pt][c]{0pt}{} $\fa(\alpha)$ & 1 & 2 & 0  \\ \hline
\end{tabular}\qquad
\end{center}

\begin{tabular}{ll}
$\j_{2,2}: \quad$ & $ e_2 \circ e_2 = e_2$
\end{tabular}
\begin{center}
\begin{tabular}[t]{|l||c|c|c|c|}
\hline \parbox[l][20pt][c]{0pt}{}   $\alpha$ & 1 & 0 & 2 &  \\
\hline
\parbox[l][20pt][c]{0pt}{} $\fa(\alpha)$ & 1 & 2 & 2 & 1  \\ \hline
\end{tabular}\qquad
\end{center}

\begin{tabular}{ll}
$\j_{2,3}: \quad$ & $ e_1 \circ e_1 = e_2$
\end{tabular}
\begin{center}
\begin{tabular}[t]{|l||c|c|c|}
\hline \parbox[l][20pt][c]{0pt}{}   $\alpha$ & 1 &   \\ \hline
\parbox[l][20pt][c]{0pt}{} $\fa(\alpha)$ & 2 & 2   \\ \hline
\end{tabular}\qquad
\end{center}

\begin{tabular}{ll}
$\j_{2,4}: \quad$ & $ e_1 \circ e_1 = e_1,\ e_1\circ e_2
=\frac{1}{2}e_2$
\end{tabular}
\begin{center}
\begin{tabular}[t]{|l||c|c|c|}
\hline \parbox[l][20pt][c]{0pt}{}   $\alpha$ & 1 & 2 &  \\ \hline
\parbox[l][20pt][c]{0pt}{} $\fa(\alpha)$ & 2 & 1 & 0  \\ \hline
\end{tabular}\qquad
\end{center}

\begin{tabular}{ll}
$\j_{2,5}: \quad$ & $ e_1 \circ e_1 = e_1,\ e_2\circ e_2 = -e_1, \
e_1\circ e_2 =e_2$
\end{tabular}
\begin{center}
\begin{tabular}[t]{|l||c|c|c|}
\hline \parbox[l][20pt][c]{0pt}{}   $\alpha$ & 1 & 2 &  \\ \hline
\parbox[l][20pt][c]{0pt}{} $\fa(\alpha)$ & 0 & 2 & 0  \\ \hline
\end{tabular}\qquad
\end{center}


\end{document}